\documentclass[sigconf,table,dvipsnames]{acmart}
\def\BibTeX{{\rm B\kern-.05em{\sc i\kern-.025em b}\kern-.08emT\kern-.1667em\lower.7ex\hbox{E}\kern-.125emX}}

\usepackage{amsmath,amsfonts}
 
\usepackage{booktabs}  % 导入 booktabs 宏包

\usepackage{pifont} %use the command \ding

\usepackage{bbding} %对勾叉号
\usepackage{pifont}
\newcommand{\cmark}{\textcolor{green!80!black}{\ding{51}}}
\newcommand{\xmark}{\textcolor{red}{\ding{55}}}

\usepackage{makecell}
\usepackage{graphicx} % for handling graphics
\usepackage{diagbox}  % for diagonal boxes in tables

\usepackage{tikz}
\usetikzlibrary{arrows.meta}
\usetikzlibrary{automata,positioning}

\usepackage{threeparttable}

\renewcommand*{\arraystretch}{1.5}%
\definecolor{tabred}{RGB}{230,36,0}%
\definecolor{tabgreen}{RGB}{0,116,21}%
\definecolor{taborange}{RGB}{250,124,30}%
\definecolor{tabbrown}{RGB}{171,70,0}%
\definecolor{tabyellow}{RGB}{251,253,169}%
\newcommand*{\vcorr}{%
  \vadjust{\vspace{-\dp\csname @arstrutbox\endcsname}}%
  \global\let\vcorr\relax
}% 

\usepackage{mathtools}

\usepackage{lscape}
\usepackage[figuresright]{rotating}
\usepackage{adjustbox}
\usepackage{caption}

\usepackage{subfigure}

\usepackage{xfrac}

\usepackage{appendix}
\usepackage{float}

\usepackage{fancybox}%use other types of boxes
\usepackage{framed}

\usepackage{multirow}

\usepackage{pgfplots}

\usepackage{tabularx,makecell, caption}
\usepackage{enumitem} % to control Itemization spacing

\usepackage{multicol}
\usepackage{blindtext}
\usepackage{etoolbox,xstring,mfirstuc,textcase}

\usepackage{subfiles}
\usepackage[most]{tcolorbox}

\usepackage{array}     % for custom column types

\newcommand{\maxval}{100}        % 最大值
\newcommand{\maxbarwidth}{1.5}   % 进度条最大宽度（单位：cm）

\newcommand{\cellbarACC}[1]{%
  \pgfmathparse{#1/\maxval*\maxbarwidth}%
  \edef\barlen{\pgfmathresult cm}%
  \begin{minipage}[c][1.5em][c]{\maxbarwidth cm}
    \colorbox{teal!25}{\makebox[\barlen][l]{}}%
    \hspace*{-\barlen}%
    \makebox[\maxbarwidth cm][c]{#1}%
  \end{minipage}%
}

\newcommand{\cellbarSE}[1]{%
  \pgfmathparse{#1/\maxval*\maxbarwidth}%
  \edef\barlen{\pgfmathresult cm}%
  \begin{minipage}[c][1.5em][c]{\maxbarwidth cm}
    \colorbox{magenta!15}{\makebox[\barlen][l]{}}%
    \hspace*{-\barlen}%
    \makebox[\maxbarwidth cm][c]{#1}%
  \end{minipage}%
}

\usepackage{amsthm}
\theoremstyle{plain}         
 
\newtheorem{thm}{Theorem}    
\newtheorem{lem}{Lemma}  
\newtheorem{defi}{Definition}
\theoremstyle{definition}

\usepackage{siunitx}
\usepackage{comment}
\usepackage{indentfirst} 
\usepackage{framed} 

\usepackage[font=small,labelfont=bf,tableposition=top]{caption}
\usepackage{dashbox}

\usepackage{hyperref}
\hypersetup{
    colorlinks=true,
    linkcolor=blue,
    filecolor=teal,      
    urlcolor=teal,
    citecolor=teal,
    pdftitle={An Example},
    pdfpagemode=FullScreen,
    }
\usepackage[switch]{lineno}

\usepackage{url}

\newenvironment{packeditemize}{
	\begin{list}{$\bullet$}{
			\setlength{\labelwidth}{4pt}
			\setlength{\itemsep}{0pt}
			\setlength{\leftmargin}{\labelwidth}
			\addtolength{\leftmargin}{\labelsep}
			\setlength{\parindent}{0pt}
			\setlength{\listparindent}{\parindent}
			\setlength{\parsep}{0pt}
			\setlength{\topsep}{1pt}}}{\end{list}}

\newenvironment{circitemize}{
	\begin{list}{$\circ$}{
			\setlength{\labelwidth}{4pt}
			\setlength{\itemsep}{0pt}
			\setlength{\leftmargin}{\labelwidth}
			\addtolength{\leftmargin}{\labelsep}
			\setlength{\parindent}{0pt}
			\setlength{\listparindent}{\parindent}
			\setlength{\parsep}{0pt}
			\setlength{\topsep}{1pt}}}{\end{list}}

\definecolor{light_cyan}{rgb}{0.53, 0.75, 0.77}
\definecolor{light_blue}{rgb}{0.466, 0.655, 0.94}
\definecolor{light_pink}{rgb}{0.98, 0.55, 0.565}
\definecolor{light_yellow}{rgb}{0.98, 0.83, 0.32}

\usepackage{color}
\usepackage{algorithm,algpseudocode} %% select one of them

\renewcommand\footnotetextcopyrightpermission[1]{} % removes footnote with conference information in first column
\pgfplotsset{compat=1.9}

\DeclareMathAlphabet{\mathcal}{OMS}{cmsy}{m}{n}

\DeclareGraphicsExtensions{%
    .png,.PNG,%
    .pdf,.PDF,%
    .jpg,.mps,.jpeg,.jbig2,.jb2,.JPG,.JPEG,.JBIG2,.JB2}

\begin{document}

\title{\textsc{PoLO}: Proof-of-Learning and Proof-of-Ownership at Once \\with Chained Watermarking}

%=================================================
%author
%=================================================

\author{Haiyu Deng$^{1}$, Yanna Jiang$^{1}$, Guangsheng Yu$^{1}$, Qin Wang$^{1,2}$,  Xu Wang$^{1}$, \\ Baihe Ma$^{1}$, Wei Ni$^{1,2}$, Ren Ping Liu$^{1}$}
\affiliation{
\textit{$^1$University of Technology Sydney} $|$ \textit{$^2$ CSIRO Data61, Australia} 
}

%=================================================
%abstract
%=================================================

\begin{abstract}
Machine learning models are increasingly shared and outsourced, raising requirements of verifying training effort (Proof-of-Learning, PoL) to ensure claimed performance and establishing ownership (Proof-of-Ownership, PoO) for transactions. When models are trained by untrusted parties, PoL and PoO must be enforced \textit{together} to enable protection, attribution, and compensation. However, existing studies typically address them separately, which not only weakens protection against forgery and privacy breaches but also leads to high verification overhead.

We propose PoLO, a unified framework that simultaneously achieves PoL and PoO using \textit{chained watermarks}. PoLO splits the training process into fine-grained training shards and embeds a dedicated watermark in each shard. Each watermark is generated using the hash of the preceding shard, certifying the training process of the preceding shard. The chained structure makes it computationally difficult to forge any individual part of the whole training process. The complete set of watermarks serves as the PoL, while the final watermark provides the PoO. PoLO offers more efficient and privacy-preserving verification compared to the vanilla PoL solutions that rely on gradient-based trajectory tracing and inadvertently expose training data during verification, while maintaining the same level of ownership assurance of watermark-based PoO schemes. Our evaluation shows that PoLO achieves \textbf{99\%} watermark detection accuracy for ownership verification, while preserving data privacy and cutting verification costs to just \textbf{1.5–10\%} of traditional methods. Forging PoLO demands \textbf{1.1–4×} more resources than honest proof generation, with the original proof retaining over \textbf{90\%} detection accuracy even after attacks.

\end{abstract}


%=================================================
\maketitle
%\linenumbers
%=================================================   

%=================================================   

\section{Introduction}\label{sec:intro}
As machine learning (ML) models evolve from academic artifacts into economically valuable digital assets~\cite{idowu2022asset}, their exchange and commercialization increasingly take place through enforceable model marketplaces (EMMs)~\cite{emm-2}, as well as through federated platforms and outsourced pipelines. Provenance and ownership are established before models can be listed, licensed, or billed. These settings therefore require verification of both legitimate training effort and rightful model ownership, since such guarantees underpin trust, attribution, and compensation in real-world ML trade~\footnote{Non-EMM redistribution, grey markets, internal theft, and pure value-destruction attacks are out of scope.}.
Some examples are: (i) \textit{incentive-driven distributed learning}~\cite{10323597}: %9796833,
participants prove that models are properly trained to receive fair rewards; (ii) \textit{ML marketplaces}~\cite{9445602}: buyers need confidence that models are genuinely trained and transferable without dispute; (iii) \textit{outsourced training}~\cite{yu2024splitunlearning}: %9384314,
organizations ensure that externally developed models are both authentic and securely attributed. 
All this introduces a fundamental challenge:

\begin{figure}[!t]
        \centering
        \includegraphics[width=0.8\columnwidth]{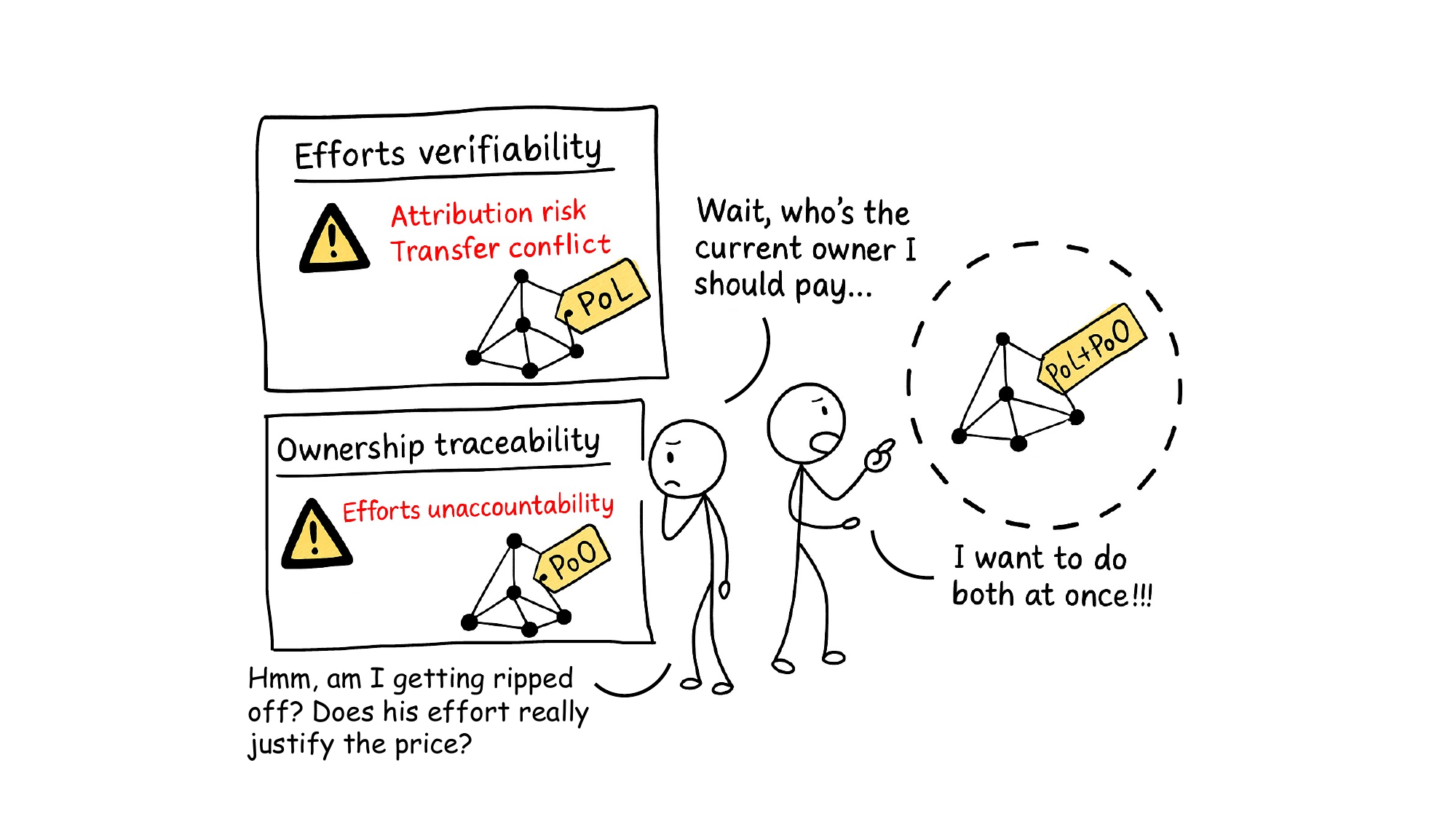}
        \caption{\textbf{Why at once?} PoL verifies training effort but lacks ownership tracking, while PoO ensures ownership but fails to justify training efforts. Separating PoL and PoO creates attribution risks and ownership conflicts.}  \label{why_polo}
\end{figure}

 \vspace{3pt}
\underline{\textit{How can a recipient verify model legitimacy in ML ecosystems?}}
 \vspace{3pt}
 
Legitimacy here has two facets~\cite{xie2025dataset}: Proof-of-Learning (PoL) verifies that the claimed computational effort was genuinely invested, while Proof-of-Ownership (PoO), typically implemented via watermarking, binds verifiable ownership information to the model in a tamper-resistant way.

We argue that \textit{PoL and PoO are intertwined and must be addressed jointly to establish a complete proof of legitimacy} (cf Fig.\ref{why_polo}). This is not just a design preference but a requirement emerging from regulatory regimes that already demand provenance and ownership checks for listing, billing, and dispute resolution in real distribution channels such as cloud providers,
%~\cite{reed2021information}
model hubs,
%~\cite{spoczynski2025atlas}
and marketplaces~\cite{9445602,EUAIActArticle53}. % sahoo2021traceability 
These EMM-style platforms operate under notice and takedown and transparency duties (for example the U.S.\ DMCA safe harbor~\cite{USDMCASafeHarbor}
% the EU Digital Services Act~\cite{EUDigitalServicesAct}, 
and EU AI Act Article~53~\cite{EUAIActArticle53}), which in practice require verifying both training effort and ownership together~\cite{EUCommissionGPAI2024}. However, existing solutions treat PoL and PoO as separate mechanisms. When naively stacked, the overall system degrades in security and effectiveness, failing primarily due to the weaknesses of current PoL approaches:

 \vspace{3pt}
\begin{packeditemize}
    \item \textit{Heightened vulnerability}: 
    As first framed by Jia et al. (S\&P’21~\cite{frist_pol}), existing PoLs require exposing intermediate training states and training data; follow-ups by Zhang et al. (S\&P’22~\cite{pol_attack1}) and Thudi et al. (USENIX’22~\cite{pol_forge}) showed these artifacts can be manipulated or fabricated to forge plausible trajectories.
    \item \textit{High verification costs}:
    In the Jia et al. formulation and its variants~\cite{frist_pol,pol_attack1,pol_attack2}, verifiers must replay or retrain from recorded snapshots, causing verification costs to scale with model size, dataset size, and update count. This approaches the full cost of training, making the approach impractical for large or high-capacity models.
\end{packeditemize}
Simply combining PoL and PoO on top of existing PoLs is ineffective, as the entire system just inherits PoL’s vulnerabilities and overhead without providing stronger legitimacy.
 \vspace{3pt}

% \smallskip
% \noindent\textbf{Deployment scope.}
% We consider EMMs, where verifiable ownership is a precondition for listing, licensing, or billing. Under this assumption, stripping a watermark without replacing it is economically irrational, since a watermark-free model cannot be legitimately monetised. Non-EMM redistribution, grey markets, internal theft, and pure value-destruction attacks are therefore out of scope for \textsc{PoLO}.
% In this paper, we focus on EMMs in which verifiable ownership is a prerequisite for listing, licensing, or billing. In these settings, simply stripping a watermark without re-embedding a new, verifiable one is economically irrational, because watermark-free models cannot be legitimately monetised. Our analysis does \textit{not} cover redistribution through non-EMM channels (e.g., public model hubs, private sharing, API-only serving, or model distillation/extraction), grey markets, internal theft, or pure availability attacks where the goal is to destroy value rather than to claim it; these are important but lie outside the economic and policy model that \textsc{PoLO} is designed for.
Our goal is to \textit{redesign PoL from the ground up to accommodate these enforceable environments for ML models}, embedding PoL with PoO to provide a single, seamless framework that validates training effort and ownership together under \underline{\textit{two core assumptions}}:
\begin{packeditemize}
    \item Any attack whose total computational cost is no less than that of the corresponding honest training is considered \textit{economically unviable} and therefore irrational.
    \item Within EMMs, models lacking a unique, verifiable watermark are deemed invalid and cannot be monetised. A rational attacker must therefore both erase the incumbent legal watermark and embed a new forged watermark to assert ownership; attacks that only remove watermarks~\cite{287178} %sok_attack
    merely destroy value and are \textit{economically irrelevant} for our security analysis.
\end{packeditemize}
 \vspace{3pt}
Under these assumptions, rational participants aim to earn rewards in EMMs. To profit, an attacker must invest in a removal-plus-re-embedding phase whose cost is comparable to, or exceeds, honest training. Any forgery whose cost is not strictly lower than legitimate training is thus treated as a failed attack. Our framework does \textit{not} cover grey markets, internal theft, or pure availability attacks where the goal is to destroy value rather than to claim it.
This presents the technical challenges of PoL+PoO and motivates the following research questions (RQs):

\begin{packeditemize}
\item \textbf{\textcolor{violet}{RQ1}: }\textit{How can PoO be extended to certify training effort for PoL?} Existing PoO schemes embed a watermark only in the final model, which cannot reflect the training trajectory and makes them incompatible with PoL under the premise that only complete and traceable effort can be rewarded.

\item \textbf{\textcolor{violet}{RQ2}: }
\textit{How can PoO-based PoL remain secure against rational adversaries who seek to forge ownership at minimal cost?}
In EMMs, the design must ensure that this removal-plus-re-embedding phase requires computational effort comparable to honest training, so that any economically rational forgery attempt is no cheaper than legitimate training and thus unattractive.
\end{packeditemize}

 \vspace{3pt}
\noindent\textbf{{Contributions.}} 
To answer these RQs, we propose a novel design, \textsc{PoLO}, which embeds \textit{chained watermarks} throughout the training process. Each watermark is deterministically derived using a hash function over partial model weights and auxiliary parameters from the previous training shard. This chaining cryptographically links training phases and embeds ownership information at every stage, rather than relying on gradient-based tracking that have been proven prone to forgery~\cite{pol_attack1,pol_attack2}. More importantly, \textsc{PoLO} traps attackers in an economic dilemma: they must either incur a cost comparable to honest training, making forgery \textit{economically impractical}, or suffer degraded model performance that fails verification. In either case, the attacker gains no benefit.

%\begin{quote}
%\end{quote}

This insight highlights the necessity of \textsc{PoLO}’s chained watermarking: It shifts the defence from protecting isolated watermarks to securing the full training trajectory. By cryptographically linking each shard, \textsc{PoLO} ensures that any ownership claim must reflect genuine cumulative effort. Attacks targeting a single watermark are ineffective, as the proof depends on consistent verifiability across all stages, aligning security with rational incentives in EMMs.

 \vspace{3pt}
We achieve this in a stepwise manner:

\begin{packeditemize}
%\noindent$\triangleright \quad$
\item \textit{We formalize the concept of Proof-of-Anything (PoX)} ($\S$\ref{sec:related work}–$\S$\ref{sec:preliminaries}) as a unified framework for analyzing PoL and PoO. Inspired by Proof-of-Work (PoW), PoX provides a structured basis for systematically evaluating PoL methods in design, efficiency, and security. We identify key limitations in existing PoL methods, including inefficient proof generation, high verification costs, fragmented security, and privacy risks. Our findings highlight the need for an integrated solution that verifies both training effort (i.e., PoL) and ownership (i.e., PoO) efficiently and securely, motivating the design of \textsc{PoLO}.

% \smallskip
% \noindent$\triangleright \quad$
% \underline {We introduce \textsc{PoLO}} ($\S$\ref{sec:design}), an efficient solution for verifying the training process while addressing \textit{non-negligible computational overhead} and \textit{privacy concerns}. In the \textsc{PoLO} framework, the model trainer embeds chained watermarks during training to verify the PoL. These watermarks are updated in each training round using a hash function that incorporates partial weights and timestamps. The embedded watermark reflects the training effort, allowing the verifier to confirm the PoL without recalculating intermediate weights or exposing sensitive data. As the watermark chain progresses, model ownership is encoded in the final watermark, \textit{achieving PoO as a natural outcome}. This approach ensures that both the training process (i.e., PoL) and ownership (i.e., PoO) are validated without the need for separate ownership verification steps — hence, PoLO. The process is efficient, requiring neither dataset access nor repeated weight recalculations, significantly reducing computational overhead and preserving privacy. The verifier simply extracts the final watermark and compares it with expected values to confirm both the training process and ownership, addressing the challenges of existing PoL schemes in a unified solution.

\item \textit{We introduce \textsc{PoLO}} ($\S$\ref{sec:design}), a method that unifies PoL and PoO through chained watermarks. Rather than stacking PoL and PoO as independent modules, \textsc{PoLO} uses watermark embedding itself as the PoL verification primitive, so the ownership proof \emph{is} the training proof. Three properties follow: (i)~each watermark is derived from a cryptographic hash of the previous checkpoint, replacing gradient replay with a PoO primitive that certifies the full training trajectory; (ii)~shard-level watermark checking reduces verification cost to 1.5\%--10\% of training; and (iii)~the chain structure forces forgery cost to scale with chain depth, closing the gap that arises when PoL and PoO are treated independently. Because \textsc{PoLO} replaces gradient trajectories with hash-based chaining, it eliminates dataset exposure and reduces verification overhead in a single unified framework.

\item \textit{We conduct extensive experiments} ($\S$\ref{sec:experiment}) using ResNet, VGG, and BERT models on CIFAR10, TinyImageNet, and AG News datasets. We implement watermarking schemes with varying sizes to generate \textsc{PoLO} proofs, compare verification overhead with traditional PoL methods, and test resilience against two novel proof forgery attacks. Results show that \textsc{PoLO} achieves \textbf{99\%} watermark detection accuracy for ownership verification while fully preserving training data privacy. Verification requires only \textbf{1.5\%–10\%} of computational overhead compared to traditional PoL methods. Forging PoL proofs against \textsc{PoLO} demands \textbf{1.1x}–\textbf{4x} more resources than legitimate \textsc{PoLO} generation, with the original proof maintaining over \textbf{90\%} detection accuracy even after attacks.

\end{packeditemize}
\section{Formalizing Concurrent Works}\label{sec:related work}

We provide the formalization of PoX as a unified framework for analyzing PoL and PoO, outlining their core concepts and concurrent works. 
Tab.\ref{tab:notation} summarizes the notation.

\begin{table}[!]
\centering
\caption{Comparison with existing PoL.}
\label{tab_comp}
\renewcommand{\arraystretch}{1} 
\resizebox{\linewidth}{!}{
\begin{threeparttable}
\begin{tabular}{c|c|ccc|cc|cc}

\midrule

% \diagbox[width=4cm, height=1.5cm]{\textbf{Method}}{\textbf{Characteristic}}
 \multicolumn{1}{c}{} 
 & \multicolumn{1}{c}{ \textbf{\makecell{Ownership\\(EMM)}}}
 & \multicolumn{3}{c}{\textbf{Security}} 
 & \multicolumn{2}{c}{\textbf{Privacy}} 
 & \multicolumn{2}{c}{\textbf{Low Ov.}}  \\
 
 \multicolumn{1}{c}{} &  \cellcolor{yellow!15}\ding{172} & \cellcolor{yellow!15}\ding{173} & \cellcolor{yellow!15}\ding{174} & \cellcolor{yellow!15}\ding{175} & \cellcolor{yellow!15}\ding{176} &  \cellcolor{yellow!15}\ding{177} & \cellcolor{yellow!15}\ding{178} & \cellcolor{yellow!15}\ding{179}  \\
 
\midrule

\makecell{PoL (GD)\\~\cite{frist_pol,pol_attack1,pol_attack2,pot-2026}}
& \xmark & \cmark & \xmark & \xmark & \xmark & \xmark  & \xmark & \xmark  \\
% &  \textcolor{purple}{High} &  \textcolor{purple}{High} & \xmark \\

  PoL (hash)~\cite{incentive_pol}   
& \xmark  & \cmark & \cmark  &  \cmark & \xmark & \xmark  & \xmark & \xmark \\
% & \textcolor{purple}{High} &  \textcolor{purple}{High} & \xmark \\

 PoL (zkp)~\cite{zkf_pol}
& \xmark & \cmark  & \cmark & \cmark & \cmark & \cmark & \multicolumn{1}{c}{\xmark} & \cmark  \\
% &  \textcolor{purple}{Very High} &  \textcolor{green!80!black}{Low} & \xmark \\

\midrule

\multicolumn{1}{c}{\textbf{PoLO} (ours)}
&\cellcolor{blue!10}   \cmark & \cellcolor{blue!10}  \cmark &\cellcolor{blue!10}   \cmark & \cellcolor{blue!10}  \cmark & \cellcolor{blue!10}  \cmark & \cellcolor{blue!10}  \cmark & \cellcolor{blue!10}  \cmark & \cellcolor{blue!10}  \cmark  \\
% & \textcolor{green!80!black}{Low} & \textcolor{green!80!black}{Low} & \cmark \\
\cmidrule{5-9}
%\bottomrule
\end{tabular}
\begin{tablenotes}
        \small
        \item  \textbf{Notation:} \cmark~ for attack-resistance/property-held; \xmark~ vice versa; \textbf{Ov.} for overhead. Prevent:
        \item  \ding{172} Ownership enforcement in EMMs? \ding{173} Reverse reconstruction attacks?
        \item \ding{174} Synthetic trajectory attacks? \ding{175} Verification loophole attacks? 
        \item \ding{176} Avoid data sharing? \ding{177} Defend against gradient leakage? 
        \item \ding{178} Proof generation?  \ding{179} Verification?
\end{tablenotes}
\end{threeparttable}
}
\end{table}

% \textbf{\makecell{Avoid \\Unauthorized Use?}}

%\textbf{\makecell{Prevent Reverse \\ Reconstruction Attacks?}}
%\textbf{\makecell{Prevent Synthetic \\ Trajectory Attacks?}} 
% \textbf{\makecell{Prevent Verification \\ Loophole Attacks?}}
%\textbf{\makecell{Avoid\\Data Sharing?}}
%\textbf{\makecell{Defend against\\Gradient Leakage?}} 
%\textbf{\makecell{Proof\\Generation}}
% \textbf{\makecell{Verification}} 

\subsection{Proof-of-Anything and Proof-of-Learning}

\begin{defi}[Proof-of-anything, PoX]
A \textit{prover} $\mathcal{P}$ sends a proof $\mathbb{P}$ to a verifier $\mathcal{V}$, where the resources required for verifying whether $\mathcal{P}$ satisfies a certain condition $\Psi$ is negligible compared to the computational overhead incurred during the generation of $\mathbb{P}$. The core of this verification process is a function $\Theta$, which is irreversible in the sense that its output cannot be practically reproduced without complete knowledge of its input $\chi$, nor can a proof $\mathbb{P}(\mathcal{\chi}’,\Psi,\Theta)=\mathbb{P}(\mathcal{\chi},\Psi,\Theta)$ be forged when $\mathcal{\chi}’ \neq \mathcal{\chi}$.
\end{defi}

\begin{defi}[Proof-of-Work, PoW]
A proof-of-work has the verification function $\Theta$ as a cryptographic hash that maps an input $\chi$ (e.g., the $t$-th block body and a nonce) to an output that is computationally infeasible to invert, ensuring irreversibility. To generate a valid proof, the prover $\mathcal{P}$ must solve a computationally expensive puzzle so that the output of $\Theta(\chi)$ satisfies a difficulty condition $\Psi$. In contrast, verification by the verifier $\mathcal{V}$ is efficient, requiring only a single evaluation of $\Theta$ to check whether the resulting output meets $\Psi$.
\end{defi}

PoW~\cite{nakamoto2008bitcoin} is one of the earliest and most widely deployed PoX instantiations, used in decentralized ledger technologies to demonstrate computational effort.

 \smallskip
 \noindent\textbf{Proof-of-(more).} Beyond computational effort (as in PoW), various physical and virtual resources can serve as proof bases, including holdings (proof-of-stake, PoS)~\cite{8835275}, time (proof-of-elapsed-time, PoET)~\cite{wang2022multi}, storage and bandwidth (proof-of-space)~\cite{moran2019simple}, and reputation (proof-of-authority, PoA)~\cite{wang2022exploring}. In each case, the verification function $\Theta$ is typically closely tied to its input $\chi$.

 \vspace{3pt}
\begin{defi}[PoL]
\label{def:pol}
A \textit{prover} $\mathcal{P}$ constructs a PoL proof $\mathbb{P}$ by recording the complete training trajectory of ML model $W_T$ with $\mathbb{P} := (W_t, \mathbf{B}_t, A_t)_{t=0}^T$, where $W_t$ denotes model weights at the $t$-th epoch, $\mathcal{D}$ denotes the full training dataset, $\mathbf{B}_t \subseteq \mathcal{D}$ stands for mini-batches drawn from it, and $A_t$ denotes auxiliary training parameters. The verification condition $\Psi$ requires $\|W_{t+1}\! - \Theta(W_t, \mathbf{B}_t)\|\!\! \leq\!\! \epsilon$ for all $t$, where $\Theta$ is the model training function. 
\end{defi}

 \vspace{3pt}
\noindent\textbf{Gradient-based PoL.}
Gradient (GD)-based PoL~\cite{frist_pol} instantiates Definition~\ref{def:pol} through iterative parameter updates via $\Theta(W_t, \mathbf{B}_t) = W_t - \eta \nabla\mathcal{L}(W_t, \mathbf{B}_t)$, where $\eta$ is the learning rate and $\mathcal{L}$ is the loss function. 
Verification extracts consecutive pairs $(W_t, W_{t+1})$ from the trajectory $\mathbb{P}$ --- both entries are present in $\mathbb{P}$ as the snapshots at indices $t$ and $t+1$ --- together with the corresponding batches $\mathbf{B}_t$. 
Studies have revealed vulnerabilities in this approach including synthetic trajectory attacks~\cite{pol_attack1} and gradient matching exploits~\cite{pol_attack2}, driving the need for enhanced verification mechanisms.
We note that the \emph{natural non-reproducibility} of SGD (i.e., two honest runs from the same starting point may diverge due to hardware non-determinism and random batching) is a distinct phenomenon from \emph{adversarial trajectory forgery}, in which an attacker deliberately constructs a plausible-looking trajectory without performing the training.
Natural non-reproducibility does not imply immunity to adversarial forgery; indeed, the attacks above exploit exactly this gap.
Recent schemes for DNN ownership verification~\cite{pot-2026} also follow this log–and–replay design of training trajectories, so they inherit these privacy, overhead, and forgery limitations.

 \vspace{3pt}
\noindent\textbf{Hash-based PoL.} 
To enhance computational efficiency and reduce communication complexity, Zhao et al.~\cite{incentive_pol} introduced an authentication and verification protocol that achieves a better balance between computational overhead and security. In this scheme, the prover $\mathcal{P}$ generates a proof $\mathbb{P}$ during the model training process by saving the intermediate weight parameters $W_t$ as the input $\chi$ to the function $\Theta$. Here, $\Theta$ is instantiated as a hash function $h(W_t)$ to ensure irreversibility.
During verification, the verifier $\mathcal{V}$ recalculates the hash value $h(W_t)$ from the intermediate weights $W_t$ provided by $\mathcal{P}$ and compares it to the original hash in the proof. To confirm the training progression, $\mathcal{V}$ retrains the model from $W_{t-1}$ and checks whether the resulting state matches $W_t$ by evaluating the hash value, which serves as the condition $\Psi$. 

This method preserves proof integrity and mitigates synthetic trajectory attacks by enforcing exact hash matching for intermediate states. However, it remains an enhanced version of GD-PoL and lacks the ability to verify current ownership. Like GD-PoL, it depends on sharing data samples $\mathbf{B}$ for retraining in order to regenerate intermediate weights, raising serious privacy concerns and incurring significant computational overhead. Moreover, it suffers from a key limitation: the recomputed weights are unlikely to exactly match those generated during the original training process due to inherent randomness in optimization and hardware-level non-determinism. As a result, even if retraining is performed correctly, the resulting model may differ slightly, leading to hash mismatches and making the verification process unreliable.

 \vspace{3pt}
\noindent\textbf{ZK-PoL.}  
Zero-knowledge cryptographic commitments have been integrated into PoL to enhance privacy. ZK-PoL methods~\cite{zkf_pol,shamsabadi2022confidential} % cryptoeprint:2025/053
verify the training process in zero knowledge, ensuring the model is derived from the training data and a random seed. This requires the prover to work proportionally to the number of iterations, proving training efforts and resource use, thus promoting fairness for parties with limited resources. By verifying the entire process, the verifier ensures adherence to the training procedure.
Concretely, the prover commits to the dataset $\mathcal{D}$ (e.g., via a Merkle tree) and to the model weights $\sigma_{W_t}$, $\sigma_{W_{t+1}}$, then generates a succinct proof $\pi$ that each gradient-descent iteration was correctly executed on a batch sampled from the committed dataset, without revealing the raw data or model weights to the verifier. The verifier thus learns nothing beyond the validity of the training computation~\cite{zkf_pol}.

However, ZK-PoL faces severe \emph{efficiency} challenges that currently preclude practical adoption.
Proof generation for a single gradient-descent iteration requires approximately 15 minutes on VGG-11 (10M parameters) with batch size 16~\cite{zkf_pol}, and prover memory consumption reaches hundreds of gigabytes.
For a typical training run of thousands of iterations, the cumulative proving overhead exceeds the training cost by orders of magnitude.
Moreover, existing constructions assume the total number of iterations is fixed in advance, precluding incremental or fine-tuning scenarios common in modern ML pipelines.
These prohibitive costs mean that, while ZK-PoL offers strong cryptographic guarantees in theory, it remains impractical for all but the smallest models and datasets, motivating the need for lightweight alternatives.

 \vspace{3pt}
Existing PoL methods share three limitations (cf.\ Tab.\ref{tab_comp}): high verification cost from recomputing intermediate states, privacy risks from exposing training data during replay, and susceptibility to forgery via synthetic trajectories or adversarial examples. Our method addresses all three by embedding hash-based chained watermarks during training, so verification reduces to watermark extraction without dataset access or weight recomputation. The same watermark chain also establishes model ownership, unifying PoL and PoO in a single framework as detailed below.

\subsection{Proof-of-Ownership}

PoL verifies that training was performed, but alone it cannot establish \emph{who} owns the resulting model: trajectories can be forged~\cite{pol_attack1,pol_attack2}, verification requires costly replay, and in outsourced settings the trainer and the rights-holder are distinct entities~\cite{9384314}. PoL and PoO therefore address orthogonal questions, and a complete proof of legitimacy requires both.
In contrast, \textit{model ownership} is a legal construct governed by intellectual property laws, granting developers exclusive rights to control the model's usage, distribution, and modification~\cite{sai2024aitrulyyoursleveraging}.
A complete ownership proof can provide a more efficient means to address such disputes~\cite{shao2025explanation}, complementing PoL's technical validation to ensure comprehensive protection for machine learning models against both attribution disputes and ownership enforcement in EMMs.

\begin{defi}[PoO]
For a model owner acting as the prover $\mathcal{P}$, a valid ownership proof is denoted $\mathbb{P}_o(W_t,\mathcal{D},A,\Lambda,\Psi)$.
This is constructed using a neural network with weights $W_t$ at epoch $t$, trained on dataset $\mathcal{D}$ under specified training settings $A$ (including hyperparameters, model architecture, optimizer, and loss functions). 
During training, the prover $\mathcal{P}$ embeds personalized ownership information $\Lambda$ into the model. 
To verify ownership, a condition $\Psi$ is evaluated on the model weights $W_T$ at the final epoch $T$.
\end{defi}

The embedded ownership information $\Lambda$ needs to exhibit a high degree of robustness, ensuring that $\Lambda$ remains extractable or detectable for verification even after undergoing adversarial modifications, such as model fine-tuning~\cite{fine-tune_attck,fine-tune_embed}, pruning~\cite{prune_attack}, and overwriting attacks~\cite{watermark_attack1}. This guarantees the integrity of ownership verification, making it resistant to common attacks aimed at removing or obfuscating embedded ownership information.
%%%%%%%%%%%%%%%%%%%%%%%%%

\smallskip
\noindent\textbf{Watermarking.} Watermarking is a type of PoO technique developed to protect multimedia content such as images~\cite{IP_image}, text~\cite{IP_text}, and videos~\cite{IP_video} by embedding unique identifiers into the content. 
Uchida et al.~\cite{Uchida} proposed embedding watermarks into model weights via a regularization term added to the training loss, forming the EDNN scheme. The watermark $\Lambda$ is embedded into weights $W$ during training, and later extracted as $\hat{\Lambda}$ for verification by checking if $\Delta(\Lambda, \hat{\Lambda})$ satisfies condition $\Psi$, where $\Delta(\cdot,\cdot)$ denotes a similarity metric (e.g., bit-accuracy based on Hamming distance). HufuNet~\cite{HufuNet} embeds a tailored autoencoder into the DNN, using the encoder as the watermark and preserving the decoder for ownership verification. RIGA~\cite{riga} introduces a GAN-based approach where the model acts as a generator producing watermarked weights, aided by a discriminator and an embedder network. FedIPR~\cite{li2022fedipr} extends these ideas to federated learning, allowing all clients to embed ownership into the global model for copyright protection~\cite{shao2024fedtracker}.

Our method is designed to be compatible with all embedded watermarking methods for achieving PoO. By embedding a given watermark into the model parameters, we associate the watermark with the model's training process, allowing for the verification of PoL through watermark verification.

\begin{figure*}[!t]
        \centering
        \includegraphics[width=\textwidth]{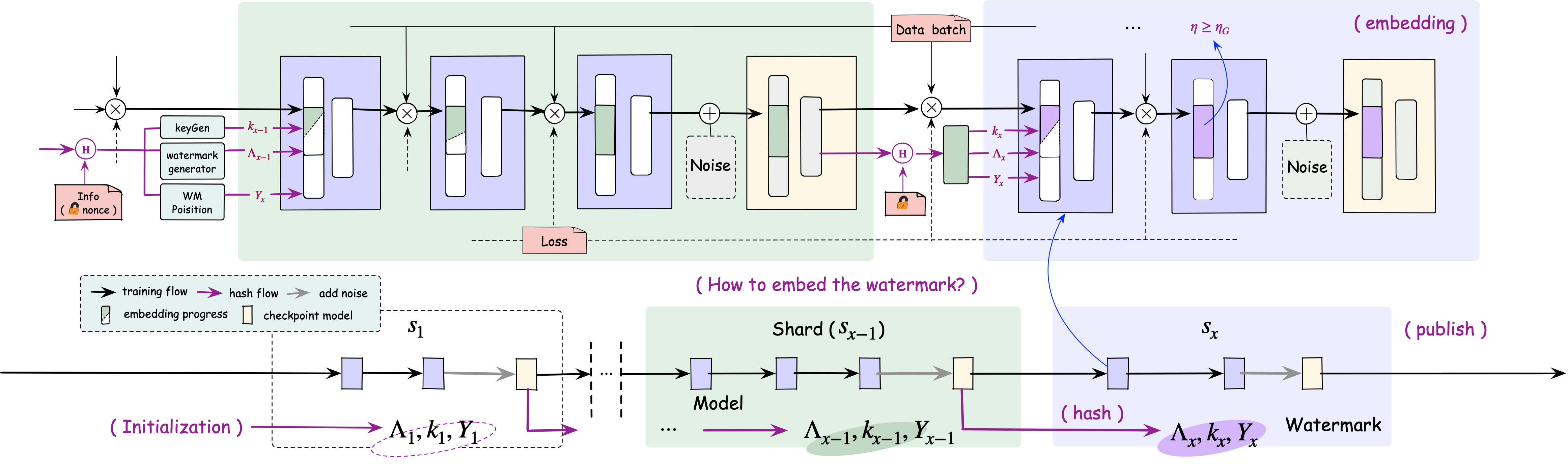}
        \caption{\textbf{\textsc{PoLO} design}: The verifier $\mathcal{V}$ shares a secret nonce with the prover  $\mathcal{P}$ to initialize watermark parameters ($\Lambda_{1}$, $k_1$, $Y_1$) for the first shard $s_1$. Prior to watermark embedding, the model owner computes the watermark $\Lambda_{x-1}$ and its corresponding embedding key $k_{x-1}$ and selection matrix $Y_x$ for shard $s_x$ using a hash function $\mathbb{H}(\cdot)$ over the previous model $W_{x-1}$, auxiliary information, and the secret nonce. During training, $\Lambda_{x-1}$ is embedded into the model using $k_{x-1}$, while monitoring the watermark detection rate $\eta$. Once $\eta$ exceeds the threshold $\eta_G$, Gaussian noise is applied to enhance robustness against inference attacks. Training proceeds to the next shard with new $\Lambda_x$, $k_x$ and $Y_x$. The process continues until the model converges.}
        \label{framwork}
\end{figure*}

\section{System Model}\label{sec:preliminaries}

\noindent\textbf{Problem definitions.}
Arguably, 
\textit{existing PoL methods fail to address scenarios where models are outsourced or exchanged for monetary assets.}  
They assume that the entity providing the PoL proof is 
the model owner, making them unsuitable for outsourced training where ownership and training efforts may belong to different entities. %~\cite{9384314}. 
In ML marketplaces, current PoL methods cannot support ownership transfer, which prevents PoL from verifying rightful ownership when models are bought, sold, or reassigned. %~\cite{9445602}.

 \vspace{3pt}
\textit{Existing PoL methods become the weakest link when simply stacking PoL and PoO.}
Those solutions~\cite{pol_attack1,zkf_pol} are increasingly inefficient. The effort invested outweighs practical outcomes threefold. 

 \vspace{3pt}
\begin{packeditemize}
\item The training function  $\Theta$  in PoL is inherently less reliable than the cryptographic hash functions used as  $\Theta$  in PoW. Adversarial patterns may compromise the irreversibility of the training process~\cite{pol_attack1}, allowing spoofing attacks to forge training paths. 

\item The condition  $\Psi$ in PoL is defined as the model distance between intermediate weights, weaker than the difficulty metric in PoW. 

\item Maintaining a dynamic distance threshold to counter spoofing attacks adds further complexity. While techniques like ZKPs have been explored to verify training in zero knowledge~\cite{zkf_pol}, their real-world application remains constrained by prohibitive computational costs (e.g., ${\sim}15$mins of proving time per iteration for a 10M-parameter model), impractical for large-scale deployments.

\end{packeditemize}
 \vspace{3pt}
 
% \begin{packeditemize}

%     \item The trajectory of existing PoL methods has become increasingly inefficient, where the effort invested outweighs practical outcomes. The training function  $\Theta$  in PoL is inherently less reliable than the cryptographic hash functions used as  $\Theta$  in PoW. As noted in \cite{pol_attack1}, adversarial patterns can compromise the irreversibility of the training process, enabling spoofing attacks to forge training paths.
%     Moreover, the condition  $\Psi$  in PoL--defined as the model distance between intermediate weights—is weaker than PoW’s difficulty metric. According to \cite{pol_attack1}, maintaining a dynamic distance threshold to counter spoofing attacks adds further complexity. While techniques like ZKPs have been explored for securely sharing inputs~\cite{zkf_pol}, their real-world application remains constrained by high time and resource costs, necessitating dataset sharing with verifiers in most existing PoL designs.

%     \item The security of current PoO methods heavily depends on a single watermark that is trained to maturity and used solely for protecting the final released version. While the watermark is linked to the training process, it remains ineffective during intermediate stages, as it has not yet fully developed. Protection is only established once the watermark reaches maturity, protecting only the later versions rather than the whole training process.
    
% \end{packeditemize}

Our method instantiates the irreversible function $\Theta$ using a cryptographically secure hash function to chain the training path, transforming it from a simple GD trajectory into a rigorously linked sequence of watermarks--leveraging PoO to achieve PoL. Each step is securely linked to the previous one through model watermarks, establishing a structured verification path. This redefines the condition $\Psi$ (i.e., the criterion for validating a PoL proof) by incorporating the robustness of PoO, making it a more practical and resilient metric against attacks. Moreover, this design preserves data privacy by eliminating dataset sharing with verifiers, significantly reducing communication overhead.  
\textsc{PoLO} also integrates the efficiency of existing watermarking techniques~\cite{Uchida} with the undetectability properties of more advanced schemes~\cite{riga}.

\begin{defi}(\textsc{PoLO}):
A \textsc{PoLO} proof for a prover $\mathcal{P}$ is defined as $\mathbb{P}((W_{x-1}, W_x), \Psi, \Theta)$, which verifies the training of the $x$-th shard $s_x$ within a model iteration partitioned into $S$ shards.
The proof is considered valid if the ownership information $\Lambda_x$, extracted from the obfuscation-protected \underline{final weight} $W_x$ of $s_x$, satisfies the condition $\Psi$. Specifically, this requires verifying that $\Lambda_x$ matches the expected value $\hat{\Lambda}_x$, where $\hat{\Lambda}_x$ is computed using the hash-based function $\Theta$ over the prior shard’s weight $W_{x-1}$.
\end{defi}

\noindent\textbf{Architecture.}
\textsc{PoLO} features a chain-based embedded watermark structure to generate PoL with no need to share any training data and with superior resilience against gradient spoofing attacks~\cite{pol_attack2}.
Specifically, multiple watermarks are embedded throughout the model training iterations, forming a chain-based structure. Realizing the unique watermark embedded at a specific point in the chain requires a hash operation based on the knowledge of the model weights from the previous point. This ensures that an attacker intending to forge the watermark at any point would also need to forge all preceding watermarks in the chain, which incurs a computational cost equivalent to honestly training a new model.

The amount of effort in learning is reflected by the number of \textit{shards} formed throughout the iteration. The process of embedding each watermark naturally forms a shard, representing one point in the watermark chain. This indicates a sequential order between each shard. The size of a shard can vary to meet the accuracy threshold required to form a valid watermark. The weights of the final model within shard $s_x$ can be regarded as a \textit{checkpoint model}, reflecting the model performance.

 \vspace{3pt}
\noindent\textbf{Entities.} 
\textsc{PoLO} consists of two primary entities: 

\begin{packeditemize}
    \item \textbf{\textit{Prover ($\mathcal{P}$)}} is the original owner and trainer of a model, responsible for training the model, embedding watermarks, and generating proofs.
    The prover provides verifiable evidence of both training effort and model ownership, ensuring the integrity and authenticity of the proof.

    \item \textbf{\textit{Verifier ($\mathcal{V}$)}} is responsible for evaluating the validity of proofs by verifying the embedded watermarks and assessing the model’s main task performance. 
    As both the issuer and evaluator of the main task, the verifier ensures that the prover has completed the required workload (i.e., PoL) and retains legitimate ownership of the model (i.e., PoO). 
    To enforce fair verification, the verifier maintains a public test dataset for performance assessment and sets a watermark detection rate threshold. 
    Based on the verification results, $\mathcal{V}$ determines \emph{rewards} for successful validation and \emph{penalties} for fraudulent proofs. Rewards are proportional to the number of validated shards and the final model accuracy, creating an economic incentive for the prover to invest genuine training effort (cf.\ $\S$\ref{subsubsec: PoLO_cumulation}).

\end{packeditemize}

 \vspace{3pt}
% \noindent\textbf{Workflow in sketch.} We describe the procedures of model training and model challenging by presenting the interactive steps between a prover and a verifier.

% \begin{packeditemize}

% \item  \textbf{\textit{Step-1}} ($\S$\ref{subsec: training_watermark}). A prover $\mathcal{P}$ trains his model during which a number of watermarks are embedded to form a watermark-chain. This indicates that the published version of the model is bound to be watermarked by the final point of the chain.

% \item  \textbf{\textit{Step-2}} ($\S$\ref{subsubsec: PoLO_generating_proof}). A verifier $\mathcal{V}$ issues a challenge to $\mathcal{P}$, requesting verification for a randomly selected shard $s_x$.

% \item   \textbf{\textit{Step-3}} ($\S$\ref{subsubsec: PoLO_generating_proof}). $\mathcal{V}$ requires to send the model weights of the checkpoint models of shards $s_x$ and $s_{x-1}$.

% \item \textbf{\textit{Step-4}} ($\S$\ref{subsubsec: PoLO_verifying_proof}). $\mathcal{V}$ computes the expected embedded watermark for $s_x$ using the information shared from the owner, and then verifies that whether the checkpoint model of $s_x$ is watermarked by the computed watermark.

% \item \textbf{\textit{Step-5}} ($\S$\ref{subsubsec: PoLO_cumulation}). The verification result is sent back to $\mathcal{P}$. $\mathcal{V}$ decides whether to repeat the same process for further preceding shards, depending on the requirement of the security level.

% \end{packeditemize}

\noindent\textbf{Threat model.} Three types of players are considered. More discussions refer to \textbf{Appendix~\ref{appendix_security_privacy}}.

 \vspace{3pt}
\begin{packeditemize}
\item \textbf{\textit{Economically rational adversaries.}} 
We consider adversaries operating in EMMs, where \textit{unique}, verifiable ownership is required for listing and payment. 
Structural attacks such as Neural Structural Obfuscation (NSO)~\cite{287178} can disrupt white-box verifiers but are reversible via canonicalization, and pure watermark removal without re-embedding yields a non-monetisable model that cannot be listed or rewarded. 
A rational adversary who wishes to impersonate ownership must therefore \emph{jointly} (i) erase the incumbent watermark, (ii) embed a forged replacement, and (iii) preserve task accuracy to pass \textsc{PoLO}'s verification. Our analysis focuses on such removal-plus-re-embedding forgery attacks; any attack whose total cost meets or exceeds honest training \textit{per shard} is treated as economically irrational.
The attackers seek to exploit the system through various attacks, including:
    \begin{circitemize}
        \item \textit{Training forgery (addressed in \textrm{$\S$\ref{subsubsec: threat_model_1})}.} The attacker may fabricate a fraudulent PoL to falsely represent a legitimate training process, such as by forging training trajectories or interpolating intermediate states, aiming to falsely claim training efforts.
        
        \item \textit{Ownership theft (addressed in $\S$\ref{subsubsec: threat_model_2}).} The attacker may attempt to steal or misuse the model by targeting its watermark through attacks such as fine-tuning, pruning, or overlapping~\cite{overlap_attack}, aiming to falsely claim model ownership.  %sok_attack,,watermark_attack1
          Watermark removal attacks~\cite{287178} alone yield no benefit in trusted marketplaces, where only models with verifiable ownership are accepted. Such removal is only meaningful when combined with the embedding of a forged watermark.

            \item \textit{Inference attacks (addressed in $\S$\ref{subsubsec: threat_model_3}).} The attacker may exploit the information contained in the published model to infer the provers’ private training data. 
            
    \end{circitemize}

\item \textbf{\textit{Rational prover.}}  
The prover functions as both the model \textit{owner}, holding full rights to its creation and usage, and the PoL \textit{prover}, responsible for proving the authenticity and integrity of the training process.
A rational prover acts in a way that maximizes their utility. While they do not compromise the integrity of model training or watermark embedding, they may attempt to falsify reported workloads to gain additional rewards.

\item \textbf{\textit{Honest-but-curious verifier.}}
Verifiers act as auditors in EMMs. They challenge and validate the prover’s PoL to decide whether the claimed training effort merits rewards. They use a public test set to evaluate performance on the main task and configure system parameters such as shard-wise accuracy thresholds and watermark size. Verifiers are honest-but-curious: they follow the protocol, but may attempt to infer the prover’s private training data from the information revealed in the PoL.

\end{packeditemize}

% We consider that the honest participants always conduct honest behaviors, where they obey all the policies during the model training, embedding watermarks, generating proofs when challenged by verifiers, and other operations related to the defense strategies against adversaries.
% We consider that verifiers are \textit{honest-and-trustworthy} that serve as auditors to challenge and verify whether a model owner gets to be awarded in terms of his efforts on training models. System settings including the accuracy threshold of the watermark for each shard and the size of embedded watermark are managed by verifiers. Verifiers have a public test dataset to ensure that the model performance on the main training task.
% Model owners are \textit{rational} by overstating the efforts to reap more profits he deserves by any means.
% We consider that the published version of a model must be watermarked by the final shard of the watermark-chain.  
% We also consider that adversaries intend to steal the training efforts and have a ability to leverage prevalent attacks against watermarks including fine-tuning attacks~\cite{}, pruning attacks~\cite{}, removal attacks~\cite{}, overlapping attacks~\cite{}, and forging attacks against the gradient~\cite{}.

\section{Design of \textsc{PoLO}}\label{sec:design}

\textsc{PoLO} is a unified framework that integrates PoL and PoO through a chained watermarking mechanism.
This section outlines its design and verification process (cf. Figs.\ref{framwork}--\ref{framwork_verify}), and explains how \textsc{PoLO} enables fair workload validation, strong ownership protection, and robustness against adversarial attacks, while preserving the integrity of the model’s primary task performance.

\subsection{Training with Chained Watermarking}
\label{subsec: training_watermark}

% Existing PoL schemes are primarily achieved by recomputing the gradient descent of model weights.
% Verifying intermediate weights requires input datasets to the neural networks and training the last intermediate weights to obtain them. 
% This poses several drawbacks, such as the time-consuming of multiple rounds of recomputing; the PoL verification process necessitates disclosing the dataset, which can compromise the privacy of the dataset.
% To address these issues, we propose a scheme based on orthogonal watermarking to generate PoL.

% \subsection{Design Overview}
% During the PoL generation, the chained watermarks, which are used to prove the model copyright and integrity of the training, are dynamically embedded into the intermediate model weights by the prover $P$.
% During the watermark generation phase, selected current weights along with a timestamp are input into the hash function $h$. 
% The output of this hash function $h$ is used as a watermark, which is then embedded into the model for the next round of training. 
% During the model training process, the prover $P$ collects the intermediate model, timestamp, and generated watermark as proof-of-learning (PoL) and discloses them to the verifier $V$.

The core of \textsc{PoLO} is the \textit{chained watermarking} mechanism, which ensures progressive proof accumulation while maintaining model integrity, directly addressing \textbf{\textcolor{violet}{RQ1}} by binding ownership verification to the training process. 
Our framework links each training phase to its previous state, preventing forgery or tampering. 

%In \textsc{PoLO}, a shard represents a segment of the training process, distinct from traditional epochs.
%While an epoch refers to a complete pass over the training dataset, a shard consists of multiple epochs but is defined as an independent, verifiable training unit.
In \textsc{PoLO}, a shard is a verifiable unit of training that spans one or more consecutive epochs. Unlike fixed epoch boundaries, a shard is defined by the successful embedding of a watermark, which may take multiple epochs. This aligns verification with actual training effort. Let $T_x$ denote the set of epochs in the $x$-th shard $s_x$. By partitioning training into shards, \textsc{PoLO} enables verifiable proof of effort. The prover’s training process with chained watermarking involves the following steps:

\smallskip
\noindent\textbf{Chained watermark generation.} 
To ensure cryptographic linkage between successive training phases, \textsc{PoLO} derives the watermark $\Lambda_{x}$ for shard $s_x$ from the final weights $W_{x-1}$ of the previous shard $s_{x-1}$.
This chained watermarking mechanism guarantees that each watermark is uniquely bound to its predecessor, preventing adversaries from fabricating or modifying individual watermarks without reconstructing the entire sequence.

Initialization begins with a verifier-provided nonce, which the prover uses to generate the initial watermark parameters: the first watermark $\Lambda_1$, the embedding key $k_1$ for shard $s_1$, and a selection matrix $Y_1$ that specifies the embedding positions within model weights. 
% The matrix $Y$ is fixed and reused across all shards to ensure consistency and verifiability. 
For each shard $s_x$, both $Y_x$ and the embedding key $k_x$ are deterministically derived from the verifier-provided secret nonce $\mu$:
% \begin{equation}
% \begin{aligned}
% Y &= \textsf{WMPosition}(\mu), \\
% \Lambda_x,\; k_x &= \textsf{WMGen}(\mathcal{H}_{x-1}),\; \textsf{KeyGen}(\mathcal{H}_{x-1}) \text{ with} \\
% \mathcal{H}_{x-1} &= \mathbb{H}(W_{x-1},\; x,\; \mu,\; id_\mathcal{P}),
% \end{aligned}
% \label{wm_k_gen}
% \end{equation}
% where $\textsf{WMPosition}(\cdot)$, $\textsf{WMGen}(\cdot)$, and $\textsf{KeyGen}(\cdot)$ are deterministic functions, with all randomness derived from the verifier-fixed nonce $\mu$ and prover identity $id_\mathcal{P}$. This ensures only the legitimate prover can generate valid watermarks and keys, while the shard index $x$ enforces sequential linkage. The fixed nonce prevents adversaries from precomputing or guessing valid parameters.
\begin{equation}
\begin{aligned}
\Lambda_x &= \textsf{WMGen}(\mathcal{H}_{x-1}),\\
k_x &= \textsf{KeyGen}(\mathcal{H}_{x-1}),\\
Y_x &= \textsf{WMPosition}(\mathcal{H}_{x-1}), \text{ with} \\
\mathcal{H}_{x-1} &= \mathbb{H}(W_{x-1},\; x,\; \mu,\; id_\mathcal{P}),
\end{aligned}
\label{wm_k_gen}
\end{equation}
where $\textsf{WMPosition}(\cdot)$, $\textsf{WMGen}(\cdot)$, and $\textsf{KeyGen}(\cdot)$ are deterministic functions driven by $\mathcal{H}_{x-1}$. The shard index $x$ enforces sequential linkage. As a result, modifying one shard invalidates subsequent watermark tuples $(\Lambda_x,k_x,Y_x)$, so selective tampering requires reconstructing the chain rather than editing an isolated checkpoint.

Since modifying any individual shard would invalidate all subsequent watermarks, attackers cannot selectively alter a watermark without regenerating the entire chain, nor can they forge valid watermarks without reconstructing the full training sequence. 
This chained construction establishes a strong cryptographic binding, guaranteeing the authenticity of \textsc{PoLO} and ensuring that the training process remains verifiable and resistant to manipulation.

\smallskip
\noindent\textbf{Watermark embedding in model training.} 
During the training of $s_x$, the prover $\mathcal{P}$ selects specific layers in the model to embed a unique watermark $\Lambda_x$. 
The watermark $\Lambda_x$ is embedded into all model weights $W_{t_x}, \forall t_x \in T_x$, gradually forming the weights $W_{\bar{t}_x}$ of the final epoch in $s_x$ where $\Lambda_x$ has been successfully embedded. This process uses a secret key $k_x$ and selection matrix $Y_x$ to perform the embedding:
% The watermark $\Lambda_x$ is embedded into the model weights over multiple epochs within shard $s_x$, gradually forming the weights $W_{t_x}$ where $\Lambda_x$ has been successfully embedded. This embedding process uses a secret key $k_x$ and a selection matrix $X$, which determines the specific weights in $W_{t_x}$ to be modified for embedding $\Lambda_x$.
% \begin{equation}
%     W_{t_x} = \mathbb{F} (W^*_{\bar{t_{x-1}}},A,l_w + \lambda l_\Lambda) + \mathbb{E} (W^*_{\bar{t_{x-1}}}, \Lambda_x,k_x,X),
% \end{equation}
\begin{equation}
W_{\bar{t}_x} = \arg\min_{W} \left( l_w(W) + \lambda l_\Lambda(W) \right) 
\left. \vphantom{\sum} \right\rvert_{\substack{W_0 = W_{x-1}}} + \delta W,
\end{equation}
where $l_w(W)$ is the standard task loss (e.g., cross-entropy), $l_\Lambda(W)$ is a watermark embedding regularizer that drives the model weights toward satisfying $\Lambda_x$ at positions $Y_x$, and $\lambda > 0$ is a balancing hyperparameter. The minimisation is initialised at $W_{x-1}$ (the final checkpoint of the previous shard). The additional term $\delta W = \mathbb{E}(W_{x-1}, \Lambda_x, k_x, Y_x)$ is a post-training explicit perturbation applied to the selected weight positions $Y_x$ to finalize the watermark embedding; it is computed once per shard after training converges. This two-stage design, regularizer-guided training followed by a targeted perturbation, ensures that the watermark reaches the required detection rate $\eta_G$ while preserving main-task accuracy. The perturbation $\delta W$ is applied once at the end of each shard rather than at every epoch. This final adjustment yields the watermarked weights $W_{\bar{t}_x} = W_{\text{trained}} + \delta W$.
Both components are necessary: removing $l_\Lambda$ prevents the watermark from converging during training (detection rate stalls near 50\%), while omitting $\delta W$ slows shard completion by 2--4$\times$ in epochs because the regularizer alone cannot reliably push $\eta$ past $\eta_G^{\text{emb}}$ in a single pass.

% $\mathbb{F}(\cdot)$ donates model training, $\mathbb{E}(\cdot)$ is the watermark embedding function; $A$ is auxiliary information, such as hyperparameters and model architecture; $l_w$ is the loss function for main task and $l_\Lambda$ is watermark embedding regularizers with $\lambda$ is watermark hyperparameters; and $X$ is the selection matrix that determines the positions for embedding the watermark $\Lambda_x$.

This embedding process is dynamically monitored, with training continuing until the watermark detection rate $\eta$, which is computed using the following  based on Hamming distance:%~\cite{pHash}: 
\begin{equation}
    \eta = 1-\frac{ \sum_{i}^{n} (\mathbb{C}(W_{t_x},Y_x,k_x
    )[i]\ne \Lambda_x[i])}{n},
\end{equation}
reaches a predefined embedding threshold $\eta_G^{\text{emb}}$ (i.e., $\eta\geq \eta_G^{\text{emb}}$). Therein, $\mathbb{C}(\cdot)$ is the watermark extraction function and  $n$ is the size of $\Lambda_x$.
By enforcing this controlled, sufficient embedding process, \textsc{PoLO} ensures ownership traceability while preserving the main task’s performance.
    
\smallskip
\noindent\textbf{Obfuscation-based privacy protection and shard formation.}  
\textsc{PoLO} applies Gaussian noise by randomly selecting a subset of weights from the non-watermarked region $ W_{t_x} \setminus \{W_{t_x}Y_x\} $ and injecting obfuscation noise. This selective Gaussian mechanism introduces statistical uncertainty, protecting sensitive training information in non-watermarked areas from gradient leakage,%~\cite{9666855}, 
while preserving the ownership integrity carried by the embedded watermark.
% such as Laplace noise injection~\cite{10568968}
The point of successful watermark embedding, denoted as $ W_{\bar{t}_x} $, yields a obfuscation-protected version $ W^{\text{Gaussian}}_{\bar{t}_x} $, which marks the completion of shard $ s_x $. Therefore, we refer to this checkpoint model as $ W_x $ that serves as a secure and verifiable training checkpoint, ensuring tamper resistance and preserving privacy throughout each phase.
\begin{equation}
    W_x=W^{\text{Gaussian}}_{\bar{t}_x} = W_{\bar{t}_x} + \mathbb{G}(\sigma , Z, W_{\bar{t}_x} ),
\end{equation}
where $\bar{t}_x = \max_{\substack{t_x \in T_x}} (t_x) $ and $ W_{\bar{t}_x}$ represents the model weights at the last epoch of $s_x$. The term $\mathbb{{G}}(\sigma, Z, W_{\bar{t}_x})$ denotes the addition of $\sigma$-Gaussian noise to the subset of weights in $W_{\bar{t}_x}$ specified by a selection matrix $Z$ for Gaussian.
The storage of $W_x$ marks the completion of shard $s_x$, allowing the training process to transition to the next phase.

% Gaussian version 1 (simpler)
% Gaussian in PoLO is aimed at mitigating inference attacks in practice rather than being the central theoretical contribution. Still, our parameter-level noise yields formal guarantees under standard assumptions: (1) add Gaussian noise to all non-watermarked weights; (2) bound their sensitivity to one-sample changes via either an L2 anchor penalty or clipped gradients; (3) derive watermarked weights only from the previous Gaussian checkpoint and a verifier nonce; and (4) calibrate noise with the Gaussian mechanism and track shard-wise composition.

% Gaussian version 2 (more rigorous)
% Gaussian in PoLO is designed to mitigate inference attacks in practice rather than serve as the paper’s central theoretical contribution. That said, our parameter-level noise admits formal Gaussian guarantees under standard assumptions: (1) Gaussian noise is added to all non-watermarked coordinates; (2) either a proximal quadratic on $\overline{Y}_w$ (ensuring strong convexity) or clipped gradients with a contraction argument is used to bound the global $\ell_2$ sensitivity $\Delta^{(2)}x$; (3) watermarked coordinates are treated as post-processing, since the watermarking step overwrites the Y-block only with $(W^{\text{Gaussian}}{x-1},\textsf{nonce})$, keeping it independent of current raw data; and (4) the noise scale $\sigma_x$ is calibrated via the Gaussian mechanism, accounting for shard-wise composition.

\smallskip
\noindent\textbf{Iterative training with watermark propagation.} 
The new watermark $\Lambda_{x+1}$, derived from $\mathbb{H}(\cdot)$, is then embedded into the training of shard $s_{x+1}$, continuing the chained watermarking process. This cycle repeats until the model either converges to a stable state or meets the performance objectives of the main task. 
With each iteration in \textsc{PoLO}, the previous watermark propagates forward, ensuring strong sequential dependency across shards. 
This incremental accumulation of watermarks makes modification attacks impossible, as any tampering would require reconstructing the entire sequence while preserving model performance. 
By enforcing hash-based chaining and progressive watermark propagation, \textsc{PoLO} provides a cryptographically verifiable proof of both training effort and ownership, making it resistant to tampering or synthetic trajectory attacks.

% 需要转移PoO的情况
\smallskip
\noindent\textbf{Ownership transfer.}
In scenarios requiring ownership transfer, such as ML marketplaces~\cite{9445602} and outsourced training~\cite{9384314}, an additional fine-tuning shard is appended to the training process to reassign ownership.  
Specifically, a new or modified owner identifier $id_\mathcal{P'}$ is used to generate a fresh watermark $\Lambda_{S+1}$ based on the latest model $W_S$. This watermark is then embedded through fine-tuning, producing a new model instance $W_{S+1}$ that remains tied to the rightful owner while preserving the integrity of existing \textsc{PoLO} proofs. The resulting \textsc{PoLO} not only verifies the legitimacy of the final model owner but also retains the entire historical chain of PoO and PoL, ensuring full traceability.
%of both training efforts and ownership transitions

\subsection{Verification in \textsc{PoLO}}
\label{subsec: verification_watermark}

% During the PoL verification, the verifier $V$ first checks the watermark of the final weights to verify the model's copyright. 
% Then, they randomly select the watermarks of intermediate weights for verification. Additionally, to validate the effectiveness of the model, the performance of the intermediate weights on the test set is also verified. The frequency of verifying intermediate state weights represents a trade-off between time consumption and reliability.

The verification mechanism in \textsc{PoLO} is designed to ensure training effort accountability and ownership authenticity by validating both the watermark chain and the model performance. 
The verifier $\mathcal{V}$ assesses the prover's \textsc{PoLO} proof $\mathbb{P}$ by sequentially verifying the stored obfuscation-protected model checkpoints and their extracted watermarks from the latest shard to the earliest.
In \textsc{PoLO}, verification is performed at the \textit{shard level}, rather than per epoch, ensuring that provers cannot overclaim their training workload while maintaining a provable, tamper-resistant sequence of training. 
The verifier can selectively challenge any shard to verify both training efforts and ownership.

\subsubsection{\underline{Generating proof}}
\label{subsubsec: PoLO_generating_proof}

When the verifier $\mathcal{V}$ requests validation for shard $s_x$, the prover $\mathcal{P}$ submits the obfuscation-protected weights $W_x$ and $W_{x-1}$ from shards $s_x$ and $s_{x-1}$, respectively, along with the hash function $\mathbb{H}(\cdot)$ used to derive the watermark $\Lambda_x$, key $k_x$ and selection matrix $Y_x$ from $W_{x-1}$ using the verifier-recorded nonce $\mu$. 
 Moreover, $\mathcal{P}$ also provides the function that deterministically derives the selection matrix $Y$ from $\mu$. 
These elements enable $\mathcal{V}$ to verify both model integrity and the sequential linkage of watermarks (cf. Algorithm~\ref{algo:polo_generation} in \textbf{Appendix}~\ref{appendix_algo}).

Unlike gradient-based PoL methods~\cite{frist_pol, pol_attack1, pol_attack2, incentive_pol} that require the prover $\mathcal{P}$ to disclose training data for verification, \textsc{PoLO} ensures that verification can be conducted without exposing any training data to the verifier $\mathcal{V}$, and does so without the prohibitive computational overhead of zero-knowledge proof systems~\cite{zkf_pol}.
Our design significantly enhances data privacy, eliminating the risk of unintended data leakage while allowing the verifier to assess both training effort and ownership authenticity.

\begin{figure}[!t]
        \centering
        \includegraphics[width=\linewidth]{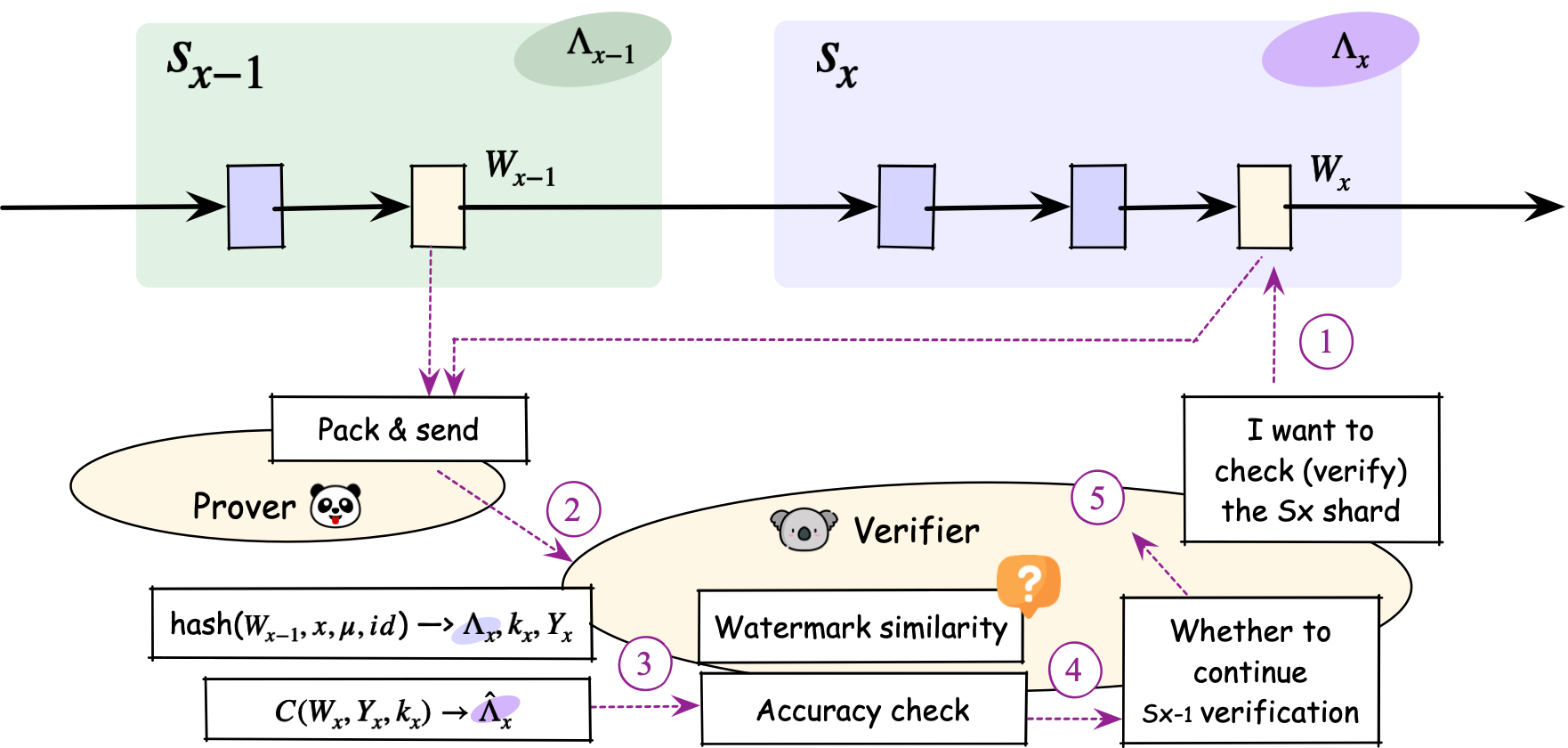}
      %\vspace{0.05in}
        \caption{\textbf{Verification of chained watermarking for PoL:}
\ding{172} (\textcolor{violet}{verifier}) selects shard $x$ and requests verification; 
\ding{173} (\textcolor{violet}{prover}) returns checkpoints $W_{x-1}$ and $W_x$; 
\ding{174} (\textcolor{violet}{verifier}) computes $(\Lambda_x, k_x, Y_x)$ from $W_{x-1}$, then extracts $\hat{\Lambda}_x$ from $W_x$ using $k_x$ and $Y_x$; 
\ding{175} (\textcolor{violet}{verifier}) checks watermark similarity and main-task accuracy; 
\ding{176} (\textcolor{violet}{verifier}) returns results and repeats earlier shards. 
        }
        \label{framwork_verify}
\end{figure}

\subsubsection{\underline{Verifying proof}}
\label{subsubsec: PoLO_verifying_proof}

Verification begins by checking whether the checkpoint $W_x$ meets the required accuracy on a public test set. This ensures the model has achieved meaningful task-specific performance, reflecting real-world expectations where only accurate models hold value in ML marketplaces~\cite{9445602}. If $W_x$ fails to meet the threshold, the proof is immediately rejected and the prover $\mathcal{P}$ is penalized, discouraging watermark embedding in inadequately trained models.

Once the model passes the main task performance evaluation, the verifier $\mathcal{V}$ proceeds to chained watermark verification. 
First, the verifier reconstructs the expected watermark $\Lambda_x$, the selection matrix $Y_x$, and the key $k_x$ for shard $s_x$ using the stored obfuscation-protected weights from shard $s_{x-1}$ (cf. \eqref{wm_k_gen}).
% \begin{equation}
% \begin{aligned}
%     \Lambda_x,\; k_x &= \textsf{WMGen}(\mathcal{H}_x),\; \textsf{KeyGen}(\mathcal{H}_x).
% \end{aligned}
% \end{equation}
This confirms that the watermark $\Lambda_x$ must have been generated from $W_{x-1}$ and cannot be arbitrarily forged or altered. 

The verifier $\mathcal{V}$ extracts the watermark $\hat{\Lambda}_x$ from the checkpoint $W_x$ using the selection matrix $Y_x$ and secret key $k_x$:
\begin{equation}
    \hat{\Lambda}_x = \mathbb{C}(W_x,Y_x,k_x).
\end{equation}
By comparing $\hat{\Lambda}_s$ with the watermark $\Lambda_s$ derived from (\ref{wm_k_gen}), the verifier $\mathcal{V}$ determines whether the watermark detection rate exceeds the predefined threshold $\eta^{ver}_{G}$. The watermark detection rate of the $x$-th shard is computed as:
\begin{equation}
    \eta_x = 1-\frac{ \sum_{i}^{n} (\hat{\Lambda}_x[i]\ne \Lambda_x[i])}{n}.
\end{equation}
If the extracted watermark satisfies $\eta_x \geq \eta^{ver}_{G}$, the \textsc{PoLO} proof is accepted. Otherwise, the verification fails due to potential tampering, watermark absence, or inconsistency in the \textsc{PoLO} proof (cf. Algorithm~\ref{algo:polo_verification} in \textbf{Appendix}~\ref{appendix_algo}).

\subsubsection{\underline{Cumulating security confidence}}
\label{subsubsec: PoLO_cumulation}

To reinforce the integrity of the \textsc{PoLO} proof, the verifier $\mathcal{V}$ may conduct backward verification by recursively validating previous shards, moving from $s_x$ to $s_{x-1}$, $s_{x-2}$, and so forth, until a sufficiently verifiable proof chain is established.
Since each watermark is cryptographically linked to its predecessor, backward verification prevents partial proof forgery. No malicious prover can manipulate only recent checkpoints while leaving earlier training shards unverifiable.
Hash-based chaining combined with cumulative verification ensures that the cost of forgery scales with depth and exceeds rational benefits, making partial manipulation economically\&computationally impractical.

 \vspace{3pt}
Upon successful verification, the verifier $\mathcal{V}$ can determine the correct rewards for the prover based on the number of validated shards $S$ and the performance on the main task.
A higher final model accuracy on the test set results in greater rewards, incentivizing the provers to optimize model performance continuously.
Moreover, each shard serves as a verifiable unit of training effort, making the total compensation directly proportional to the number of recognized shards $S$. 
This shard-based reward mechanism prevents the provers $\mathcal{P}$ from inflating their reported training workload by falsely claiming an excessive number of epochs for watermark embedding.
Such manipulation would not only fail to improve the main task but also distort fair reward distribution, making it impractical to assign credit or rewards based solely on individual epochs, making an epoch-based system impractical.

 \vspace{3pt}
The number of shards $S$ primarily depends on the number of training epochs required for watermark embedding. 
A higher learning rate enables faster watermark embedding, leading to more shards. 
However, an excessively high learning rate in later training stages can hinder the model's ability to improve its accuracy on the main task. 
As \textsc{PoLO} ties rewards directly to measurable contributions, it introduces a natural incentive mechanism that aligns computational effort with economic return.
Since rewards are determined by both the number of validated shards and the final model accuracy, a rational prover strategically balances the learning rate to maximize shard count and accuracy for optimal rewards.

 \vspace{3pt}
The verification mechanism in \textsc{PoLO} provides strong security guarantees. 
The chained watermarks ensure that forging a valid \textsc{PoLO} proof requires reconstructing the entire training sequence while preserving the model performance, making forgery computationally infeasible. 
Modifying any single shard invalidates all subsequent watermarks, rendering selective tampering impractical. 
Moreover, the shard-based verification strategy ensures fairness by tying rewards to the number of validated shards, preventing the provers $\mathcal{P}$ from inflating their workload claims.

 \vspace{3pt}
\noindent\textbf{Communication overhead.}
Verifying shard $s_x$ requires the prover to transmit two checkpoint models ($W_{x-1}$ and $W_x$), so the per-challenge bandwidth is $2 \times |W|$, where $|W|$ denotes the model size.
In practice, the verifier need not check every shard: random shard sampling (cf.\ $\S$\ref{subsubsec: PoLO_cumulation}) keeps the number of challenged shards small (e.g., $\mathcal{O}(\log S)$ challenges suffice for high confidence), making the total communication cost $\mathcal{O}(|W| \log S)$ rather than $\mathcal{O}(|W| \cdot S)$.
For settings where even two full checkpoints are prohibitively large, standard techniques such as checkpoint compression or delta encoding (transmitting $W_x - W_{x-1}$) can further reduce bandwidth.

\subsection{Security and Privacy Analysis}

%   - Training efforts
%       - spoofing

\textsc{PoLO} combines hash-based watermark chaining, backward verification, and shard-level performance checks to ensure training effort and ownership are provable, tamper-resistant, and cryptographically verifiable.
Building on the threat models in $\S$\ref{sec:preliminaries}, we address \textbf{\textcolor{violet}{RQ2}} by demonstrating how \textsc{PoLO} unifies PoL and PoO under economically rational and dual-legitimacy premises.
Key takeaways are summarized in \textbf{Appendix}~\ref{appendix_security_privacy}.

\subsubsection{\underline{Training efforts}}\label{subsubsec: threat_model_1}
\textsc{PoLO} counters adversarial attempts to manipulate or falsify the training process.

 \vspace{3pt}
\noindent\textbf{Forging \textsc{PoLO}.} 
\textsc{PoLO} introduces a chained watermarking framework, where each shard’s watermark is securely linked to the previous shard’s output via a hash function $\mathbb{H}(\cdot)$. This chaining enforces that any tampering or forgery attempt requires reconstructing the entire sequence of valid watermarks while maintaining main task performance. Leveraging the immutability of cryptographic hashes, 
%PoLO ensures that proof forgery is computationally infeasible, thereby preserving the integrity of verification.
\textsc{PoLO} ensures that honest training remains the only viable path, as attackers must either bear retraining costs or fail verification, making all shortcuts ineffective and unprofitable.

% One attack in this regard that is not easily realized is attackers can forge the key that is used to extract the watermark given that he has been aware of how the watermark is computed

 \vspace{3pt}
\noindent\textbf{Exaggerated workload claims.} 
A rational prover may attempt to inflate training workload by increasing the number of claimed shards. \textsc{PoLO} counters this by verifying at the shard level, where each shard is a complete unit defined by successful watermark embedding rather than by a claimed epoch count. Higher learning rates may produce more shards by accelerating embedding, but they can also harm task performance; because rewards depend on both verified shards and model accuracy, the prover is incentivized to balance efficiency and fidelity.

A prover might also try to fabricate additional shards without performing the required computation. However, $(k_{x+1},Y_{x+1})$ are deterministically tied to the previous checkpoint and verifier nonce $\mu$, so valid watermark placement and extraction cannot be chosen arbitrarily. This binds shard formation to genuine trainings and makes shard inflation computationally\&economically unattractive.

\subsubsection{\underline{Ownership}}\label{subsubsec: threat_model_2}
\textsc{PoLO} is compatible with existing watermarking techniques and does not compromise their inherent security guarantees. On the contrary, the chaining structure introduced by \textsc{PoLO} can enhance robustness against certain attack vectors, such as watermark forgery, by cryptographically binding each watermark to prior authenticated states.

 \vspace{3pt}
\noindent\textbf{Ownership enforcement in EMMs.}
Within EMMs, \textsc{PoLO} provides verifiers with strong tools to enforce ownership and training provenance. By integrating advanced watermarking methods such as RIGA~\cite{riga} and FedIPR~\cite{li2022fedipr}, \textsc{PoLO} keeps watermarks stealthy and robust against adversarial modifications, enabling reliable ownership verification and resilience against watermark-related attacks within the EMM verification pipeline.
\textsc{PoLO}'s ownership guarantees apply specifically to environments that enforce watermark-based verification (i.e., EMMs). Redistribution through channels that do not perform such verification, including public model hubs, private sharing, API serving, and model distillation, lies outside the EMM threat model and is not claimed as a protection target.

 \vspace{3pt}
\noindent\textbf{Ownership evasion and impersonation.}
Structural watermark-removal strategies, including weight shifting and Neural Structural Obfuscation (NSO)~\cite{287178}, do not address the attacker's fundamental requirement to re-embed a forged watermark in EMMs; moreover, NSO transforms can be reversed via canonicalization (cf.\ \textbf{threat model}). Consequently, the relevant attack surface in \textsc{PoLO} is the full \emph{removal-plus-re-embedding} phase. To impersonate ownership and profit, an attacker must \emph{erase} the incumbent legal watermark and \emph{embed} a forged watermark (or set of watermarks) that preserves task performance and passes verification, typically via full-model fine-tuning. Within a shard (from $s_{x-1}$ to $s_x$), the cost of such full-model fine-tuning plus re-embedding is empirically comparable to, or exceeds, honest per-shard training. As a result, ownership impersonation through watermark removal and replacement is economically unattractive for rational adversaries in EMMs.

% \noindent\textbf{On the fixed selection matrix $Y$.}
% By design, the watermark position mask $Y$ is fixed across all shards to enable deterministic reconstruction by the verifier from the shared nonce $\mu$. An attacker with access to multiple checkpoint models could, in principle, observe which weight positions are consistently modified and infer $Y$. However, knowledge of $Y$ alone is insufficient to compromise ownership:
% (1)~the embedding key $k_x$ changes per shard (derived from $\mathbb{H}(W_{x-1}, x, \mu, id_\mathcal{P})$), so the attacker cannot extract or forge the watermark content $\Lambda_x$ without $k_x$;
% (2)~the watermark value $\Lambda_x$ itself is hash-derived and unpredictable;
% (3)~any attempt to overwrite the marked region must still preserve task accuracy and pass the $\eta_G$ threshold for the legitimate watermark, which our experiments show is resilient even under full-model fine-tuning (cf.\ $\S$\ref{sec:experiment}).
% As a potential extension, per-shard position rotation ($Y_x = \textsf{WMPosition}(\mu, x)$) could further strengthen positional secrecy, though at the cost of increased verification complexity.

 \vspace{3pt}
\noindent\textbf{On the chained selection matrix $Y_x$.}
\textsc{PoLO} derives a shard-specific position mask $Y_x=\textsf{WMPosition}(\mathcal{H}_{x-1})$, unlike static-position designs. This reduces cross-shard transferability of position inference: even if an attacker partially infers one shard's marked region, subsequent shards use different positions tied to different chain states. Overwriting must preserve task accuracy while defeating the legitimate watermark.

\vspace{3pt}
Our experimental results (\S\ref{sec:experiment}) provide empirical support for this ownership argument: an attack may preserve utility but often leaves the legitimate watermark detectable, while a cheaper attack typically fails the task-admissibility requirement.

\vspace{3pt}
Tab.\ref{tab:pol-attack-applicability} (cf.~\textbf{Appendix}~\ref{appendix_reply-pol_attack}) gives an analytical (non-experimental) applicability check for representative replay-PoL attacks. The key point is that \textsc{PoLO} verifies shard receipts and chained watermark extraction instead of replaying gradient trajectories, so those replay-specific exploit surfaces are removed by design.

\vspace{3pt}
\noindent\textbf{Unforgeability game.}
We formalize the security goal via the following game between a challenger $\mathcal{C}$ and an adversary $\mathcal{A}$:

\vspace{3pt}

\begin{packeditemize}
\item \textbf{Setup.} $\mathcal{C}$ generates a nonce $\mu$ and prover identity $id_\mathcal{P}$, trains a model with \textsc{PoLO} producing $S$ shards with checkpoints $\{W_0, \ldots, W_S\}$, and gives all checkpoints to $\mathcal{A}$.
\item \textbf{Forgery.} $\mathcal{A}$ outputs a forged proof $\mathbb{P}' = (W'_{x-1}, W'_x)$ for some shard $s_x$, along with the total computational cost $C_\mathcal{A}$.
\item \textbf{Win condition.} $\mathcal{A}$ wins if: (i)~the verifier $\mathcal{V}$ accepts $\mathbb{P}'$ (i.e., the forged watermark $\hat{\Lambda}'_x$ extracted from $W'_x$ satisfies $\eta'_x \geq \eta^{ver}_{G}$ against the expected watermark derived from $W'_{x-1}$, and the main task accuracy meets the threshold), (ii)~$C_\mathcal{A} < C_{\text{honest}}$, the cost of honestly training shard $s_x$, and (iii)~$(W'_{x-1}, W'_x) \neq (W_{x-1}, W_x)$ (non-triviality: the forged proof must differ from the honest one).
\end{packeditemize}

\vspace{3pt}
\noindent\textsc{PoLO} is \emph{economically unforgeable} if no polynomial-time adversary can win the above game with non-negligible probability, assuming the cryptographic hash function $\mathbb{H}$ is collision-resistant and preimage-resistant. The chained structure ensures that producing a valid $W'_{x-1}$ requires either (a)~possessing the genuine checkpoint (no cost advantage) or (b)~reconstructing the chain from scratch (cost $\geq C_{\text{honest}}$), while the per-shard key derivation prevents shortcutting individual shards.

\subsubsection{\underline{Dataset privacy}}\label{subsubsec: threat_model_3}
\textsc{PoLO} enhances data privacy by removing the need for provers to share training data during verification. Unlike prior PoL methods that rely on dataset disclosure, \textsc{PoLO} achieves verifiability through obfuscation-protected model checkpoints and chained watermarks. To mitigate inference risks from gradient leakage, Gaussian noise is applied at the end of each shard, ensuring the prover’s dataset remains secure. %~\cite{du2024sok}.

%  An adversary cannot construct a valid proof with less cost (in both computation and storage) than that made by the prover in generating the proof.

\section{Experiments}\label{sec:experiment}

We evaluate \textsc{PoLO} to demonstrate its practical viability. 
Our experiments focus on three key aspects: 
% \textit{fidelity} (model performance impact), 
\textit{Compatibility} (broad applicability), 
\textit{Overhead} (low computational costs), 
and \textit{Unforgeability} (strong resistance to forgery).

\subsection{Experimental Settings}
To ensure comprehensive evaluation and broad applicability, we conduct experiments across diverse model architectures and multiple datasets.
All experiments are performed on a dedicated micro-server infrastructure, providing a standardized and controlled environment for consistent model evaluation and comparison.

\smallskip
\noindent\textbf{Implementation environment.}
The system features dual Intel Xeon Gold 6126 CPUs (12 cores each), 192GB of 2666MHz ECC DDR4 RAM across six channels, and 2×1.2TB 10,000 RPM SAS II drives in RAID 1. It is powered by two NVIDIA Tesla V100 GPUs with 5,120 CUDA cores, 640 Tensor cores, and 32GB memory each, running on 64-bit Ubuntu 18.04.

\begin{table}[t]
\centering
\caption{Analytical mapping of watermark-attack classes under \textsc{PoLO}'s admissibility rules.}
\label{tab:attack-taxonomy}
\renewcommand{\arraystretch}{1.05}
\resizebox{\linewidth}{!}{
\begin{threeparttable}
\begin{tabular}{l|cc|c}
\midrule
\multicolumn{1}{c}{\textbf{Watermark Attack}} 
  & \multicolumn{2}{c}{\textbf{Attack class}} 
  & \multicolumn{1}{c}{\textbf{Applicable in PoLO?}} \\

\multicolumn{1}{c}{} 
  & \cellcolor{yellow!15}OFA 
  & \cellcolor{yellow!15}WMA 
  & \cellcolor{yellow!15}Yes/No \\
\midrule

Fine tuning overwrite~\cite{fine-tune_attck,fine-tune_embed} 
  & \cmark & \xmark & \cmark \\

Pruning or weight masking~\cite{prune_attack} 
  & \xmark & \cmark & \cmark \\

Watermark overlap~\cite{overlap_attack} 
  & \xmark & \cmark & \cmark \\

Collusion or ensemble stripping~\cite{9728937} 
  & \xmark & \cmark & \cmark\tnote{a} \\

Key recovery or position guessing~\cite{li2021securewatermarkdeepneural} 
  & \xmark & \cmark & \cmark\tnote{b} \\
\midrule

\rowcolor{blue!10}
\textbf{Removal (no re-embedding)~\cite{287178,watermark_attack1}}  % sok_attack,
  & \xmark & \xmark & \xmark\tnote{c} \\

\rowcolor{blue!10}
\textbf{Removal+re-embedding (unknown positions)} 
  & \cmark & \cmark & \cmark\tnote{d} \\
\cmidrule{1-4}
\end{tabular}
\begin{tablenotes}

\item[a] Collusion acts as a weight-manipulation attack: colluders average or tweak weights, then re-embed a watermark, so we place it in the WMA class.
\item[b] Key recovery / position guessing is only a pre-step; an attacker still needs OFA or WMA to submit an admissible forged claim.
\item[c] Removal without re-embedding is non-monetizable, out of target for impersonation analysis.
\item[d] Removal+re-embedding remains in scope; feasibility depends on jointly passing receipt, chain, and task-admissibility checks.
\end{tablenotes}
\end{threeparttable}
}
\end{table}

\vspace{3pt}
\noindent\textbf{Datasets.} We evaluate on three image classification datasets and one text classification dataset. CIFAR-10 and CIFAR-100 each contain 60{,}000 color images of size $32 \times 32$ pixels (50{,}000 training / 10{,}000 test), with 10 and 100 classes respectively. TinyImageNet comprises 100{,}000 training images across 200 classes, resized to $64 \times 64$ pixels. AG News provides 120{,}000 training and 7{,}600 test samples across 4 news categories.

\vspace{3pt}
% \noindent\textbf{Model architectures.} We adopt a diverse set of neural architectures: \textit{AlexNet}~\cite{alexnet}, \textit{ResNet18}/\textit{34}~\cite{resnet}, \textit{WideResNet}~\cite{wideresnet}, \textit{VGG16}~\cite{vgg}, \textit{TextCNN}~\cite{textcnn}, and \textit{MiniBert}~\cite{bert}. \textit{AlexNet} and \textit{VGG16} offer strong baselines for vision tasks; \textit{ResNet} introduces residual connections for deeper training; \textit{WideResNet} boosts accuracy via increased layer width. For NLP, \textit{TextCNN} captures local features through 1D convolutions, while \textit{MiniBert} provides efficient contextual embedding via transformer-based pretraining.
\noindent\textbf{Model architectures and training.} We evaluate six architectures spanning vision and NLP: \textit{AlexNet}, \textit{ResNet18}/\textit{34}, \textit{WideResNet}, \textit{VGG16}, \textit{TextCNN}, and \textit{MiniBert}. Specifically, we pair AlexNet and ResNet18 with CIFAR-10, ResNet18 and WideResNet with CIFAR-100, VGG16 and ResNet34 with TinyImageNet, and TextCNN and MiniBERT with AG News. All models are trained with SGD (momentum 0.9) and a batch size of 256, using learning rates of 0.1 for CIFAR-10/100, 0.05 for TinyImageNet, and 0.005 for AG News. The watermark regularization coefficient in Eq.~2 is set to $\lambda = 0.01$ uniformly across all experiments and watermarking schemes.

\vspace{3pt}
\noindent\textbf{Attack settings.}
Assuming that the attacker $\mathcal{A}$ has access to all checkpoint models $W_x$, they may attempt to forge a \textsc{PoLO} proof $\mathbb{P}^{'}$.  As summarized in Tab.\ref{tab:attack-taxonomy}, most attacks studied in the watermarking literature instantiate either of the two representative attack families, while pure removal attacks fall into the economically irrational bucket that is excluded by the EMM context.

\vspace{3pt}
\begin{packeditemize}
\item \textit{Overhaul Fine-tuning Attack (OFA) with training data.}  
OFA attempts to replace the original watermark and falsely assert ownership by fine-tuning. %~\cite{fine-tune_attck}. %, fine-tune_embed
Given access to checkpoint models but not the legitimate keys ($k_x$, $Y_x$), the attacker $\mathcal{A}$ forges illicit proofs $\mathbb{P}^{'}$ by fine-tuning all model weights to replace the authentic proof $\mathbb{P}$. 
When training data is available, $\mathcal{A}$ additionally embeds fake watermarks $\Lambda^{'}_{x}$ during fine-tuning to overwrite legitimate ones. To balance cost and impact, five fine-tuning rounds are applied for 2048-bit and three rounds for 1024-bit versions.

\item \textit{Weight Manipulation Attack (WMA) without training data.}
WMA abstracts pruning-style and overlap attacks~:%\cite{overlap_attack}: %prune_attack,
without access to training data, the attacker $\mathcal{A}$ selects a target watermark $\Lambda’_x$ and embedding keys, directly manipulates chosen weights in each intermediate model to overwrite them with $\Lambda’_x$.
\end{packeditemize}
If an attacker $\mathcal{A}$ can generate a forged proof $\mathbb{P}^{'}$ with lower computational cost than a legitimate proof $\mathbb{P}$ while successfully corrupting the original watermarks $\Lambda_x$, the forgery may pass verification.

\smallskip
\noindent\textbf{Evaluation metrics.}
We consider three metrics:
\begin{packeditemize}
% \item \textit{Fidelity} evaluates the extent to which the model preserves its original performance after watermark embedding. This is measured by the main task accuracy ($Acc_{main}$), expectating that watermarking introduces little to no degradation in predictive capability.
\item \textit{Compatibility} evaluates the applicability of various watermarking schemes for \textsc{PoLO} implementation by assessing the work transition efficiency, measured as the \textit{shard rate} $|s| = \frac{\text{shards}}{\text{epochs}}$. An optimal shard rate falls within the range $0 < |s| < 1$, indicating efficient shard generation relative to training epochs.
\item \textit{Overhead} measures computational efficiency: proof generation should add only a small fraction of training cost, and verification should be much cheaper than generation. We also convert wall-clock time to monetary cost (USD) via a cost-based pricing model.%~\cite{time2money}.
% \item \textit{Security} evaluates the ability to protect watermarks embedded in early training shards from extraction or recovery through analysis of final model weights, guaranteeing temporal privacy.
\item \textit{Unforgeability} evaluates whether the computational overhead required to forge a \textsc{PoLO} proof $\mathbb{P}^{'}$ surpasses that of generating a legitimate proof $\mathbb{P}$, while simultaneously verifying that any forgery attempt results in the destruction of the legitimate proof $\mathbb{P}$. A forged proof $\mathbb{P}^{'}$ can only pass verification if both these conditions are satisfied.
% \item \textit{Integrity} measures the ability to correctly differentiate between genuinely trained models and independently trained ones, ensuring that honest models are not mistakenly identified as forgeries.
\end{packeditemize}

\begin{table*}[!ht]
\centering
\caption{The comparison of shard rate $|s|$, main task performance $Acc_{main}$, and PoLO generation time (+training), using different watermarks.} 
\label{shard rate}
\renewcommand{\arraystretch}{1}
\resizebox{\textwidth}{!}{
\begin{threeparttable}
\begin{tabular}{ccc cc |cc cc |cc}
\toprule
\makecell{Watermark \\ \textbf{methods}} & \makecell {Watermark \\ \textbf{size}(bits)} &  \textbf{Training} & \makecell{CIFAR-10 \\ AlexNet} & \multicolumn{1}{c}{\makecell{CIFAR-10 \\ ResNet18}} & \makecell{CIFAR-100 \\ ResNet18} & \multicolumn{1}{c}{\makecell{CIFAR-100 \\ WideResNet} } & \makecell{TinyImageNet \\ VGG16} &\multicolumn{1}{c}{ \makecell{TinyImageNet \\ ResNet34}} & \makecell{AG News \\ TextCNN} & \makecell{AG News \\ MiniBert} \\ 

             \midrule
                        \multirow{9}{*}{RIGA~\cite{riga}}  & \multirow{3}{*}{2048} & \multicolumn{1}{c|}{\makecell{$|s|$(\%)}}   & \cellbarSE{37.07} & \multicolumn{1}{c}{ \cellbarSE{37.62}} & \cellbarSE{38.61} & \cellbarSE{35.83} & \cellbarSE{74.51} & \multicolumn{1}{c}{\cellbarSE{73.08}} & \cellbarSE{43.64} & \cellbarSE{64.15} \\ 
                        & & \multicolumn{1}{c|}{\makecell{$Acc_{main}$(\%)}}  & \cellbarACC{89.22} & \multicolumn{1}{c}{\cellbarACC{91.89}} & \cellbarACC{74.81} & \cellbarACC{72.83} & \cellbarACC{73.00} & \multicolumn{1}{c}{\cellbarACC{72.46}} & \cellbarACC{91.04} & \cellbarACC{91.88} \\
                        & & \multicolumn{1}{c|}{Time(s)}  & 990.74 & \multicolumn{1}{c}{1618.78} & 1707.59 & 9357.71 & 25051.72 & 7332.32 & 1070.12 & 2573.38 \\ 
                        
                        \cline{2-3} 
                        
                        & \multirow{3}{*}{1024}  & \multicolumn{1}{c|}{\makecell{$|s|$(\%)}}  & \cellbarSE{64.04} & \multicolumn{1}{c}{\cellbarSE{64.60}} & \cellbarSE{71.29} & \cellbarSE{63.48} & \cellbarSE{96.15} & \cellbarSE{96.15} & \cellbarSE{71.29} & \cellbarSE{93.07} \\
                        & & \multicolumn{1}{c|}{\makecell{$Acc_{main}$(\%)}} & \cellbarACC{88.94} & \multicolumn{1}{c}{\cellbarACC{91.96}} & \cellbarACC{74.93} & \cellbarACC{72.96} & \cellbarACC{73.33} & \cellbarACC{73.31} & \cellbarACC{90.94} & \cellbarACC{92.06} \\
                        & & \multicolumn{1}{c|}{Time(s)}  & 1066.78 & \multicolumn{1}{c}{1861.62} & 1681.42 & 8968.80 & 24934.03 & 7238.95 & 1067.23 & 2687.89 \\ 
                        
                        \cline{2-3} 
                        
                        & \multirow{3}{*}{512}  & \multicolumn{1}{c|}{\makecell{$|s|$(\%)}}   & \cellbarSE{87.74} & \multicolumn{1}{c}{\cellbarSE{92.08}} & \cellbarSE{86.79} & \cellbarSE{87.74} & \cellbarSE{96.15} & \cellbarSE{96.15} & \cellbarSE{74.29} & \cellbarSE{95.05} \\ 
                        & & \multicolumn{1}{c|}{\makecell{$Acc_{main}$(\%)}}  & \cellbarACC{89.54} & \multicolumn{1}{c}{\cellbarACC{92.22}} & \cellbarACC{75.18} & \cellbarACC{72.64} & \cellbarACC{73.91} & \cellbarACC{73.22} & \cellbarACC{90.65} & \cellbarACC{91.82} \\ 
                        & & \multicolumn{1}{c|}{Time(s)}  & 993.12 & \multicolumn{1}{c}{1783.50} & 1798.00 &  \multicolumn{1}{c|}{8433.43} & 25442.66 & 6775.54 & 1051.52 & 2486.51 \\
                        
              \cline{1-3}
                        
                        \multirow{9}{*}{EDNN~\cite{Uchida}} & \multirow{3}{*}{2048}  & \multicolumn{1}{c|}{\makecell{$|s|$(\%)}}  & \cellbarSE{35.64} & \multicolumn{1}{c}{\cellbarSE{35.64}} & \multicolumn{1}{c}{\cellbarSE{35.64}} & \multicolumn{1}{c|}{\cellbarSE{36.27}} & \cellbarSE{70.00} & \cellbarSE{70.00} & \cellbarSE{37.62} & \cellbarSE{47.06} \\
                        & & \multicolumn{1}{c|}{\makecell{$Acc_{main}$(\%)}}  & \cellbarACC{89.88} & \multicolumn{1}{c}{\cellbarACC{92.07}} & \cellbarACC{74.32} & \multicolumn{1}{c|}{\cellbarACC{72.65}} & \cellbarACC{73.61} & \cellbarACC{71.89} & \cellbarACC{90.29} & \cellbarACC{91.71} \\ 
                        & & \multicolumn{1}{c|}{Time(s)}  & 880.91 & \multicolumn{1}{c}{1518.95} & 1689.89 & \multicolumn{1}{c|}{8423.67} & 24492.53 & 6955.89 & 896.11 & 2333.33 \\ 
                        
                        \cline{2-3} 
                        
                        & \multirow{3}{*}{1024} & \multicolumn{1}{c|}{\makecell{$|s|$(\%)}}  & \cellbarSE{39.17} & \multicolumn{1}{c}{\cellbarSE{45.54}} & \cellbarSE{45.54} & \multicolumn{1}{c|}{\cellbarSE{44.12}} & \cellbarSE{83.00}  & \multicolumn{1}{c}{\cellbarSE{78.00}} & \cellbarSE{46.60} & \cellbarSE{60.00} \\
                        & & \multicolumn{1}{c|}{\makecell{$Acc_{main}$(\%)}} & \cellbarACC{89.67} & \multicolumn{1}{c}{\cellbarACC{91.55}} &  \cellbarACC{73.98} &  \multicolumn{1}{c|}{\cellbarACC{72.86}} &  \cellbarACC{73.91} &  \multicolumn{1}{c}{\cellbarACC{72.01}} &  \cellbarACC{90.78} &  \cellbarACC{91.99} \\
                        & & \multicolumn{1}{c|}{Time(s)} & 923.03 & \multicolumn{1}{c}{1493.05} & 1708.98 & \multicolumn{1}{c|}{7876.67} & 24342.67 & \multicolumn{1}{c}{7501.34} & 956.65 & 2435.32 \\ 
                         
                        \cline{2-3}
                         
                        & \multirow{3}{*}{512} & \multicolumn{1}{c|}{\makecell{$|s|$(\%)}}   & \cellbarSE{57.27} &  \cellbarSE{60.95} & \cellbarSE{52.00} & \multicolumn{1}{c|}{\cellbarSE{54.24}} & \multicolumn{1}{c}{\cellbarSE{96.00}} & \multicolumn{1}{c}{\cellbarSE{96.00}} & \cellbarSE{61.39} & \cellbarSE{73.79} \\ 
                        & & \multicolumn{1}{c|}{\makecell{$Acc_{main}$(\%)}}  & \cellbarACC{89.89} & \cellbarACC{91.62} & \cellbarACC{73.89} & \multicolumn{1}{c|}{\cellbarACC{72.16}} & \cellbarACC{73.30} & \multicolumn{1}{c}{\cellbarACC{71.88}} & \cellbarACC{90.56} & \cellbarACC{91.89} \\ 
                        & & \multicolumn{1}{c|}{Time(s)}  & 949.75 & 1711.73 & 1665.76 &  7966.58 & 24689.07 & \multicolumn{1}{c}{7683.83} & 998.02 & 2534.32 \\

             \cline{1-3}
                        
                       \multirow{9}{*}{FedIPR~\cite{li2022fedipr}} & \multirow{3}{*}{2048} & \multicolumn{1}{c|}{\makecell{$|s|$(\%)}} & \cellbarSE{40.78} &  \cellbarSE{40.59} &  \cellbarSE{58.00} & \cellbarSE{32.67} & \cellbarSE{70.00} & \multicolumn{1}{c}{\cellbarSE{80.39}} & \cellbarSE{65.69} & \cellbarSE{35.64} \\ 
                       & & \multicolumn{1}{c|}{\makecell{$Acc_{main}$(\%)}} & \cellbarACC{88.92} & \cellbarACC{91.79} & \cellbarACC{74.67} & \cellbarACC{72.93} & \cellbarACC{73.23} & \multicolumn{1}{c}{\cellbarACC{71.99}} & \cellbarACC{90.32} & \cellbarACC{91.43} \\
                        & & \multicolumn{1}{c|}{Time(s)} & 993.91 & 1779.48 & 1594.15 & 9330.49 & 25089.67 & \multicolumn{1}{c}{6941.68} & 1050.88 & 2691.72 \\ 
                        
                        \cline{2-3}
                        
                        & \multirow{3}{*}{1024} & \multicolumn{1}{c|}{\makecell{$|s|$(\%)}}   & \cellbarSE{74.26} & \cellbarSE{86.27} & \cellbarSE{67.33} & \cellbarSE{88.12} & \cellbarSE{92.31} & \multicolumn{1}{c}{\cellbarSE{96.15}} & \cellbarSE{71.29} & \cellbarSE{69.31} \\ 
                        & & \multicolumn{1}{c|}{\makecell{$Acc_{main}$(\%)}} & \cellbarACC{89.38} & \cellbarACC{91.43} & \cellbarACC{74.30} & \cellbarACC{71.54} & \cellbarACC{74.30} & \multicolumn{1}{c}{\cellbarACC{71.76}} & \cellbarACC{90.37} & \cellbarACC{91.87} \\ 
                        & & \multicolumn{1}{c|}{Time(s)}  & 996.87  & 1693.35 & 1643.65 & 8901.83 & 24810.68 & \multicolumn{1}{c}{6891.76} & 994.08 & 2499.38 \\

                        \cline{2-3} 
                        
                        & \multirow{3}{*}{512}  & \multicolumn{1}{c|}{\makecell{$|s|$(\%)}}   & \cellbarSE{92.08} & \multicolumn{1}{c}{\cellbarSE{93.14}} & \cellbarSE{81.37} & \cellbarSE{91.18} & \multicolumn{1}{c}{\cellbarSE{96.15}} & \multicolumn{1}{c}{\cellbarSE{96.15}} & \cellbarSE{94.12} & \cellbarSE{93.14} \\ 
                        & & \multicolumn{1}{c|}{\makecell{$Acc_{main}$(\%)}}  & \cellbarACC{89.04} & \multicolumn{1}{c}{\cellbarACC{92.10}} & \cellbarACC{74.09} & \cellbarACC{72.01} & \cellbarACC{73.87} & \multicolumn{1}{c}{\cellbarACC{72.01}} & \cellbarACC{90.50} & \cellbarACC{91.87} \\ 
                        & & \multicolumn{1}{c|}{Time(s)}  & 1004.12 & \multicolumn{1}{c}{1750.65} & 1674.89 & 9354.67 & 25102.65 & \multicolumn{1}{c}{7354.67} & 1059.54 & 2652.31 \\

\bottomrule
\end{tabular}
\begin{tablenotes} % Added tablenotes environment
    \item $\bullet$ A larger shard rate $|s|$ means less time spent embedding watermarks, indicating weaker security but reduced waste of intra-shard training effort, as effort is accounted for at the shard level regardless of how much is invested within each shard.

\end{tablenotes}
\end{threeparttable}
}
\end{table*}

\subsection{Experimental Results}

\begin{table*}[t]
\centering
\caption{The comparison between \textsc{PoLO} and existing PoL in terms of the time consumption of \textsc{PoLO} generation (+training) and verification.}
\label{tab_gen-ver}
\renewcommand{\arraystretch}{1}
\resizebox{\linewidth}{!}{
\begin{threeparttable}
\begin{tabular}{cc|c| cccccccc}
\toprule
\multicolumn{1}{l}{\makecell{Watermark\\ \textbf{size}(bit)}} & \multicolumn{1}{l}{\makecell{\textbf{Time} (s)}} & \multicolumn{1}{c}{\textbf{Methods}} & \makecell{CIFAR-10 \\ AlexNet} & \makecell{CIFAR-10\\ResNet18} & \makecell{CIFAR-100\\ResNet18} & \makecell{CIFAR-100\\WideResNet} & \makecell{TinyInageNet\\ VGG16} & \makecell{TinyInageNet\\ResNet34} & \makecell{AG News\\TextCNN} & \makecell{AG News\\Bert} \\ 

\toprule

\multirow{5}{*}{2048} & Training & \multicolumn{1}{c|}{Baseline}   & 970.45 & 1598.37 & 1691.38 & 9340.34 & 25031.45 & 7312.89 & 1052.38 & 2551.45      \\

\cline{2-3}

 &  \multirow{3}{*}{ \makecell{Training+$\mathbb{P}$ gen \\  (v.s.) \textcolor{violet}{$\mathbb{P}$ gen}}}  & \multicolumn{1}{c|}{PoLO} & \cellcolor{blue!8}990.74/\textcolor{violet}{20.29} &\cellcolor{blue!8}1618.78/\textcolor{violet}{20.41} & \cellcolor{blue!8}1707.59/\textcolor{violet}{16.21 } & \cellcolor{blue!8}9357.91/\textcolor{violet}{17.57 } & \cellcolor{blue!8}25051.72/\textcolor{violet}{20.27 } & \cellcolor{blue!8}7332.32/\textcolor{violet}{19.43 } & \cellcolor{blue!8}1070.12/\textcolor{violet}{17.74 } & \cellcolor{blue!8}2573.38/\textcolor{violet}{21.93 } \\ 
 &  & \multicolumn{1}{c|}{PoL (Vanilla)} & 975.33/\textcolor{violet}{4.88 } & 1603.47/\textcolor{violet}{5.10 } & 1695.58/\textcolor{violet}{4.20 } & \multicolumn{1}{c|}{9343.83/\textcolor{violet}{3.49 }} & 25035.45/\textcolor{violet}{4.00 } & 7316.12/\textcolor{violet}{3.23 } & 1055.23/\textcolor{violet}{2.85 } & 2556.15/\textcolor{violet}{4.70 } \\ 
 &   & \multicolumn{1}{c|}{PoL (hash)} & 979.13/\textcolor{violet}{8.68} &  \multicolumn{1}{c}{1604.67/\textcolor{violet}{6.30}} & 1696.78/\textcolor{violet}{5.40} &  \multicolumn{1}{c|}{9347.43/\textcolor{violet}{7.09}} & 2504.65/\textcolor{violet}{11.20} & 7317.42/\textcolor{violet}{4.53} & 1056.63/\textcolor{violet}{4.25} & 2557.55/\textcolor{violet}{6.10} \\
 
\cline{2-2}
 
 &  \multirow{3}{*}{$\mathbb{P}$ verify} & \multicolumn{1}{c|}{PoLO} & \cellcolor{yellow!15}83.37 & \cellcolor{yellow!15}75.96 & \cellcolor{yellow!15}82.73 & \multicolumn{1}{c|}{\cellcolor{yellow!15}211.06} & \cellcolor{yellow!15}342.17 & \cellcolor{yellow!15}135.43 & \cellcolor{yellow!15}44.54 & \cellcolor{yellow!15}56.29 \\
 &   & \multicolumn{1}{c|}{PoL (Vanilla)} & 891.99 & 1530.55 & 1612.91 & \multicolumn{1}{c|}{9132.82} & 24693.31 & 7180.73 & 1010.71 & 2495.89   \\ 
 &   & \multicolumn{1}{c|}{PoL (hash)} & 895.79 & 1531.75 & 1614.10 & \multicolumn{1}{c|}{9136.42} & 24700.51 & 7182.03 & 1012.11 & 2497.29   \\

\cline{1-3}\cline{10-11} %-------------

\multirow{5}{*}{1024} & Training & \multicolumn{1}{c|}{Baseline}   & 1043.32 & 1832.45 & 1653.78 & \multicolumn{1}{c|}{8945.56} & 24910.32 & 7215.56 & 1040.67 & 2653.75 \\

\cline{2-3}

& \multirow{3}{*}{\makecell{Training+$\mathbb{P}$ gen \\  (v.s.) \textcolor{violet}{$\mathbb{P}$ gen}}} & \multicolumn{1}{c|}{PoLO} & \cellcolor{blue!8}1066.78/\textcolor{violet}{23.46} &  \multicolumn{1}{c|}{\cellcolor{blue!8}1861.62/\textcolor{violet}{29.17}} & \cellcolor{blue!8}1681.42/\textcolor{violet}{27.64} &  \multicolumn{1}{c|}{\cellcolor{blue!8}8968.80/\textcolor{violet}{23.24}} & \cellcolor{blue!8}24934.03/\textcolor{violet}{23.71} & \cellcolor{blue!8}7238.95/\textcolor{violet}{23.39} & \cellcolor{blue!8}1067.23/\textcolor{violet}{26.56} & \cellcolor{blue!8}2687.89/\textcolor{violet}{34.14} \\
 &   & \multicolumn{1}{c|}{PoL (Vanilla)} & 1047.57/\textcolor{violet}{4.25} &  \multicolumn{1}{c|}{1837.13/\textcolor{violet}{4.68}} & 1658.59/\textcolor{violet}{4.81} &  \multicolumn{1}{c|}{8948.85/\textcolor{violet}{3.29}} & 24914.52/\textcolor{violet}{4.20} & 7219.32/\textcolor{violet}{3.76} & 1044.31/\textcolor{violet}{3.64} & 2658.57/\textcolor{violet}{4.82} \\
  &   & \multicolumn{1}{c|}{PoL (hash)} & 1051.37/\textcolor{violet}{8.05} &  \multicolumn{1}{c|}{1838.32/\textcolor{violet}{5.88}} & 1659.78/\textcolor{violet}{6.00} &  \multicolumn{1}{c|}{8952.45/\textcolor{violet}{6.89}} & 24921.71/\textcolor{violet}{11.40} & 7220.61/\textcolor{violet}{5.57} & 1045.71/\textcolor{violet}{5.04} & 2659.97/\textcolor{violet}{6.22} \\
 
\cline{2-2}

& \multirow{3}{*}{$\mathbb{P}$ verify}  & \multicolumn{1}{c|}{PoLO} & \cellcolor{yellow!15}157.60 &  \multicolumn{1}{c|}{\cellcolor{yellow!15}149.07} & \cellcolor{yellow!15}149.63 & \cellcolor{yellow!15}380.86 & \cellcolor{yellow!15}456.77 & \cellcolor{yellow!15}172.15 & \cellcolor{yellow!15}67.90 & \cellcolor{yellow!15}78.57 \\
&     & \multicolumn{1}{c|}{PoL (Vanilla)} & 887.03 &  \multicolumn{1}{c|}{1688.14} & 1509.04 & 8566.07 & 24455.79 & 7045.22 & 974.43 & 2580.03 \\
&   & \multicolumn{1}{c|}{PoL (hash)} & 890.83 & \multicolumn{1}{c|}{1689.34} & 1510.24 & \multicolumn{1}{c}{8569.67} & 24462.98 & 7046.51 & 975.83 & 2581.43   \\

\cline{1-3}\cline{10-11}  %-------------

\multirow{5}{*}{512}  & Training    & \multicolumn{1}{c|}{Baseline} & 963.13 &  \multicolumn{1}{c|}{1754.67} & 1763.65 & 8402.83  & 25419.39 & 6751.64 & 1023.67 & 2451.54 \\ 

\cline{2-3} 

 & \multirow{3}{*}{\makecell{Training+$\mathbb{P}$ gen \\ (v.s.) \textcolor{violet}{$\mathbb{P}$ gen}}} & \multicolumn{1}{c|}{PoLO} & \cellcolor{blue!8}993.12/\textcolor{violet}{29.99} &  \multicolumn{1}{c|}{\cellcolor{blue!8}1783.50/\textcolor{violet}{28.83}} & \cellcolor{blue!8}1798.00/\textcolor{violet}{34.35} & \cellcolor{blue!8}8438.43/\textcolor{violet}{35.60} & \cellcolor{blue!8}25442.66/\textcolor{violet}{23.27} & \cellcolor{blue!8}6775.54/\textcolor{violet}{23.90} & \cellcolor{blue!8}1051.52/\textcolor{violet}{27.85} & \cellcolor{blue!8}2486.51/\textcolor{violet}{34.97} \\
 &  & \multicolumn{1}{c|}{PoL (Vanilla)} & 967.13/\textcolor{violet}{4.00} &  \multicolumn{1}{c|}{1756.94/\textcolor{violet}{2.27}} & 1768.85/\textcolor{violet}{5.20} & 8407.34/\textcolor{violet}{4.51} & 25425.03/\textcolor{violet}{5.64} & 6754.39/\textcolor{violet}{2..75} & 1026.55/\textcolor{violet}{2.88} & 2456.71/\textcolor{violet}{5.17}     \\ 
 &   & \multicolumn{1}{c|}{PoL (hash)} & 970.93/\textcolor{violet}{7.80} &  \multicolumn{1}{c|}{1758.14/\textcolor{violet}{3.47}} & 1770.05/\textcolor{violet}{6.40} &  \multicolumn{1}{c}{8410.94/\textcolor{violet}{8.11}} & 25432.23/\textcolor{violet}{12.84} & 6755.69/\textcolor{violet}{4.05} & 1027.95/\textcolor{violet}{4.28} & 2458.43/\textcolor{violet}{6.57} \\
 
\cline{2-2} 
 
& \multirow{3}{*}{$\mathbb{P}$ verify}   & \multicolumn{1}{c|}{PoLO}  & \cellcolor{yellow!15}193.79 & \cellcolor{yellow!15}213.85 & \cellcolor{yellow!15}188.58 & \cellcolor{yellow!15}495.44 & \cellcolor{yellow!15}445.18 & \cellcolor{yellow!15}167.15 & \cellcolor{yellow!15}76.69 & \cellcolor{yellow!15}81.88 \\
 &   & \multicolumn{1}{c|}{PoL (Vanilla)} & 768.41 & 1541.18 & 1580.35 & 7912.02 & 24975.89 & 6587.29 & 949.88 & 2374.87      \\ 
 &   & \multicolumn{1}{c|}{PoL (hash)} & 722.21 & 1542.38 & 1581.55 & \multicolumn{1}{c}{7915.62} & 24983.09 & 6588.59 & 951.28 & 2376.27   \\
                     
\bottomrule
\end{tabular}
\begin{tablenotes} % Added tablenotes environment
    \item $\bullet$ The baseline in this table refers to pure model training without applying any PoL or PoO techniques.
    \item $\bullet$ PoL (hash)~\cite{incentive_pol} follows a similar PoL process to that of PoL (Vanilla)~\cite{frist_pol,pol_attack1,pol_attack2}, with the key difference being the hash computation, which takes only a negligible amount of time.
    % --approximately xxx seconds per PoL snapshot.
    \item $\bullet$ PoL (zkp)~\cite{zkf_pol} incurs around $15$ mins of proof generation time per iteration on VGG-11 with CIFAR-10 (excluding training time); in contrast, \textsc{PoLO} completes proof generation for all iterations in $\sim 20$s on VGG-16 with TinyImageNet.
\end{tablenotes}
\end{threeparttable}
}
\vspace{-0.1in}
\end{table*}

\subsubsection{\underline{Compatibility}}
\textsc{PoLO} supports multiple watermarkings, including RIGA~\cite{riga}, EDNN~\cite{Uchida}, and FedIPR~\cite{li2022fedipr}. During training, watermark convergence directly tracks optimisation progress: driving the detection rate $\eta$ from near zero to $\eta_G^{\text{emb}} = 0.99$ within a shard requires many gradient-descent steps and cannot be shortcut by direct weight manipulation without corrupting task accuracy, making per-shard embedding cost a reliable proxy for genuine effort. We use the shard rate $|s|$ to measure computational efficiency. To evaluate compatibility, we instantiate \textsc{PoLO} with the three schemes above, each with watermark sizes of 2048 bits and 1024 bits, and report $|s|$, main task accuracy $Acc_{main}$, and total time for training plus $\mathbb{P}$ generation.
As shown in Tab.\ref{shard rate}, $|s|$ varies significantly, ranging from 32.67\% to 96.15\%. A clear inverse relationship is observed between watermark size and $|s|$, where a shorter watermark size results in a proportionally higher $|s|$ across all watermarking schemes. This trend is particularly evident in the TinyImageNet dataset, which consistently exhibits higher $|s|$ under all watermarking configurations. Notably, the $Acc_{main}$ values remain stable across different watermarking approaches, with minimal deviation from baseline model performance. 
%(cf. Fig.\ref{Fidelity Acc} in \textbf{Appendix}~\ref{appendix_fidelity})
Total training-plus-generation time also stays stable across watermarking schemes, confirming that the choice of method does not affect computational overhead.

\vspace{3pt}
\noindent\textbf{Shard rate and security trade-off.}
High shard rates (e.g., $|s|=96.15\%$ on TinyImageNet) indicate that watermark convergence is reached within roughly one epoch, which might suggest that per-shard forgery is cheap.
However, this conflates \emph{watermark convergence speed} with \emph{forgery difficulty}.
To succeed, a forgery adversary must both embed a new watermark ($\Lambda'$) and \emph{suppress} the existing watermark $\Lambda_x$ within one epoch. Our experiments show that $\Lambda_x$ persists with $\eta$ between 64\% and 98\% even after full-model fine-tuning across all evaluated configurations (Tab.\ref{table-attack}), making single-epoch suppression infeasible in practice. Stronger fine-tuning either degrades $Acc_{main}$ or raises costs above the honest per-shard training baseline.
Moreover, even if a single shard were cheaply forgeable in isolation, an attacker must forge \emph{all} preceding shards to produce a valid chain, so total forgery cost grows linearly with chain depth.

\vspace{3pt}
\begin{center}
\begin{tcolorbox}[
    colback=blue!5,
    colframe=blue!75!black,
    sharp corners,
    boxrule=0.25pt,
    width=\linewidth,
    enhanced jigsaw,
    % —— 内边距更紧 —— 
    left=2mm, right=2mm, top=0.5mm, bottom=0.5mm, boxsep=0.5mm,
    % —— 与上下正文的间距更小 —— 
    before skip=4pt, after skip=4pt,
    drop shadow=black!50
]
    \textbf{Takeaway (high compatibility).} \textsc{PoLO} supports a wide range of watermarkings and can maintain stable main task accuracy and  computational overhead, enabling seamless integration into different implementations without compromising performance.
\end{tcolorbox}
\end{center}

\begin{figure*}[!t]
    \centering
    \subfigure[AlexNet\_CIFAR10\_2048]{
    \begin{minipage}[t]{0.23\textwidth}
    \centering
    \includegraphics[width=1.5in]{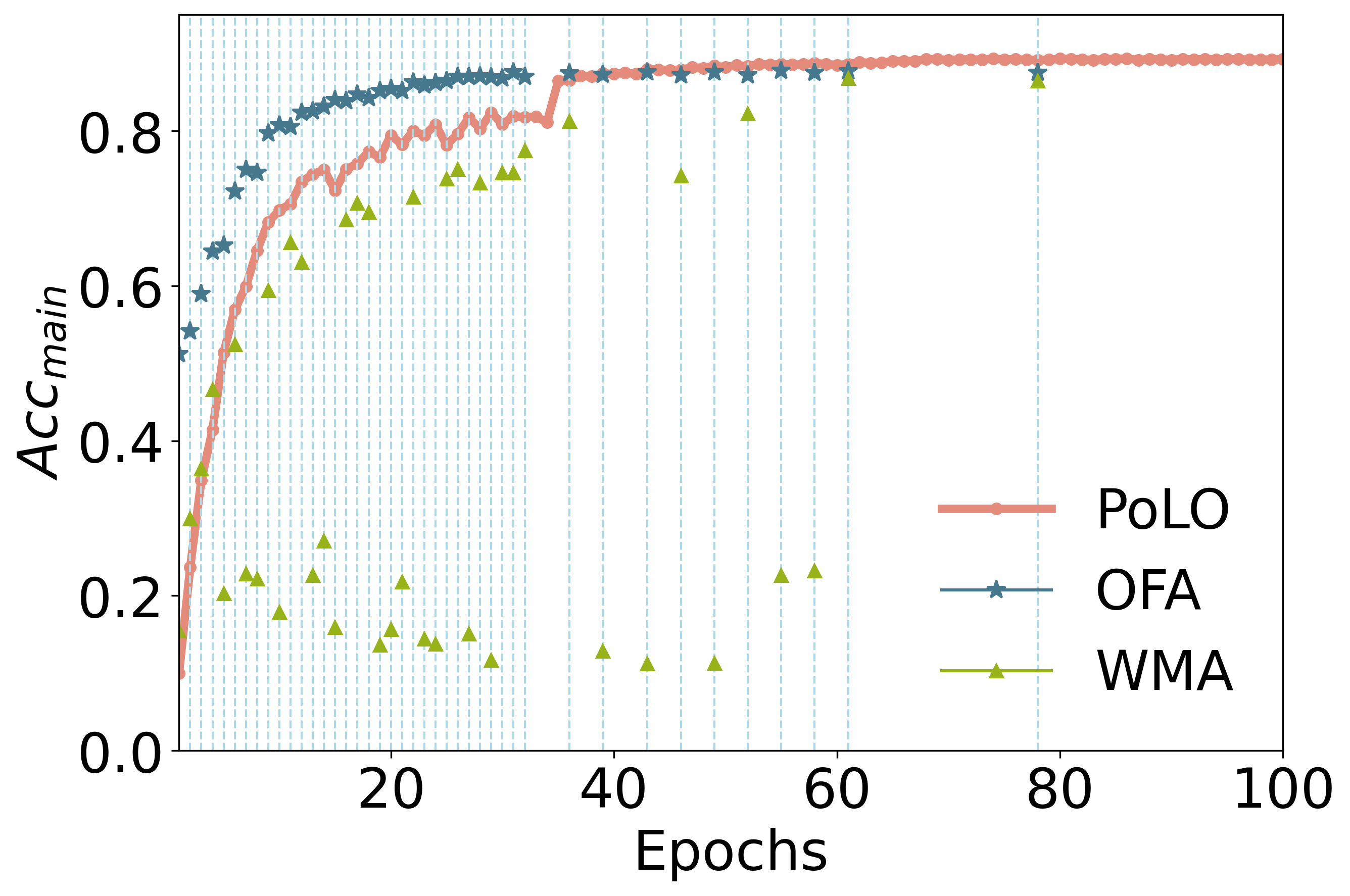}
    \end{minipage}
    \label{alexnet_cifar10_2048_attack}
    }
    % \hfill
    % \subfigure[AlexNet\_CIFAR10\_1024]{
    % \begin{minipage}[t]{0.23\textwidth}
    % \centering
    % \includegraphics[width=1.5in]{images/attacks/c10_alexnet_1024_attack.png}
    % \end{minipage}
    % \label{alexnet_cifar10_1024_attack}
    % }
    % \hfill
    \subfigure[ResNet18\_CIFAR10\_2048]{
    \begin{minipage}[t]{0.23\textwidth}
    \centering
    \includegraphics[width=1.5in]{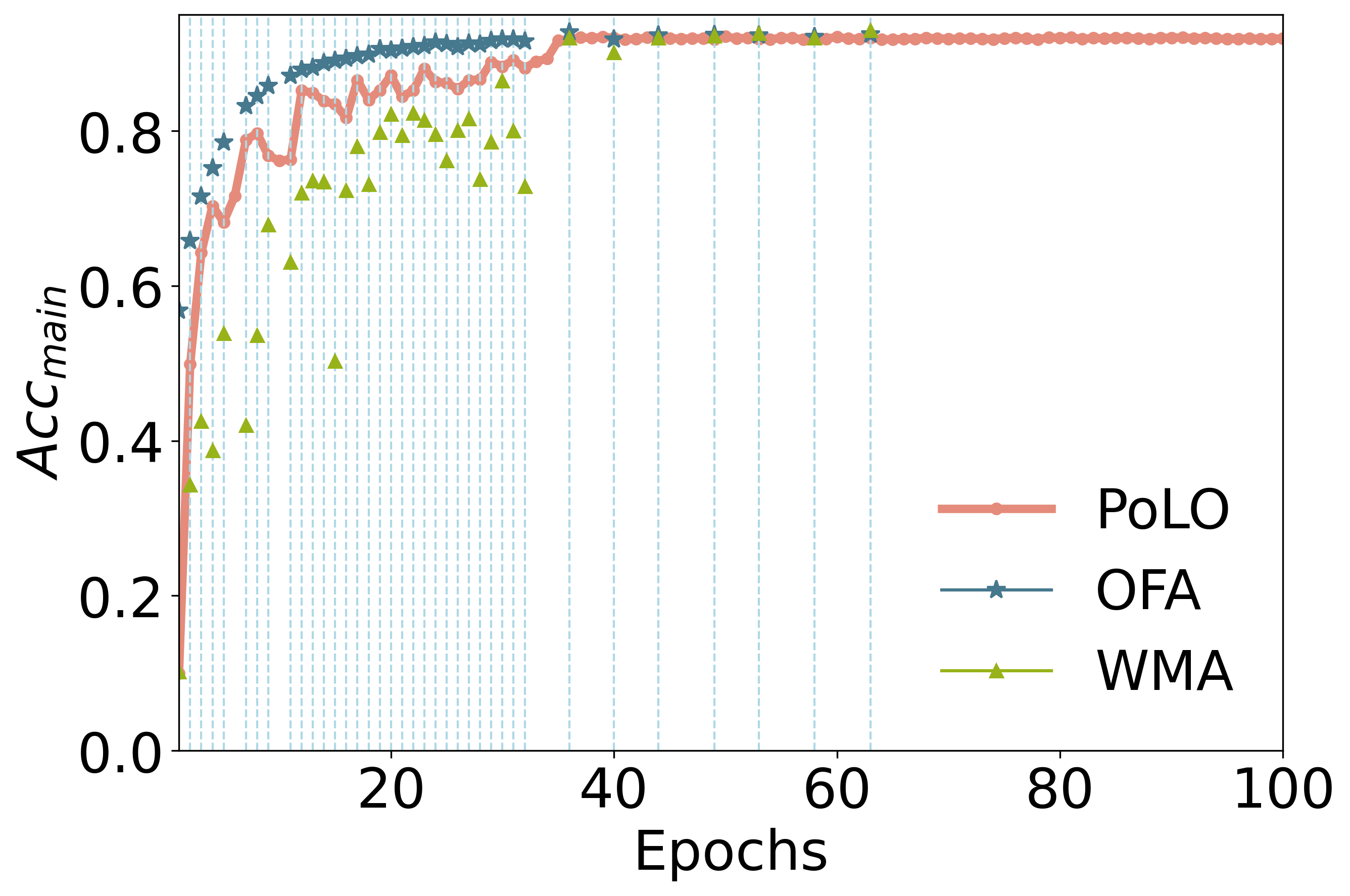}
    \end{minipage}
    \label{resnet18_cifar10_2048_attack}
    }
    % \hfill
    % \subfigure[ResNet18\_CIFAR10\_1024]{
    % \begin{minipage}[t]{0.23\textwidth}
    % \centering
    % \includegraphics[width=1.5in]{images/attacks/c10_resnet18_1024_attack.png}
    % \end{minipage}
    % \label{resnet18_cifar10_1024_attack}
    % }
    %%%%%%%%%%%%%%%%%%%%%%%%%%%%%%%%%%%%
    \subfigure[ResNet18\_CIFAR100\_2048]{
    \begin{minipage}[t]{0.23\textwidth}
    \centering
    \includegraphics[width=1.5in]{images/attacks/c10_alexnet_2048_attack.png}
    \end{minipage}
    \label{resnet18_cifar100_2048_attack}
    }
    % \hfill
    % \subfigure[ResNet18\_CIFAR100\_1024]{
    % \begin{minipage}[t]{0.23\textwidth}
    % \centering
    % \includegraphics[width=1.5in]{images/attacks/c10_alexnet_1024_attack.png}
    % \end{minipage}
    % \label{resnet18_cifar100_1024_attack}
    % }
    % \hfill
    \subfigure[WideResNet\_CIFAR100\_2048]{
    \begin{minipage}[t]{0.23\textwidth}
    \centering
    \includegraphics[width=1.5in]{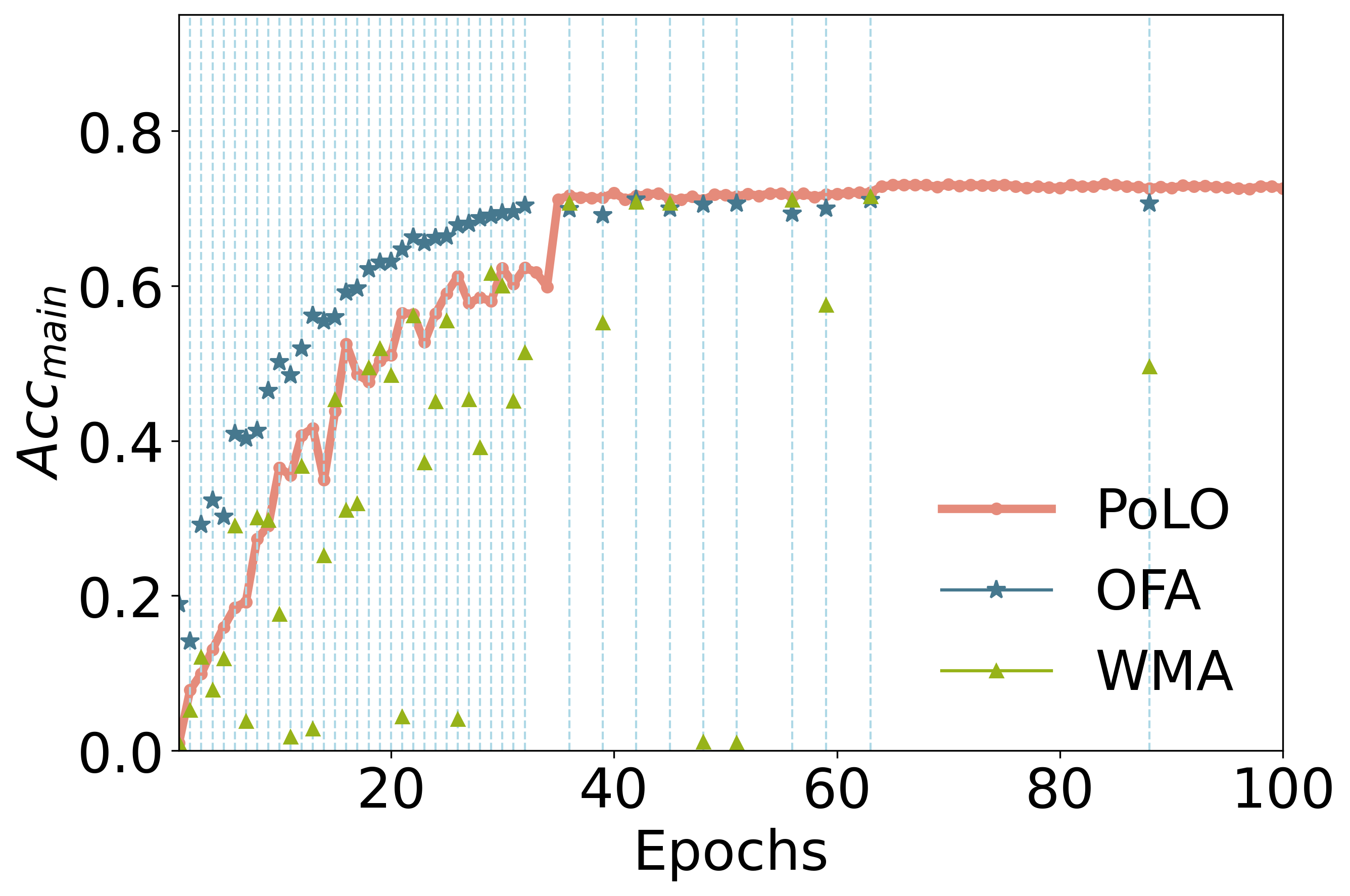}
    \end{minipage}
    \label{wideresnet_cifar100_2048_attack}
    }
    % \hfill
    % \subfigure[WideResNet\_CIFAR100\_1024]{
    % \begin{minipage}[t]{0.23\textwidth}
    % \centering
    % \includegraphics[width=1.5in]{images/attacks/c100_wideresnet_1024_attack.png}
    % \end{minipage}
    % \label{wideresnet_cifar100_1024_attack}
    % }
    %%%%%%%%%%%%%%%%%%%%%%%%%%%%%%%%%%%%%%%
    
    \subfigure[VGG16\_TinyImageNet\_2048]{
    \begin{minipage}[t]{0.23\textwidth}
    \centering
    \includegraphics[width=1.5in]{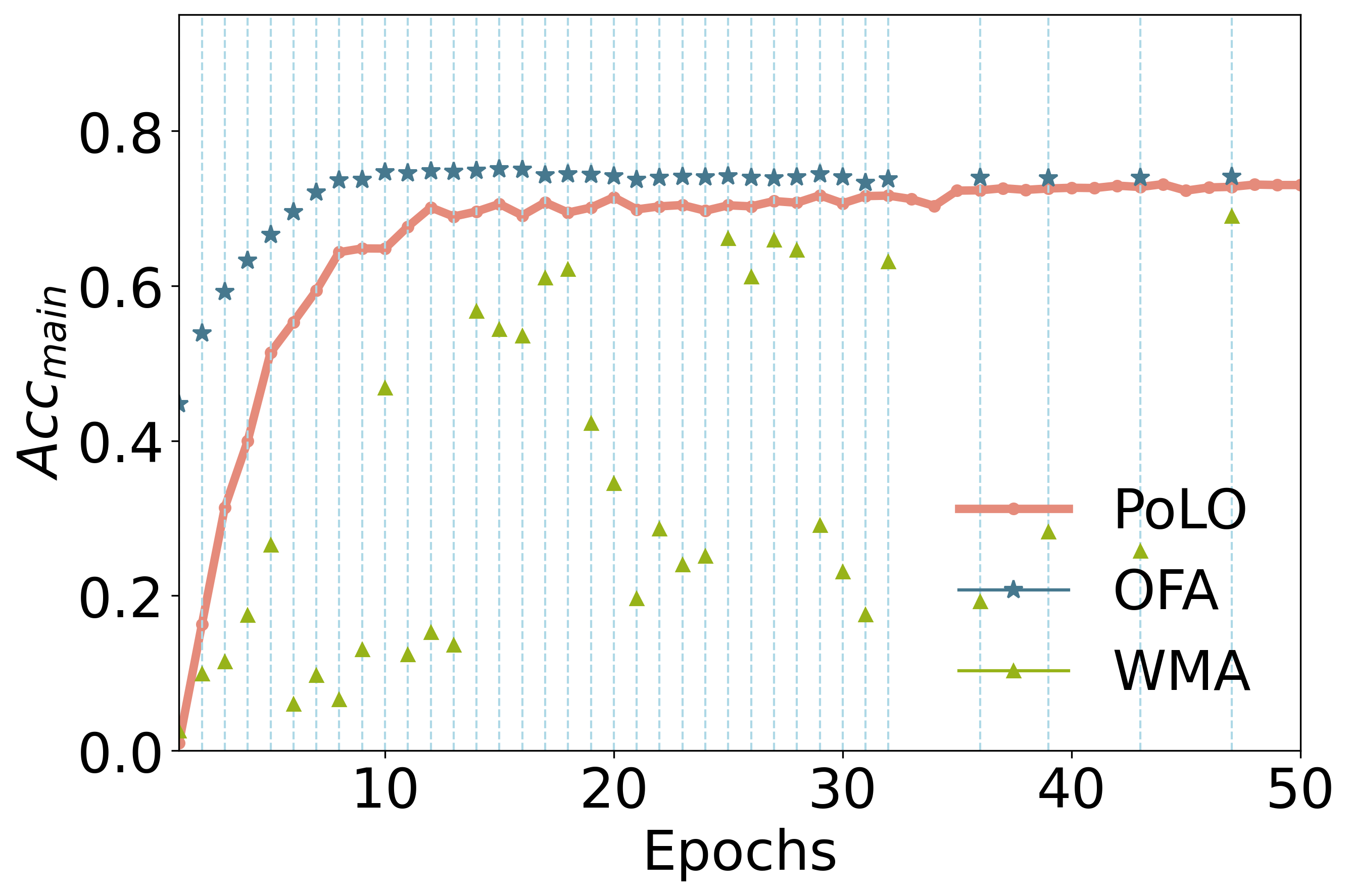}
    \end{minipage}
    \label{vgg16_image_2048_attack}
    }
    % \hfill
    % \subfigure[VGG16\_TinyImageNet\_1024]{
    % \begin{minipage}[t]{0.23\textwidth}
    % \centering
    % \includegraphics[width=1.5in]{images/attacks/image_vgg16_1024_attack.png}
    % \end{minipage}
    % \label{vgg16_image_1024_attack}
    % }
    % \hfill
    \subfigure[ResNet34\_TinyImageNet\_2048]{
    \begin{minipage}[t]{0.23\textwidth}
    \centering
    \includegraphics[width=1.5in]{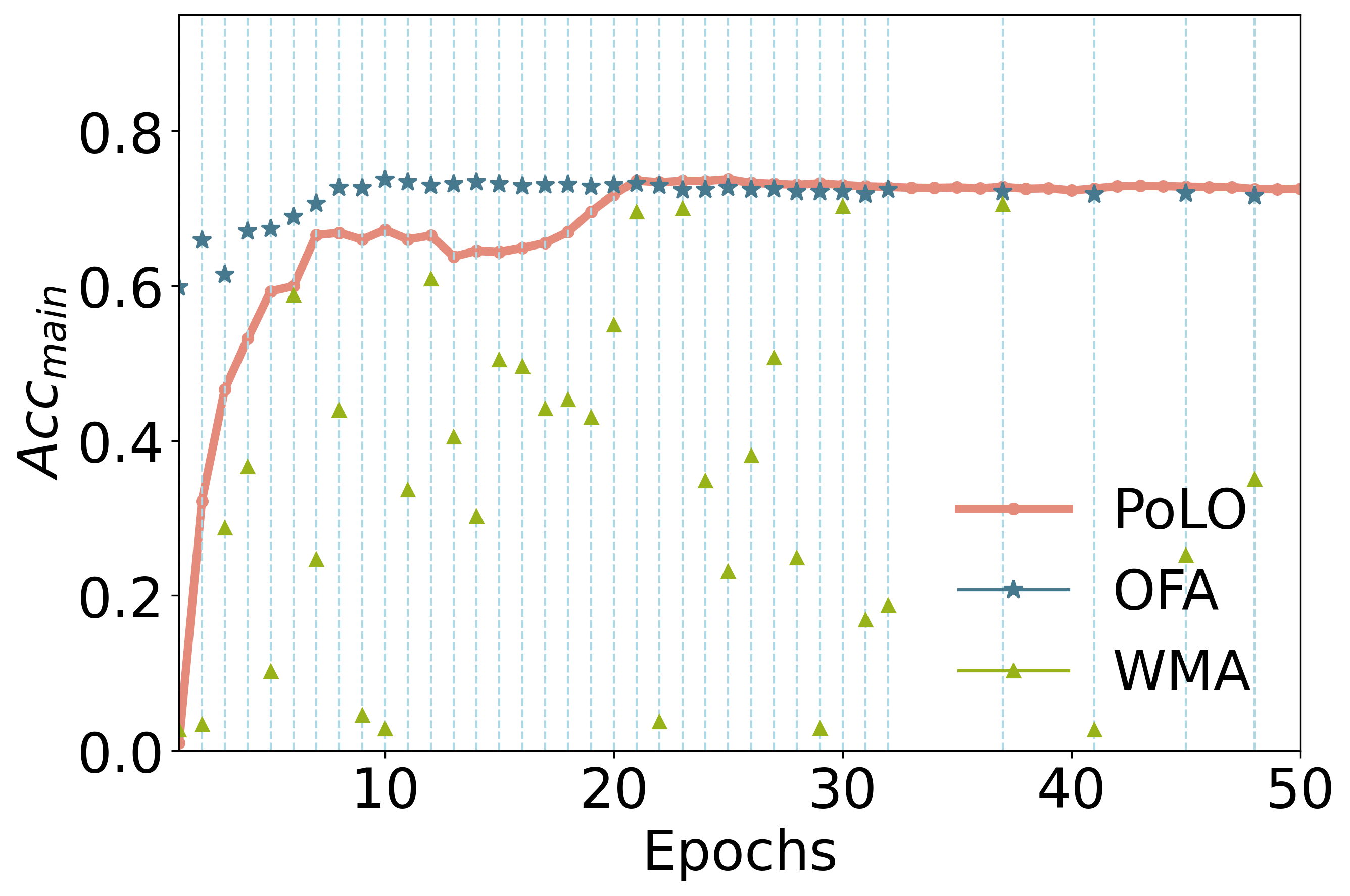}
    \end{minipage}
    \label{resnet_image_2048_attack}
    }
    % \hfill
    % \subfigure[ResNet34\_TinyImageNet\_1024]{
    % \begin{minipage}[t]{0.23\textwidth}
    % \centering
    % \includegraphics[width=1.5in]{images/attacks/image_resnet34_1024_attack.png}
    % \end{minipage}
    % \label{resnet_image_1024_attack}
    % }
    %%%%%%%%%%%%%%%%%%%%%%%%%%%%%%%%%%%%%%%
    \subfigure[TextCNN\_AG News\_2048]{
    \begin{minipage}[t]{0.23\textwidth}
    \centering
    \includegraphics[width=1.5in]{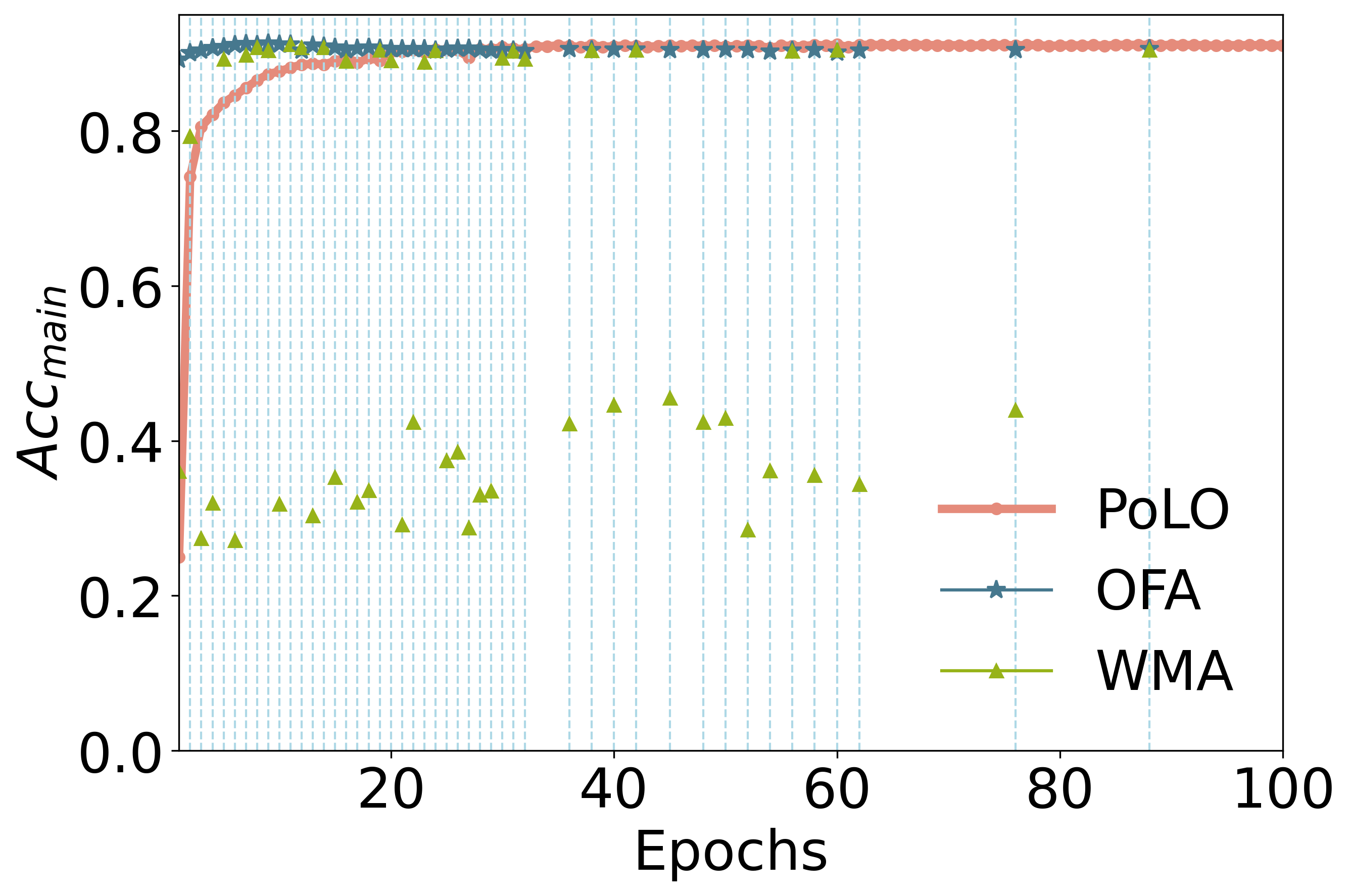}
    \end{minipage}
    \label{textcnn_news_2048_attack}
    }
    % \hfill
    % \subfigure[TextCNN\_AG News\_1024]{
    % \begin{minipage}[t]{0.23\textwidth}
    % \centering
    % \includegraphics[width=1.5in]{images/attacks/news_textcnn_1024_attack.png}
    % \end{minipage}
    % \label{textcnn_news_1024_attack}
    % }
    % \hfill
    \subfigure[MiniBert\_AG News\_2048]{
    \begin{minipage}[t]{0.23\textwidth}
    \centering
    \includegraphics[width=1.5in]{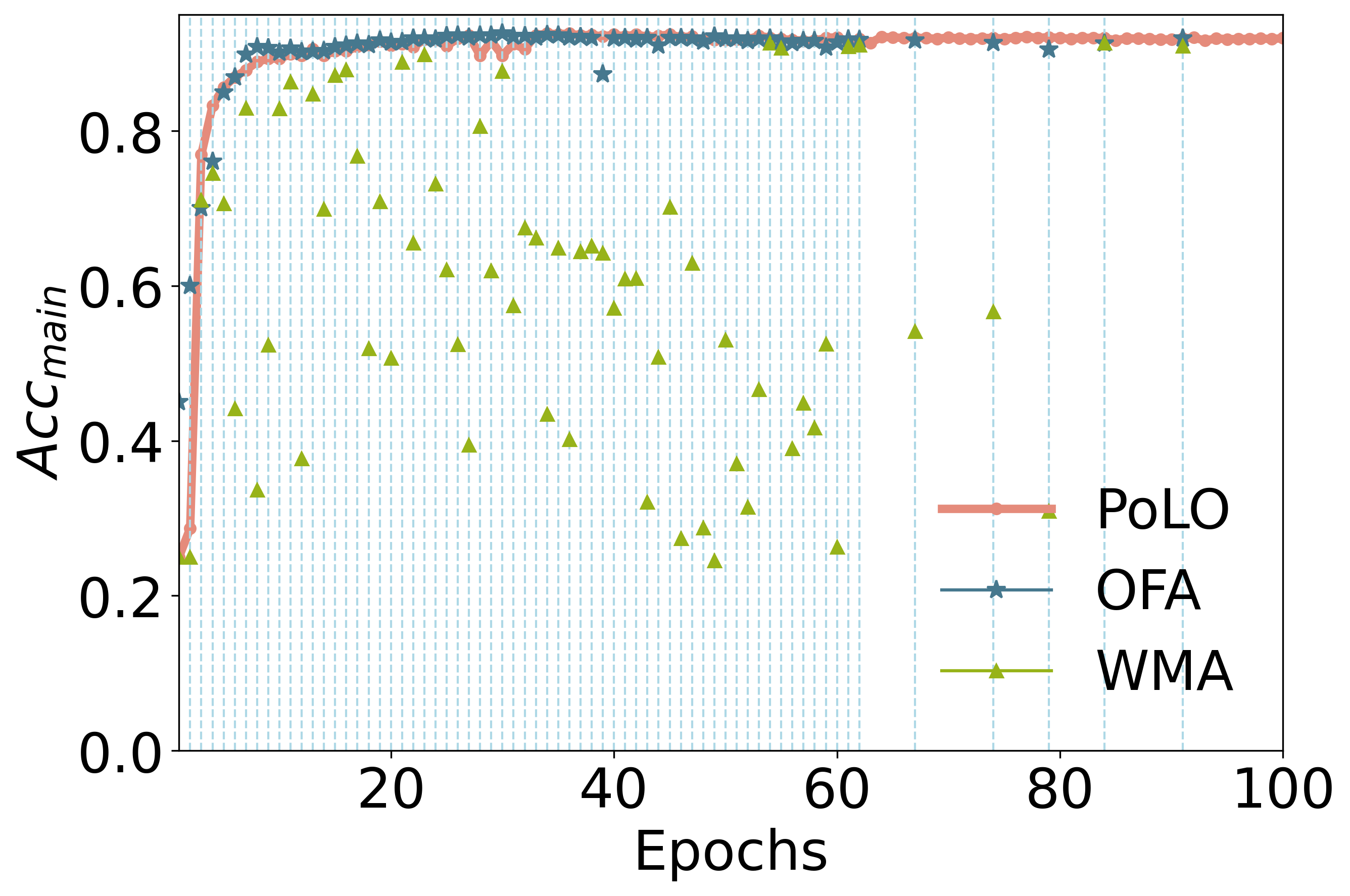}
    \end{minipage}
    \label{bert_news_2048_attack}
    }
    % \hfill
    % \subfigure[MiniBert\_AG News\_1024]{
    % \begin{minipage}[t]{0.23\textwidth}
    % \centering
    % \includegraphics[width=1.5in]{images/attacks/news_bert_1024_attack.png}
    % \end{minipage}
    % \label{bert_news_1024_attack}
    % }

\caption{
The $Acc_{main}$ comparison of the two $\mathbb{P}$ forge attacks with baseline - \textsc{PoLO}. The number followed by the dataset is the watermark size. In each subplot, the vertical lines perpendicular to the x-axis represent the completion times for each shard.
}
\label{attack}
\end{figure*}

\subsubsection{\underline{Overhead}}

A critical consideration in our evaluation is ensuring that $\mathbb{P}$ generation does not impose substantial computational overhead beyond baseline model training costs. 
Furthermore, for practical implementation, the verification overhead must be significantly lower than $\mathbb{P}$ generation costs.
To validate the efficiency of our \textsc{PoLO} framework, we conduct comparative analyses against traditional PoL schemes~\cite{frist_pol,pol_attack1,pol_attack2,incentive_pol,zkf_pol} and baseline training approaches.
Our experimental setup incorporates three distinct watermark sizes for PoLO generation.
% with baseline model training and traditional PoL generation repeated three times for statistical validity. 
To maintain experimental rigor and ensure fair comparison, we standardize computational resources and parameter settings across all trials. 
The verification process is designed to validate each intermediate $\mathbb{P}$, providing a comprehensive assessment of the framework's efficiency.

Tab.\ref{tab_gen-ver} compares \textsc{PoLO} with existing PoLs in proof generation and verification time. Because \textsc{PoLO} embeds watermarking directly into training, proof generation adds only modest overhead (e.g., 2.09\% on CIFAR-10/AlexNet with a 2048-bit watermark), whereas traditional PoLs incur generation costs comparable to a full training run. The gap is wider during verification: \textsc{PoLO}'s verification-to-generation ratio ranges from 1.3\% (VGG16, TinyImageNet) to under 12\% across all models, while conventional PoL verification reaches 91-98\% of generation time as it must replay intermediate states.

\vspace{3pt}
\begin{center}

\begin{tcolorbox}[
    colback=blue!5, % Light blue background
    colframe=blue!75!black, % Blue border with a slight black tint
    sharp corners,
    boxrule=0.25pt, % Border thickness
    width=\linewidth,
    enhanced jigsaw,
    % —— 内边距更紧 —— 
    left=2mm, right=2mm, top=0.5mm, bottom=0.5mm, boxsep=0.5mm,
    % —— 与上下正文的间距更小 —— 
    before skip=4pt, after skip=4pt,
    drop shadow=black!50
]
    % \centering
\textbf{Takeaway (faster, privacy-preserving verification).} \textsc{PoLO} reduces verification time to 1.3–1.8\% of generation time, versus up to 98\% in conventional PoLs, by avoiding retraining of intermediate models.
\end{tcolorbox}
\end{center}

\subsubsection{\underline{Unforgeability}}

To assess \textsc{PoLO}’s resistance to forgery, we evaluate two attack strategies: OFA and WMA, under identical computational conditions. 
%All attacks are run under identical computational conditions. 
The legitimate model embeds a 2048-bit watermark, while attackers attempt to insert 2048 or 1024-bit illicit ones. Effectiveness is evaluated via computational overhead, post-attack detection rate $\eta$ of the legitimate watermark, and main task accuracy $Acc_{main}$. In OFA, attackers reuse RIGA-style embedding and fine-tune all model weights with a reduced learning rate (10\% of the original) %~\cite{adi2018turning} 
to inject their watermark while degrading the original. In WMA, they directly overwrite selected weights using values derived from their illicit watermark and key.

Fig.\ref{attack} shows the impact of \textsc{PoLO} forgery attacks across different shard stages.
%while Fig.\ref{attack_more} in \textbf{Appendix}~\ref{appendix_non-forgeability} provides additional results for illicit watermark sizes of 1024 and 512 bits. 
The watermark embedding speed (shard changing rate) is sensitive to the learning rate. Fig.\ref{attack} shows that the shard indices lines (vertical blue lines) vary in density every 30 epochs due to the scheduled learning rate decay by a factor of $10$ %~\cite{wu2023faster}.
Models under WMA exhibit noticeable instability, with $Acc_{main}$ fluctuating significantly during the attack. This degradation causes the forged proof $\hat{\mathbb{P}}$ to fail verification.
In contrast, OFA maintains a comparable $Acc_{main}$ to that of the legitimate \textsc{PoLO}-protected model.

However, 
Tab.\ref{table-attack} shows that forging \textsc{PoLO} is economically impractical under rational assumptions. OFA (full-model fine-tuning to erase and re-embed) costs 113{,}771.31s (\$20.22) on VGG16 with TinyImageNet, compared with 25{,}051.72s (\$4.43) for honest proof generation, roughly a 5$\times$ increase. Despite this cost, the legitimate watermark survives with $\eta$ between 64.53\% and 98.43\%. WMA is much cheaper (11.32–344.94s, \$0.01–\$0.06) but likewise ineffective: legitimate watermark detection remains above 91.48\% (up to 99.81\%), and Fig.\ref{attack} shows that WMA significantly degrades task accuracy, causing models to fail \textsc{PoLO} verification.

% {\color{black}These results expose fundamental limits of both attacks at the per-shard level. OFA uses full fine-tuning, yet the legitimate watermark remains clearly detectable, so the attacker fails to erase it and the forged model cannot be monetised in EMMs. Achieving a clean overwrite via partial fine-tuning would be even harder for erasing. Consequently, OFA’s per-shard forgery cost quickly matches or exceeds honest shard training, breaking economic rationality. WMA avoids most training by directly editing weights but cannot reliably target the owner’s watermark and often degrades accuracy enough to fail PoLO verification. Across all settings, attackers either pay per-shard costs comparable to honest training or end up with unverifiable models, making forgery both economically and functionally unviable in EMMs.}

\begin{table*}[!ht]
\centering
\caption{The time consumption, legitimate watermark detection rate $\eta$, and $Acc_{main}$ of launching attacks under different watermark sizes.}
\label{table-attack}
\renewcommand{\arraystretch}{1}
\resizebox{\textwidth}{!}{
\begin{threeparttable}
\begin{tabular}{cc|c cc  cc | cc cc}
% \hline
\toprule
\makecell{Watermark \\ \textbf{size}(bits)} & \multicolumn{2}{c}{\textbf{\makecell{PoLO \\ forgery attacks}}} & \makecell{CIFAR-10 \\ AlexNet} & \makecell{CIFAR-10 \\ ResNet18} & \makecell{CIFAR-100 \\ ResNet18} & \multicolumn{1}{c}{\makecell{CIFAR-100 \\ WideResNet}} & \makecell{TinyInageNet \\ VGG16} & \makecell{TinyInageNet \\ ResNet34} & \makecell{AG News \\ TextCNN} & \makecell{AG News \\ MiniBert} \\ 
% \hline
\toprule
                      \multirow{9}{*}{2048} & \multirow{3}{*}{Time(s)}& \multicolumn{1}{c|}{{PoLO}} &  \cellcolor{blue!8} 990.74 &  \cellcolor{blue!8} 1618.78 & \cellcolor{blue!8}  1707.59 & \multicolumn{1}{c}{ \cellcolor{blue!8} 9357.91 } &  \cellcolor{blue!8} 25051.72 & \cellcolor{blue!8}  7332.32 & \cellcolor{blue!8}  1070.12 &  \cellcolor{blue!8} 2573.38 \\
                      & & \multicolumn{1}{c|}{launch OFA} & 1696.91 & 2723.17 & 2381.83 & \multicolumn{1}{c}{15498.77} & 113771.31 & 30420.92 & 2275.26 & 7954.11    \\
                      & & \multicolumn{1}{c|}{launch WMA} & 56.51 & 62.04 & 64.37 & \multicolumn{1}{c}{218.91} & 344.94 & 158.55 & 11.33 & 30.46 \\ 
                      
                       \cline{2-3}
                      & \multirow{3}{*}{\makecell{$\eta$(\%)}} & \multicolumn{1}{c|}{\makecell{{PoLO}}} &  \cellcolor{blue!8} 99.24 &  \cellcolor{blue!8} 99.02 &  \multicolumn{1}{c}{\cellcolor{blue!8}99.07} &  \multicolumn{1}{c}{\cellcolor{blue!8}99.07} &  \cellcolor{blue!8} 99.46 &  \multicolumn{1}{c|}{\cellcolor{blue!8}99.71} &  \cellcolor{blue!8} 99.21 & \cellcolor{blue!8} 99.41 \\
                      & & \multicolumn{1}{c|}{launch OFA} & 93.62 & 77.82 & 77.52 & \multicolumn{1}{c}{96.83 }& 86.37 & \multicolumn{1}{c|}{75.27} & 69.84 & 78.09     \\
                      & & \multicolumn{1}{c|}{launch WMA} & 98.31 & 94.17 & 94.75 & \multicolumn{1}{c}{99.78} & 95.78 & \multicolumn{1}{c|}{95.75} & 91.48 & 96.61 \\ 
                      
                       \cline{2-3} 
                       
                      & \multirow{3}{*}{\makecell{$Acc_{main}$(\%)}} & \multicolumn{1}{c|}{{PoLO}} &  \cellcolor{blue!8} 89.22 &  \cellcolor{blue!8} 91.89 &  \cellcolor{blue!8} 74.81 &  \multicolumn{1}{c}{\cellcolor{blue!8}72.83} &  \cellcolor{blue!8} 73.00 &  \multicolumn{1}{c|}{\cellcolor{blue!8}72.46} &  \cellcolor{blue!8} 91.04 &  \cellcolor{blue!8} 91.88 \\
                      & & \multicolumn{1}{c|}{launch OFA} & 87.47 & 91.72 & 70.00 & \multicolumn{1}{c}{69.84} & 73.91 & \multicolumn{1}{c|}{71.94} & 90.37 & 91.45 \\
                      & & \multicolumn{1}{c|}{launch WMA} & 47.52 & 80.37 & 36.84 & \multicolumn{1}{c}{46.56} & 43.25 & \multicolumn{1}{c|}{39.81} & 58.69 & 55.20 \\
                      
            \cline{1-3}
                      
                      \multirow{9}{*}{1024} & \multirow{3}{*}{Time(s)} & \multicolumn{1}{c|}{{PoLO}} &  \cellcolor{blue!8} 1066.78 &  \cellcolor{blue!8} 1861.62 &  \cellcolor{blue!8} 1681.42 &  \multicolumn{1}{c}{\cellcolor{blue!8} 8968.8} &  \cellcolor{blue!8} 24934.03 &  \multicolumn{1}{c|}{\cellcolor{blue!8}7238.95} &  \cellcolor{blue!8} 1067.23 &  \cellcolor{blue!8} 2687.89 \\
                      & & \multicolumn{1}{c|}{launch OFA} & 1136.86 & 2106.95 & 1802.52 & \multicolumn{1}{c}{9355.78} & 61437.63& \multicolumn{1}{c|}{17808.86} & 1355.37 & 4545.51 \\
                      & & \multicolumn{1}{c|}{launch WMA} & 52.21 & 62.98 & 63.66 & \multicolumn{1}{c}{208.04} & 343.46 & \multicolumn{1}{c|}{148.24} & 11.32 & 30.33 \\ 
                      
                      \cline{2-3} 
                      
                      & \multirow{3}{*}{\makecell{$\eta$(\%)}} & \multicolumn{1}{c|}{{PoLO}} & \cellcolor{blue!8}  99.02 & \cellcolor{blue!8} 99.12 & \cellcolor{blue!8} 99.51 & \cellcolor{blue!8} 92.21 &\cellcolor{blue!8}  100.00 & \cellcolor{blue!8} 100.00 & \cellcolor{blue!8} 99.22 & \cellcolor{blue!8} 99.32 \\
                      & & \multicolumn{1}{c|}{launch OFA} & 96.93 & 82.28 & 75.91 & 98.43 & 92.23 & 84.11 & 73.52 & 64.53 \\
                      & & \multicolumn{1}{c|}{launch WMA} & 97.01 & 94.34 & 95.41 & 99.81 & 95.08 & 94.21 & 92.69 & 96.11 \\ 
                      
                      \cline{2-3} 

                      & \multirow{3}{*}{\makecell{$Acc_{main}$(\%)}} & \multicolumn{1}{c|}{{PoLO}} &  \cellcolor{blue!8} 88.94 &  \cellcolor{blue!8} 91.96 &  \cellcolor{blue!8} 74.93 &  \cellcolor{blue!8} 72.96 &  \cellcolor{blue!8} 73.33 &  \cellcolor{blue!8} 73.31 &  \cellcolor{blue!8} 90.94 &  \cellcolor{blue!8} 92.06 \\
                      & & \multicolumn{1}{c|}{launch OFA} & 87.45 & 91.68 & 69.15 & 69.38 & 73.65 & 72.00 & 90.36 & 91.62 \\
                      & & \multicolumn{1}{c|}{launch WMA} & 39.07 & 79.30 & 34.35 & 40.52 & 45.31 & 31.09 & 61.19        & 58.96 \\
                      
             \cline{1-3}
                       
                      \multirow{9}{*}{512}  & \multirow{3}{*}{Time(s)} & \multicolumn{1}{c|}{{PoLO}} & \cellcolor{blue!8} 993.12 & \cellcolor{blue!8} 1783.50 & \cellcolor{blue!8} 1798.00 & \cellcolor{blue!8} 8433.43 & \cellcolor{blue!8} 25442.66 &\cellcolor{blue!8}  6775.54 & \cellcolor{blue!8} 1051.52 & \cellcolor{blue!8} 2486.51 \\
                      & & \multicolumn{1}{c|}{launch OFA} & 1077.32 & 1839.35 & 1967.88 & 9369.91 & 61578.36 & 18359.39 & 1355.33 & 4575.48 \\
                      & & \multicolumn{1}{c|}{launch WMA} & 52.71 & 66.92 & 63.92 & 206.76 & 341.63 & 144.61 & 11.39 & 30.42 \\ 
                      
                     \cline{2-3} 
                     
                      & \multirow{3}{*}{\makecell{$\eta$(\%)}} & \multicolumn{1}{c|}{{PoLO}} & \cellcolor{blue!8}  99.20 &\cellcolor{blue!8} 99.22 & \cellcolor{blue!8} 99.21 & \multicolumn{1}{c}{\cellcolor{blue!8} 99.41} & \cellcolor{blue!8} 100.00 & \cellcolor{blue!8} 100.00 & \cellcolor{blue!8} 100.00 & \cellcolor{blue!8} 99.81 \\
                      & & \multicolumn{1}{c|}{launch OFA} & 96.98 & 75.03 & 75.99 & \multicolumn{1}{c}{98.47} & 84.24 & 63.52 & 79.67 & 64.27 \\
                      & & \multicolumn{1}{c|}{launch WMA} & 96.35 & 93.25 & 92.84 & \multicolumn{1}{c}{99.79} & 96.08 & 93.95 & 91.15 & 93.47 \\ 
                      
                      \cline{2-3} 
                      
                      & \multirow{3}{*}{\makecell{$Acc_{main}$(\%)}} & \multicolumn{1}{c|}{{PoLO}} & \cellcolor{blue!8} 89.54 & \cellcolor{blue!8}  92.22 & \cellcolor{blue!8}  75.18 & \multicolumn{1}{c}{\cellcolor{blue!8}72.64} & \cellcolor{blue!8}  73.91 & \cellcolor{blue!8}  73.22 & \cellcolor{blue!8}  90.65 & \cellcolor{blue!8}  91.82 \\
                      & & \multicolumn{1}{c|}{launch OFA} & 87.49 & 91.80 & 69.61 & \multicolumn{1}{c}{69.41} & 73.76 & 72.17 & 90.38 & 91.60 \\
                      & & \multicolumn{1}{c|}{launch WMA} & 40.50 & 79.51 & 38.88 & \multicolumn{1}{c}{40.86} & 37.97 & 38.41 & 65.39 & 65.19 \\
                      \bottomrule
\end{tabular}
\begin{tablenotes} % Added tablenotes environment
    \item $\bullet$ For launching attacks, a smaller $\eta$ indicates greater attack effectiveness against the legitimate watermark, while a higher $Acc_{main}$ reflects better fidelity to the main task performance (i.e., less impact on accuracy).
\end{tablenotes}
\end{threeparttable}
}
\end{table*}

Both attack families face per-shard limitations. OFA cannot reliably erase the legitimate watermark even with full-model fine-tuning; pushing harder raises the per-shard cost to match or exceed honest training, violating economic rationality. WMA is cheaper but imprecise: it cannot target the owner's watermark positions and typically degrades $Acc_{main}$ enough to fail verification. In either case, the attacker pays costs comparable to honest training or produces an unverifiable model. A shard-wise visualization in \textbf{Appendix}~\ref{attack-curve} provides additional empirical evidence for the unforgeability of \textsc{PoLO} against economically rational attackers.

\smallskip
\noindent\textbf{Attacker verification outcome.}
To clarify whether forged proofs actually pass \textsc{PoLO} verification, we report the attacker's forged watermark detection rate $\eta'$ and the resulting verification outcome. Under OFA, the attacker's forged watermark achieves high $\eta'$ (typically $>$95\%) because the attacker controls the fine-tuning objective; however, the \emph{legitimate} watermark simultaneously persists with $\eta$ between 64\% and 98\%, meaning the EMM verifier can still attribute the model to the original owner. The attacker thus fails to \emph{exclusively} claim ownership. Under WMA, although the attacker can embed a forged watermark in selected weights, the main-task accuracy $Acc_{main}$ drops below the verification threshold in most configurations, causing the forged proof to be rejected outright. In neither case does the attacker succeed in both passing verification and suppressing the legitimate proof.

\smallskip
\noindent\textbf{Watermark detection threshold.}
Throughout our experiments, we set the verification threshold $\eta_G^{\text{ver}} = 0.70$, distinct from the embedding completion threshold $\eta_G^{\text{emb}} = 0.99$ used to determine when a shard is finished (Algorithm~1; cf.\ \textbf{Appendix}~\ref{appendix_algo}). This value is empirically grounded: legitimate watermarks consistently achieve $\eta > 0.95$ after embedding (Tab.\ref{shard rate}), while the lowest surviving detection rate under OFA is $\eta = 64.27\%$ (MiniBert, 512-bit, Tab.\ref{table-attack}). Setting $\eta_G^{\text{ver}} = 0.70$ therefore remains well below the legitimate post-embedding rate yet above the worst-case attacked rate, maximising the false-rejection-free operating range. Tab.\ref{tab:merged_lr_eta} in \textbf{Appendix}~\ref{appendix_config_shard} further confirms that varying $\eta_G^{\text{emb}}$ across 0.85, 0.95, and 0.99 does not affect $Acc_{\text{main}}$, indicating stable shard dynamics. A full parametric sweep of $\eta_G^{\text{ver}}$ is left to future work.

\begin{center}
\begin{tcolorbox}[
colback=blue!5,
colframe=blue!75!black,
sharp corners,
boxrule=0.3pt,
width=\linewidth,
enhanced jigsaw,
% —— 内边距更紧 —— 
left=2mm, right=2mm, top=0.5mm, bottom=0.5mm, boxsep=0.5mm,
% —— 与上下正文的间距更小 —— 
before skip=4pt, after skip=2pt,
drop shadow=black!50
]
\textbf{Takeaway (robust against forgery with practical cost barriers).} Both OFA and WMA fail to compromise the legitimacy of \textsc{PoLO}: OFA incurs up to 5× the cost of proof generation while still yielding detection rates between 64.27\% and 98.47\%, and WMA, though cheaper, remains ineffective with detection consistently above 90\%.
\end{tcolorbox}
\end{center}

\vspace{0.1in}

\section{Concluding Remarks}
\label{sec_con}

% We introduced \textsc{PoLO}, a unified framework that integrates PoL and PoO via chained watermarking to overcome limitations of prior approaches. Addressing \textbf{\textcolor{violet}{RQ1}}, we embed evolving watermarks throughout training to link ownership to the entire learning process, not just the final model. For \textbf{\textcolor{violet}{RQ2}}, we replace gradient-based proofs with cryptographic hash chains, making verification tamper-resistant and robust against any forgery attacks. Unlike conventional methods that expose data or require recomputation, \textsc{PoLO} enables efficient, privacy-preserving verification. Experiments show that \textsc{PoLO} achieves \textbf{99\%} watermark detection accuracy, reduces verification cost to just \textbf{1.5 to 10\%} of conventional PoL methods, and requires \textbf{1.1 to 4 times} more effort to forge than to generate legitimately, while still maintaining over \textbf{90\%} detection accuracy after attacks. These results confirm \textsc{PoLO} as a secure, efficient, and privacy-preserving solution for model verification.

We introduced \textsc{PoLO}, a unified framework that integrates PoL and PoO via chained watermarking in EMMs. For \textbf{\textcolor{violet}{RQ1}}, we turn PoO into PoL by embedding evolving watermarks throughout training, enabling verifiable effort tracking over the full trajectory. For \textbf{\textcolor{violet}{RQ2}}, we replace gradient-based proofs with cryptographic hash chains, improving tamper resistance and forgery robustness. In economically driven settings of EMMs where ownership gates monetization, \textsc{PoLO} raises an economic barrier: forgery requires effort comparable to honest training or yields unverifiable models. Unlike prior PoL schemes that expose data or require expensive recomputation, \textsc{PoLO} supports efficient, privacy-aware verification. Experiments show \textsc{PoLO} achieves \textbf{99\%} watermark detection accuracy, reduces verification cost to \textbf{1.5\%–10\%} of training, and demands \textbf{1.1x–4x} more effort to forge than to generate legitimately, while retaining over \textbf{90\%} accuracy under attack.

% \smallskip
% \noindent\textbf{Limitations.}
% \textsc{PoLO}'s security guarantees are scoped to EMMs that enforce watermark-based verification. Redistribution through non-EMM channels---such as public model hubs, private sharing, API-only serving, or model distillation/extraction---falls outside this economic model and is not claimed as a protection target. The economic rationality assumption further restricts the threat model: adversaries with unlimited budgets or purely destructive intent (e.g., availability attacks) are excluded. Finally, our evaluation covers two representative forgery families (OFA and WMA); stronger adaptive attacks that jointly optimise watermark suppression, re-embedding, and accuracy preservation across multiple shards deserve further investigation.

%================================================
%=================================================

\bibliographystyle{unsrt}
\bibliography{bib(short)}

\renewcommand{\thesubsection}{\arabic{subsection}}

\section*{Appendix}
%%%%%%%%%%%%%%%%%%%%%%%%%%%%%%%%%%%%%%%%%%%%%%%%%%%%%%
\subsection{Notation (Tab.\ref{tab:notation})}
\label{appendix_notation}

%Tab.\ref{tab:notation} summaries the notations used throughout the paper.

\begin{table}[!hbt]
    \renewcommand\arraystretch{1}
    \small
    \caption{List of Notation}
    \label{tab:notation}
    \begin{center}
    \begin{threeparttable}
    \resizebox{\linewidth}{!}{
    \begin{tabular}{c|l|}
    \toprule
    \multicolumn{1}{c}{\textbf{Notation}} &  \multicolumn{1}{c}{\textbf{Definition}} \\
    \midrule
    \multicolumn{1}{c}{\cellcolor{yellow!15}$\mathbb{P}$} &\multicolumn{1}{l}{ Proof-of-Anything (PoX) proof }\\
    %\cellcolor{yellow!15}$\mathbb{P}_o$ & Proof-of-Ownership (PoO) proof \\
   \multicolumn{1}{c}{\cellcolor{yellow!15}$\mathcal{P}$/$\mathcal{V}$/$\mathcal{A}$} & \multicolumn{1}{l}{Prover/Verifier/Attacker} \\
     \multicolumn{1}{c}{\cellcolor{yellow!15}$\Theta$} & \multicolumn{1}{l}{Irreversible function with input $\chi$} \\
    \cellcolor{yellow!15}$\Psi$ & \multicolumn{1}{l}{Verification condition required for proof validity} \\

    \cellcolor{blue!8}$t$ & \multicolumn{1}{l}{The $t$-th training epoch out of $T$ total epochs} \\
    \cellcolor{blue!8}$W_t$ & \multicolumn{1}{l}{Model weights at epoch $t$} \\
    \cellcolor{blue!8}$\mathbf{B}_t$ &  \multicolumn{1}{l}{Data batch sampled from dataset $\mathcal{D}$ at epoch $t$} \\
    \cellcolor{blue!8}$H_t, A_t$ & \multicolumn{1}{l}{Signature and auxiliary info at epoch $t$} \\
    \cellcolor{blue!8}$\Lambda$ & \multicolumn{1}{l}{Ownership information embedded in the model} \\
    
    \cellcolor{yellow!15}$s_x$ & \multicolumn{1}{l}{The $x$-th shard out of $S$ shards} \\
    \cellcolor{yellow!15}$T_x$ & Set of epochs in $s_x$, $\bigcup_{x=1}^{S} T_x = \{1, 2, \dots, T\}$, \\ 
    \cellcolor{yellow!15}  & \text{and} $T_x \cap T_{x'} = \emptyset$ \text{for } $x \ne x'$ \\
    {\cellcolor{yellow!15}$t_x$} & The $t$-th training epoch, which belongs to shard $s_x$ \\
    {\cellcolor{yellow!15}$\bar{t}_x$} & The $t$-th training epoch, which is the final epoch of shard $s_x$ \\
    {\cellcolor{yellow!15}$W_x$} & The checkpoint model of shard $s_x$ \\
    \multicolumn{1}{c}
    {\cellcolor{yellow!15}$\Lambda_x$} & Embedded watermark in $s_x$ \\
    \multicolumn{1}{c}{\cellcolor{yellow!15}$k_x$} & Embedding key for $\Lambda_x$ \\
    \multicolumn{1}{c}{\cellcolor{yellow!15}$Y_x$} & Shard-specific selection matrix derived as $Y_x=\textsf{WMPosition}(\mathcal{H}_{x-1})$ \\
    \multicolumn{1}{c}{\cellcolor{yellow!15}$Z$} & Selection matrix for weights subject to the Gaussian obfuscation \\
    \multicolumn{1}{c}{\cellcolor{yellow!15}$\mu_x$} & Unique secret nonce used in hash computation \\
    \multicolumn{1}{c}{\cellcolor{yellow!15}$\eta$} & Watermark detection rate \\
    \multicolumn{1}{c}{\cellcolor{blue!8}$\mathbb{E}(\cdot)$/$\mathbb{C}(\cdot)$} & \multicolumn{1}{l}{Watermark embedding/extraction function} \\
    \multicolumn{1}{c}{\cellcolor{blue!8}$l_w(\cdot)$/$l_\Lambda(\cdot)$} & \multicolumn{1}{l}{The loss function for main task and embedding watermark} \\
    % \multicolumn{1}{c}{\cellcolor{blue!8}$l_\Lambda(\cdot)$} & \multicolumn{1}{l}{The loss function for embedding ownership information} \\
     \multicolumn{1}{c}{\cellcolor{blue!8}$\sigma$} & \multicolumn{1}{l}{The privacy  budget of  Gaussian noise} \\
    \bottomrule
    \end{tabular}
    }
    \end{threeparttable}
    \end{center}
    \end{table}

%%%%%%%%%%%%%%%%%%%%%%%%%%%%%%%%%%%%%%%%%%%%%%%%%%%%%%

\begin{algorithm}[b]
    \linespread{1.1}
    \footnotesize
    \caption{\textsc{PoLO} generation}
    \begin{algorithmic}[1]
    \Require $\mathbb{D}$ (\textcolor{violet}{dataset}), $\mu$/$id$ (\textcolor{violet}{secret nonce/information}), 
    $\sigma$ (\textcolor{violet}{privacy budget of Gaussian noise}),
    $t_{max}$ (\textcolor{violet}{maximum training iterations}),
    key$: Z$ (\textcolor{violet}{selection matrix for weights subject to Gaussian noise}), $\eta_G^{\text{emb}}$ (\textcolor{violet}{threshold for successful watermark embedding}), $l_\Lambda$ (\textcolor{violet}{watermark embedding regularizers}).
    \Ensure \textsc{PoLO} proof $\mathbb{P}$
    
    \State $W_0$, $\mu$, $x$ \Comment{initial model, and shard number.}
    % \State Initialize $\Lambda_1, k_1 = \mathbb{H}(W_0, 1, \mu_0, id_\mathcal{P})$ \Comment{$x=1$}.
    \State Initialize $\Lambda_1, k_1, Y_1 = \textsf{WMGen}(\mathcal{H}_0), \textsf{KeyGen}(\mathcal{H}_0), \textsf{WMPosition}(\mathcal{H}_0)$,
    
     % \State \
     % \quad 
     where $\mathcal{H}_0=\mathbb{H}(W_0, 1, \mu, id_\mathcal{P})$.

     % \State $Y = \textsf{WMPosition}(\mu)$.\Comment{selection matrix for weights used for $\Lambda$}

    \While{ ($\lnot$ converging) $\lor$ ($t_{max}$ is reached)} \Comment{until convergence}

        \State $W_{t_x} = \arg\min_{W} \left( l_w(W) + \lambda l_\Lambda(W) \right) 
            \left. \vphantom{\sum} \right\rvert_{\substack{W_0 = W_{x-1}}} + \delta W$ 
            \State \quad where $\delta W = \mathbb{E}(W_{x-1}, \Lambda_x, k_x, Y_x)$ \Comment{update weights}

            \State $\eta = 1 - \frac{\sum_{i}^{n} \left( \mathbb{C}(W_{t_x}, Y_x, k_x)[i] \ne \Lambda_x[i] \right)}{n}$ \Comment{update watermark detection rate}

        \If{$\eta \geq \eta_G^{\text{emb}}$}

            \State  Save $W_x = W_{t_x} + \mathbb{G}(\sigma, Z, W_{t_x} )$ \Comment{apply Gaussian noise}
            \State $\Lambda_{x+1}, k_{x+1}, Y_{x+1} = \textsf{WMGen}(\mathcal{H}_x), \textsf{KeyGen}(\mathcal{H}_x),$
            \State $ \textsf{WMPosition}(\mathcal{H}_x)$,
            % \State \quad 
            where $\mathcal{H}_x=\mathbb{H}(W_x, x+1, \mu, id_\mathcal{P})$ \Comment{hash anchored at completed shard $x$; consistent with Eq.~\ref{wm_k_gen} via $\mathcal{H}_{x+1}^{\text{Eq}} = \mathcal{H}_x^{\text{Alg}}$}
            \State $x += 1$
        \Else
            $\quad W_{\text{prev}} = W_{t_x}$ \Comment{warm-start next training step within shard $s_x$; does \emph{not} alter the chain anchor $W_{x-1}$, which remains the last committed checkpoint of shard $s_{x-1}$}
        \EndIf
        
    \EndWhile

    % \For{$t \leftarrow 1,...,t_{max}$ } $\triangleright$ training epochs
    %     \State $W_{t_x} = \mathbb{F} (W_{t_{x-1}},A,l_w + \lambda l_\Lambda) + \mathbb{E} (W_{t_{x-1}}, \Lambda_x,k_x,X),$ $\triangleright$ update the model weights
    %     \State $\eta = 1-\frac{ \sum_{i}^{n} (\mathbb{C}(W_{t_x},X,k_x
    % )[i]\ne \Lambda_x[i])}{n},$ $\triangleright$ calculate the watermark detection rate $\eta$
    %     \State \textbf{if} $\eta \geq \eta_G$ \textbf{then} 
    %         \State  \hspace{4mm} $W_x = W_t + \mathbb{G}(\sigma, Z, W_{t_x} )$ $\triangleright$ apply Gaussian noise
    %         % \State  \hspace{4mm} Test and save($Acc_{{main}_s}$)
    %         \State \hspace{4mm} Save model($W_x$)
    %         \State \hspace{4mm} $\Lambda_{x+1}, k_{x+1} = \mathbb{H}(W_x, x+1, \sigma_x, id_\mathcal{P})$ $\triangleright$ calculate next shard's watermark
    %         \State \hspace{4mm} $x += 1$
    %     % \State \textbf{if} $e> t_{max}$  \textbf{then}
    %     %     \State \hspace{4mm} \textbf{break}
    % \EndFor  
    \label{alg_1:polo generation}
    \State \textbf{RETURN} $\mathbb{P} \in \{W_1,W_2,...,W_S\}$
    \end{algorithmic}
    \label{algo:polo_generation}
    \end{algorithm}

\noindent\textbf{Note on chain integrity (Algorithm~\ref{algo:polo_generation}, Else branch).}
When the watermark detection rate has not yet reached $\eta_G$, the algorithm continues training from the current weight state $W_{t_x}$. In the pseudocode, $W_{\text{prev}}$ is a local training cursor used to warm-start the next gradient step; it is distinct from the \emph{chain anchor} $W_{x-1}$, which always refers to the final Gaussian-protected checkpoint saved at the end of the \emph{previous} shard $s_{x-1}$. The chain anchor $W_{x-1}$ is never overwritten during shard $s_x$'s training loop — it is only updated when a new shard is committed (line after $x \mathrel{+}= 1$). The verifier reconstructs $(\Lambda_x, k_x)$ exclusively from the committed $W_{x-1}$, so the cryptographic binding remains well-defined throughout.

\begin{algorithm}[!htb]
    \linespread{1.1}
    \footnotesize
    \caption{\textsc{PoLO} verification}
    \begin{algorithmic}[1]
    \Require $\mathbb{P} \in \{W_1,W_2,...,W_S\}$ (\textcolor{violet}{\textsc{PoLO} proof}), 
    $\mu$/$id$ (\textcolor{violet}{secret nonce and identity information}), 
    $\mathbb{D}^t $ (\textcolor{violet}{public test dataset}), $\eta_{G}^{\text{ver}}$(\textcolor{violet}{verification acceptance threshold for watermark similarity}). 
 % key:$ Y $ (\textcolor{violet}{Selection matrix for weights in $W_t, t \in T$, used for $\Lambda$})
    \Ensure ``\textbf{Fail}'' or ``\textbf{Success}''.
    \State Choose $i \in \{S-1, S-2,..., 1\}$ \Comment{$\mathcal{V}$ selects the shard at which to terminate verification}

    \State $\mathcal{V}$ receives $\mathbb{P} \in \{W_1,W_2,...,W_S\}$, and $\mu$/$id$ from $\mathcal{P}$

    % \State $Y = \textsf{WMPosition}(\mu)$.\Comment{recover the selection matrix for watermark}
    
    \For{$x \leftarrow S,S-1,...,i$} \Comment{verification loop}
        \State $\hat{{Acc}}_{{main}}=\textsf{Test}(W_x, \mathbb{D}^t)$ \Comment{test the model main task accuracy} 
        
        \State $\Lambda_{x}, k_{x}, Y_{x} = \textsf{WMGen}(\mathcal{H}_{x-1}), \textsf{KeyGen}(\mathcal{H}_{x-1}),$
        \State$ \textsf{WMPosition}(\mathcal{H}_{x-1})$,
        where $\mathcal{H}_{x-1}=\mathbb{H}(W_{x-1}, x, \mu, id_\mathcal{P})$ \Comment{calculate the watermark from the previous model by $\mathbb{H}(\cdot)$}
            
        \State $\hat{\Lambda}_x = \mathbb{C}(W_{x},Y_x,k_x)$ \Comment{extract the current watermark}
        
        \State $\eta_x = 1-\frac{ \sum_{i}^{n} (\hat{\Lambda}_x[i]\ne \Lambda_x[i])}{n}$ \Comment{calculate the watermarks similarity}
        
        % \State \textbf{if} $\left | \hat{{Acc}}_{{main}}-{Acc}_{exp} \right |  > \eta_{acc}$ \textbf{then}
        \State \textbf{if} $\hat{Acc}_{main}$ not meet expectation \textbf{or} $\eta_x < \eta_G^{\text{ver}}$ \textbf{then}
        \State \hspace{4mm} \textbf{RETURN} ``\textbf{Fail}'', \textbf{break}
    \EndFor  \label{alg_4:polo verification}
    \State \textbf{RETURN} ``\textbf{Success}''
    \end{algorithmic}
    \label{algo:polo_verification}
    \end{algorithm}

\subsection{Algorithm of the Frameworks}\label{appendix_algo}

The proof generation process and verification process of \textsc{PoLO} ($\S$\ref{subsec: training_watermark}--$\S$\ref{subsec: verification_watermark}) can be found from \textbf{Algorithms}~\ref{algo:polo_generation}\&\ref{algo:polo_verification}.

\smallskip
\noindent\textbf{Proof generation.} Algorithm~\ref{algo:polo_generation} details the procedure for \textsc{PoLO} proof generation. The algorithm starts with an initial model 
$W_0$ and computes the first watermark $\Lambda_1$ and embedding key $k_1$ using a hash function $\mathbb{H}(\cdot)$, which incorporates the prover’s identity $id_\mathcal{P}$, a secret verifier-provided nonce $\mu$, and the shard index $x$.

Training proceeds in a loop over epochs, where in each iteration, the model is updated through standard training combined with watermark embedding. 
The watermark $\Lambda_x$ is embedded into selected weights of $W_{t_x}$ at each epoch $t\in T$ using a secret embedding key $k_x$ and a selection matrix $Y$. The watermark detection rate $\eta$ is dynamically computed after each epoch to assess whether the embedded watermark has reached the desired robustness threshold $\eta_G$. Once satisfied, differential privacy protection is selectively applied to a subset of the non-watermarked weights using another selection matrix $Z$, and the resulting model $W_x$ is stored as the final checkpoint of shard $s_x$.
A new watermark $\Lambda_{x+1}$ is computed from the obfuscation-protected model $W_x$ using chained hashing, cryptographically linking it to the previous shard. This process repeats until both the maximum epoch $t_{max}$ and target watermark quality $\eta_G$ are met.

% To ensure convergence and prevent indefinite training, two critical thresholds guide the training stop conditions: $E_{expect}$-the expected number of epochs for model convergence, and $E_{max}$-the maximum allowed training epochs.
% Once the training surpasses $E_{expect}$, the training may continue only if $\eta < \eta^{'}_{G}$, where $\eta^{'}_{G} < \eta_G$ is a relaxed threshold used to encourage early stopping. If $\eta > \eta^{'}_{G}$, training stops immediately to prevent unnecessary overhead. Training is forcefully terminated upon reaching 
% $E_{max}$, regardless of the current watermark detection rate.

In our implementation, we set $\eta_G^{\text{emb}}=0.99$, ensuring strong watermark robustness while maintaining training efficiency. Note that $\eta_G^{\text{emb}}$ is the shard-completion threshold during training and is distinct from the verification acceptance threshold $\eta_G^{\text{ver}}=0.70$ used in the verification algorithm (cf.\ Algorithm~\ref{algo:polo_verification} and $\S$\ref{sec:experiment}).
The final output of the algorithm is a \textsc{PoLO} proof $\mathbb{P}$. This guarantees traceable and tamper-resistant proof of training and model ownership.

\begin{figure*}[!t]
    \centering
    \subfigure[\textbf{Honest Training - } EfficientNet\_CIFAR10\_256]{
    \begin{minipage}[t]{0.48\textwidth}
    \centering
    \includegraphics[width=3.3in]{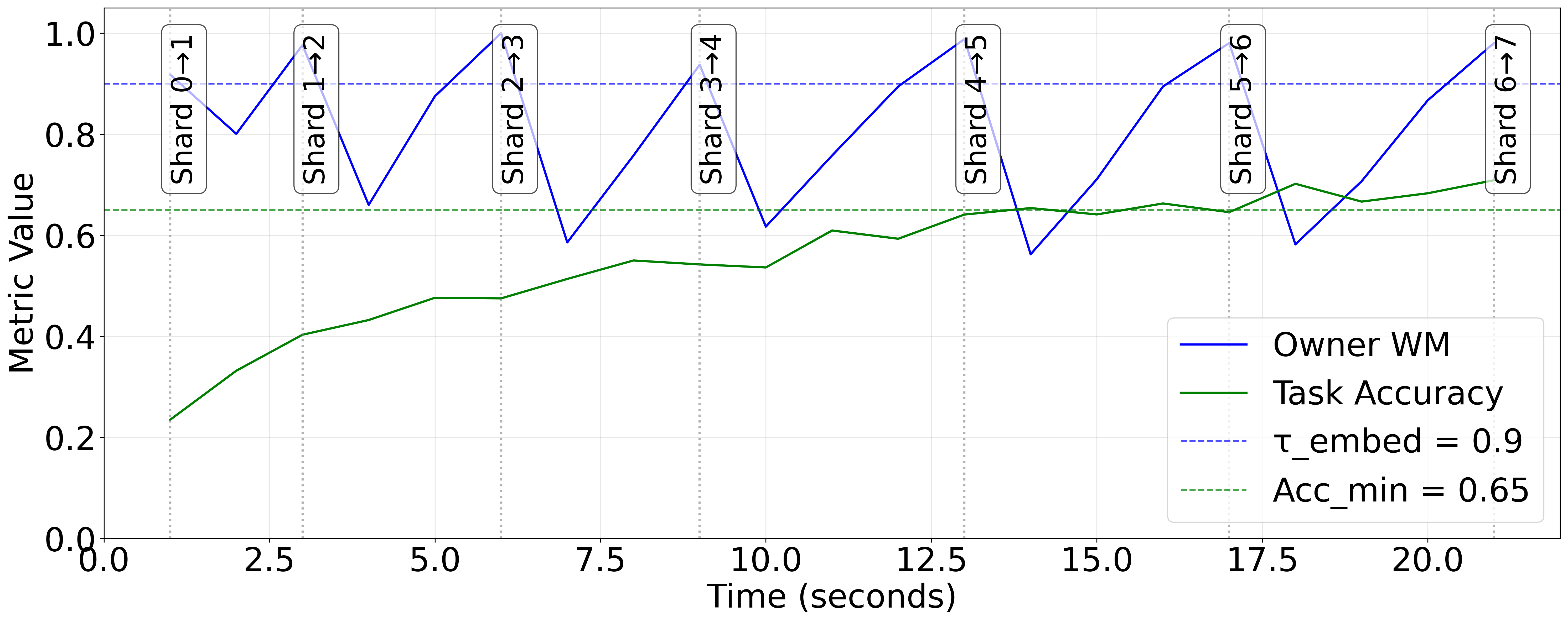}
    \end{minipage}
    \label{efficientnet_cifar10_honest}
    }
    \subfigure[\textbf{(OFA) Multi-shard Takeover} - EfficientNet\_CIFAR10\_256]{
    \begin{minipage}[t]{0.48\textwidth}
    \centering
    \includegraphics[width=3.3in]{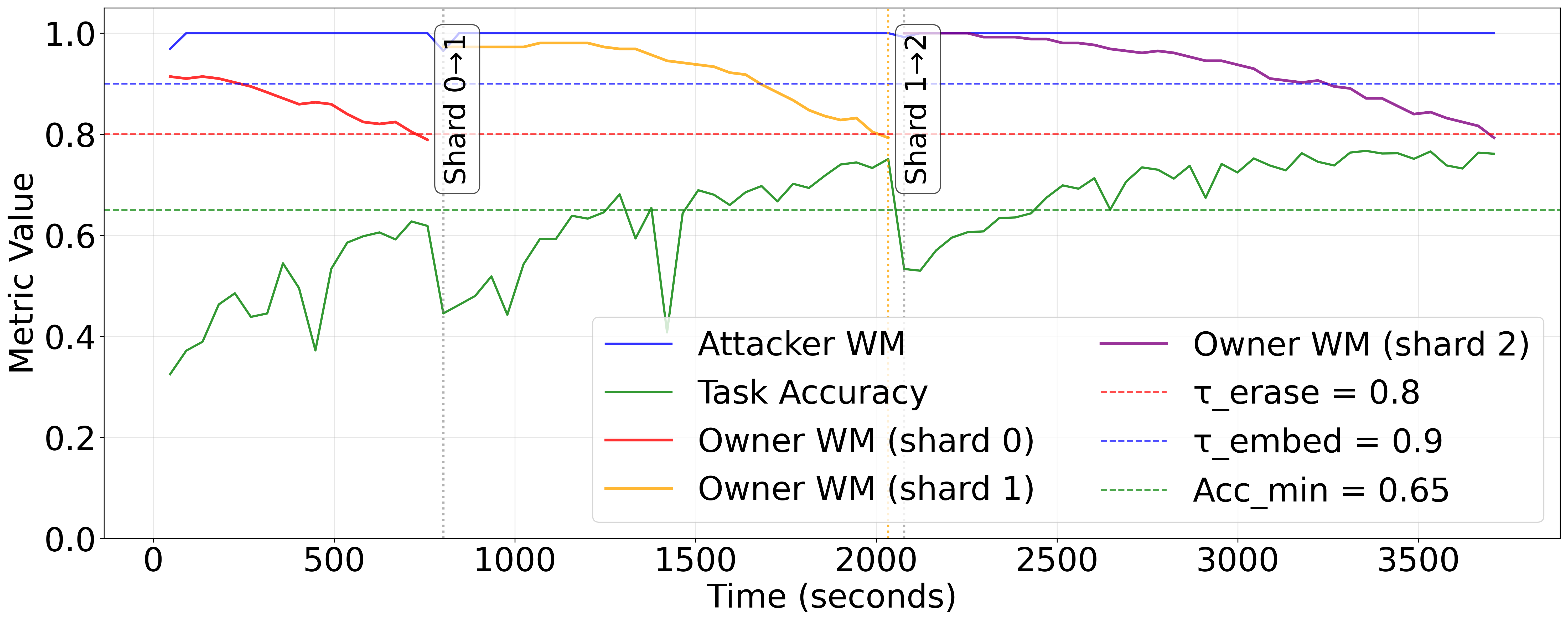}
    \end{minipage}
    \label{efficientnet_cifar10_attack}
    }

    \subfigure[\textbf{Honest Training - } ResNet18\_CIFAR100\_256]{
    \begin{minipage}[t]{0.48\textwidth}
    \centering
    \includegraphics[width=3.3in]{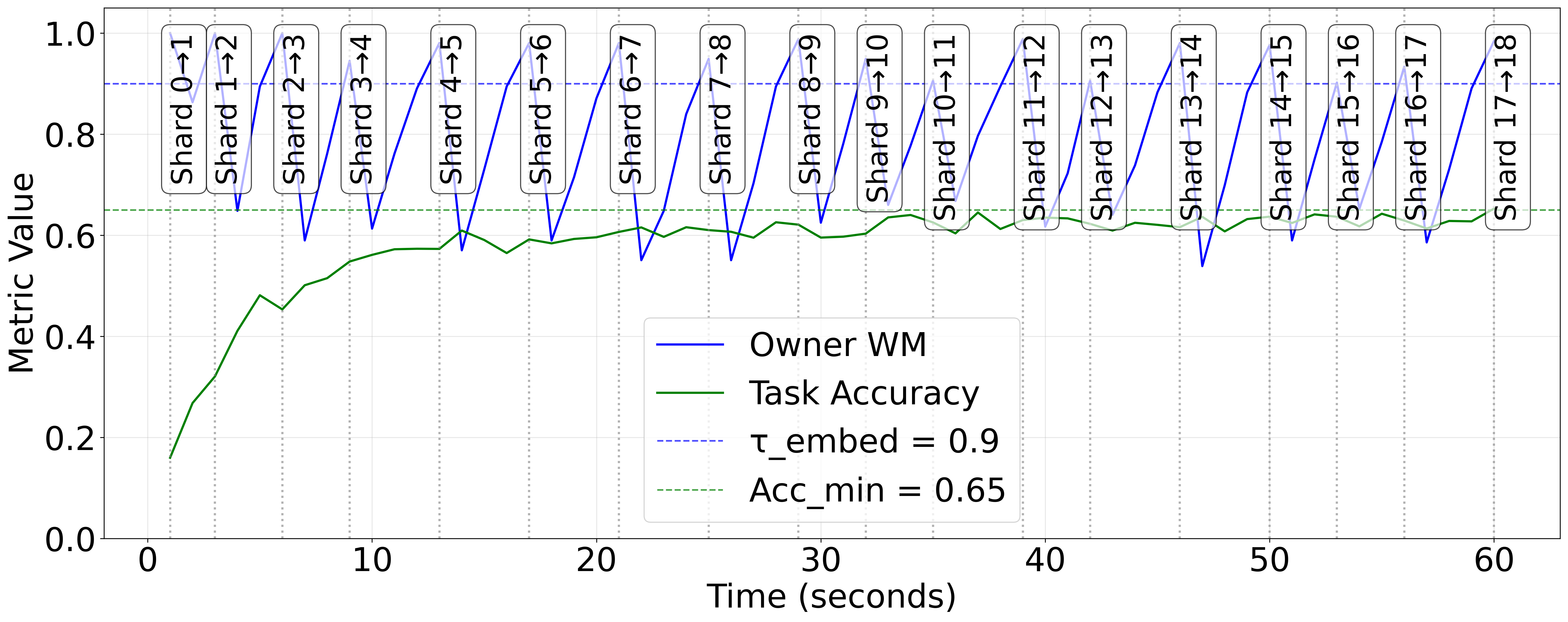}
    \end{minipage}
    \label{resnet18_cifar100_honest}
    }
    \subfigure[\textbf{(OFA) Multi-shard Takeover} - ResNet18\_CIFAR100\_256]{
    \begin{minipage}[t]{0.48\textwidth}
    \centering
    \includegraphics[width=3.3in]{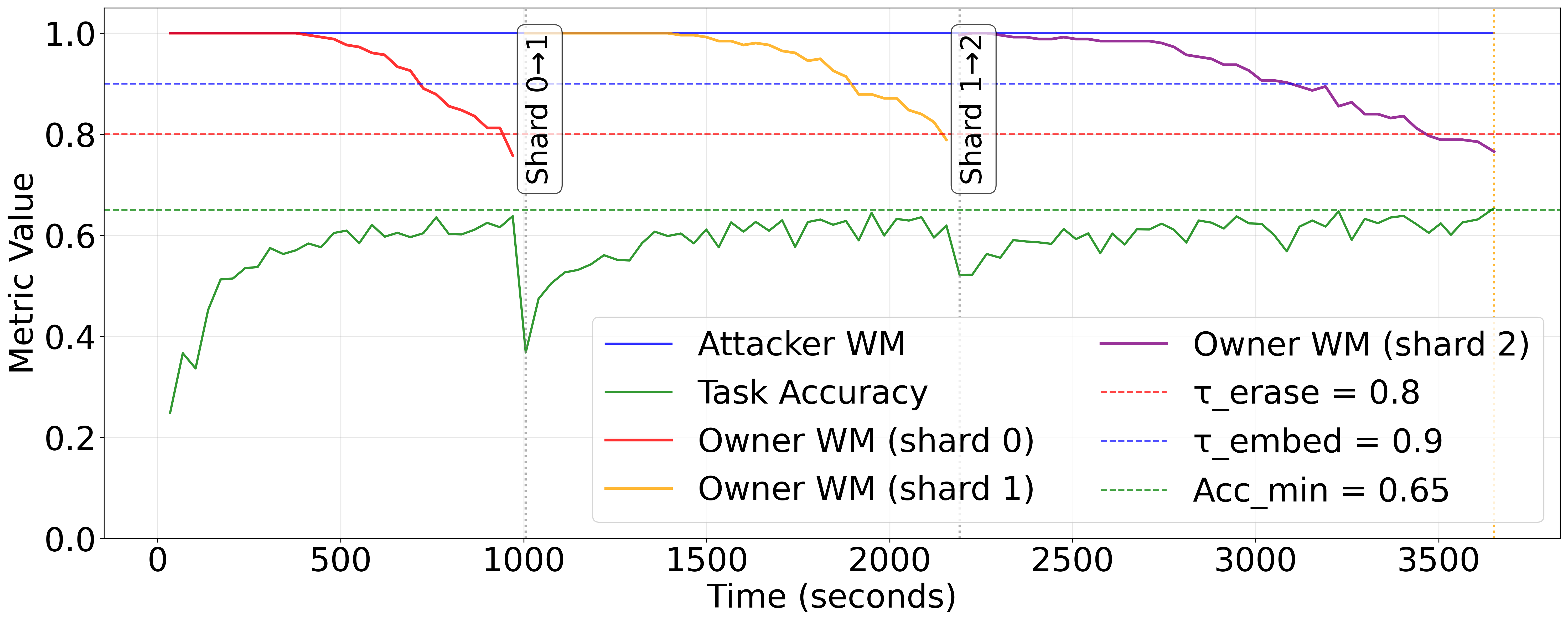}
    \end{minipage}
    \label{resnet18_cifar100_attack}
    }
\vspace{0.05in}
\caption{
Honest training showing owner watermark embedding and accuracy progression across shards. Multi-shard attack showing sequential takeover attempts with per-shard owner watermark erasure, attacker watermark embedding, and maintained accuracy.
}
\label{attack-observation}
\vspace{0.1in}
\end{figure*}

\smallskip 
\noindent\textbf{Proof verification.} Algorithm~\ref{algo:polo_verification} gives \textsc{PoLO}’s verification procedure, which jointly validates training effort and ownership by checking both the chained watermarks and the model accuracy $Acc_{main}$. The verifier $\mathcal{V}$ works at the shard level: it chooses a starting shard $i \in {1,\dots,S-1}$ from the submitted proof $\mathbb{P}$, which contains the obfuscation-protected checkpoints ${W_1,\dots,W_S}$ and the information needed to reconstruct the per–shard watermarks $\Lambda_x$ and keys $k_x$.

\begin{packeditemize}
    
\item\textbf{Main task accuracy validation.} The model $W_x$ is evaluated on the public test dataset $\mathbb{D}^t$ to measure the main task accuracy $\hat{{Acc}}_{{main}}$. If the $\hat{{Acc}}_{{main}}$ does not meet the expectation, the verification fails, and the proof is rejected.
% deviates from the expected accuracy $Acc_{exp}$ by more than a predefined threshold $\eta_{acc}$, the verification fails, and the proof is rejected.

\item\textbf{Watermark consistency verification. }
$\mathcal{V}$ reconstructs the expected watermark $\Lambda_x$ and key $k_x$ from the previous checkpoint $W_{x-1}$ using the hash function $\mathbb{H}(\cdot)$, then extracts the current watermark $\hat{\Lambda}_x$ from $W_x$ using $k_x$ and the selection matrix $Y$. It computes the similarity score $\eta_x$ via Hamming similarity; if $\eta_x < \eta_G$, verification fails.
The verifier proceeds backward from shard $S{-}1$ to $1$, with each watermark cryptographically linked to the previous one. This chaining prevents partial forgeries or selective tampering. If all checks pass, verification returns “\textbf{Success},” ensuring \textsc{PoLO}’s integrity and efficiency.

\end{packeditemize}

\subsection{Multi-shard Takeover: OFA Analysis}
\label{attack-curve}

To further examine the practicality of ownership forgery under \textsc{PoLO}, we provide a supplementary analysis of a shard-wise instantiation of the OFA attack.
Unlike conventional single-shot fine-tuning attack that targets only the final model, the OFA attack is a multi-shard takeover that applies fine-tuning sequentially across shard checkpoints $\{W_x\}$, aiming to progressively erase the owner’s chained watermarks $\{\Lambda_x\}$ and embed attacker-controlled watermarks while preserving main-task accuracy.
This setting directly stress-tests whether OFA remains viable once ownership claims are cryptographically bound to the full training trajectory.

We consider an attacker who has access to all obfuscation-protected checkpoint models $W_x$ but does not possess the legitimate watermark keys and embedding position information $(k_x, Y)$.
At each shard, the attacker fine-tunes the model to jointly suppress the owner’s watermark and inject a forged watermark $\Lambda_x'$.
The attack objective is formalized as the following optimization problem:
\begin{equation}
\label{eq:attack_loss}
% loss = task\_loss + attack\_lambda * watermark\_loss + blind\_erasure\_lambda * perturbation\_loss
l_{\text{ofa}} = l_{w} + \lambda_{\Lambda'} \, l_{\Lambda'} + \lambda_{\text{blind}} \, l_{\text{pert}},
\end{equation} 
where $l_w$ denotes the main task loss, $l_{\Lambda'}$ enforces the attacker’s target watermark, and $l_{\text{pert}}$ penalizes excessive deviation from the shard-initial parameters to maintain model fidelity.
Specifically,
\begin{equation}
\label{eq:attack_loss_p}
% perturbation\_loss = sum(mean((param - initial\_param)^2)) \text{ for all parameters}
l_{\text{pert}} = \sum_{w \in W_x'} \mathbb{E}\!\left[ \left\| w - w^{(0)} \right\|_2^2 \right],
\end{equation}
where $w^{(0)}$ denotes the checkpoint parameters at the beginning of the shard $W_x$, while $W_x'$ represents all trainable weights in the attack process.
Here, we set $\lambda_{\text{blind}} = 0.1$, following standard OFA practice to balance watermark suppression and task accuracy preservation.

Fig.\ref{attack-observation} visualizes both honest training and multi-shard takeover OFA on EfficientNet with CIFAR-10 and ResNet18 with CIFAR-100.
Under honest training, the owner’s watermark detection rate $\eta$ exhibits a clear shard-wise embedding pattern while the main-task accuracy increases steadily.
In contrast, during a multi-shard takeover, the attacker must repeatedly fine-tune each shard checkpoint for extended durations in order to reduce the legitimate watermark below the erasure threshold and embed a forged one.
Because multi-shard takeover incurs substantially higher wall-clock cost than honest training, we only snapshot the first three shards in the attack plots for clarity.

More importantly, the attack fails to achieve a favorable trade-off between watermark erasure and task fidelity.
As shown in Figs.\ref{efficientnet_cifar10_attack} and \ref{resnet18_cifar100_attack}, suppressing the owner’s legitimate watermark inevitably requires aggressive fine-tuning that either degrades the main-task accuracy or demands prolonged optimization comparable to, or exceeding, honest shard training.
Across all settings, attempts to maintain main task accuracy above the threshold leave detectable remnants of the owner’s watermark, while successful erasure of the watermark leads to noticeable accuracy degradation that violates the verifier’s acceptance criteria.

These observations reinforce the core security intuition of \textsc{PoLO}.
Because each shard watermark $\Lambda_x$ is deterministically derived from the preceding checkpoint, multi-shard takeover forces the attacker to reconstruct the training process shard by shard.
Any attempt to shortcut this process either fails watermark verification or violates the main task performance constraints.
Consequently, multi-shard OFA becomes economically irrational in the EMM setting, requiring cost comparable to or exceeding honest training while offering no reliable path to successful ownership forgery.

% Multi-shard takeover applies OFA-style fine-tuning sequentially across shard checkpoints, aiming to replace the owner’s chained watermarks.

% Multi-shard Takeover spends way longer time than honest training, so for multi-shard figure we only snapshot the first three shards.

\begin{table*}[!ht]
\centering
\caption{Visual comparison of reconstructed samples under different Gaussian noise levels ($\sigma$) for CIFAR-10 and CIFAR-100.}
\label{table-cluster}
\renewcommand{\arraystretch}{1.2} % increase row height for readability
\setlength{\tabcolsep}{8pt} % adjust horizontal spacing
\resizebox{\textwidth}{!}{ % still keeps table width within page
\begin{threeparttable}
{\Large % <<-- make font larger (try \large, \Large, or \LARGE)
\begin{tabular}{ccc}
\toprule
       & \textbf{CIFAR-10} & \textbf{CIFAR-100} \\ 
\midrule
\textbf{Target} & \includegraphics[width=\linewidth]{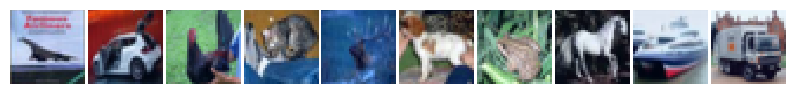}     
       & \includegraphics[width=\linewidth]{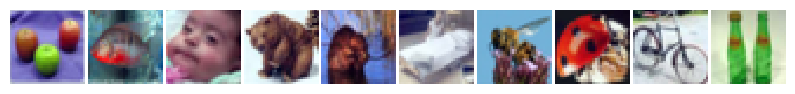} \\ 
\midrule
\textbf{Baseline} & \includegraphics[width=\linewidth]{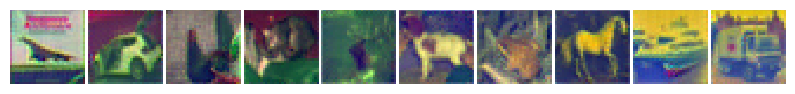}     
          & \includegraphics[width=\linewidth]{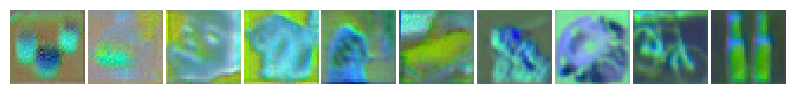} \\ 
\midrule
\textbf{$\sigma$=1} & \includegraphics[width=\linewidth]{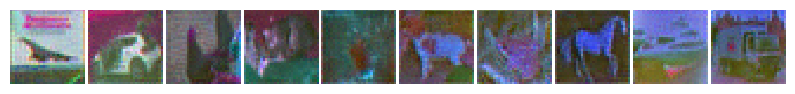}     
           & \includegraphics[width=\linewidth]{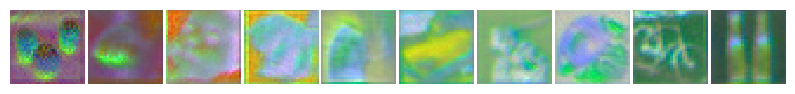} \\ 
\midrule
\textbf{$\sigma$=4} & \includegraphics[width=\linewidth]{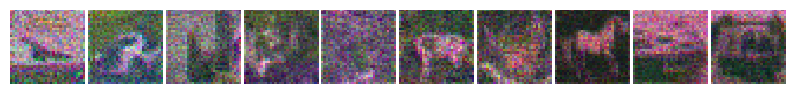}     
           & \includegraphics[width=\linewidth]{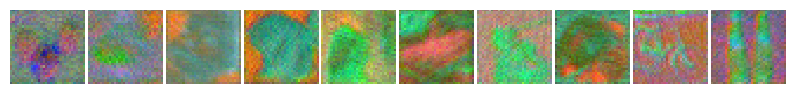} \\ 
\midrule
\textbf{$\sigma$=8} & \includegraphics[width=\linewidth]{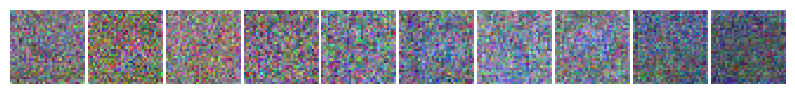}     
           & \includegraphics[width=\linewidth]{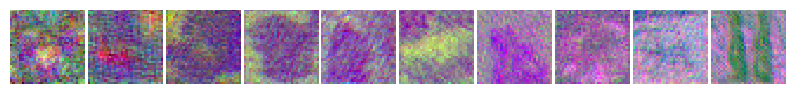} \\ 
\bottomrule
\end{tabular}
}

\begin{tablenotes}
{\large
    \item $\bullet$ Reconstructions from unprotected checkpoints exhibit relatively high similarity with the private data ($SSIM$ = 0.49 on CIFAR-10 and 0.25 on CIFAR-100). When Gaussian noise is injected at release time, reconstruction quality degrades substantially: on CIFAR-10, $SSIM$ decreases to 0.39, 0.26, and 0.06; on CIFAR-100, it decreases to 0.18, 0.14, and 0.09 for $\sigma$=1,4,8, respectively.
}
\end{tablenotes}
\end{threeparttable}
}
\end{table*}

%%%%%%%%%%%%%%%%%%%%%%%%%%%%%%%%%%%%%%%%%%%%%%%%%%%%%%%%%%%%%

\subsection{Data Inference Attackers}
\label{data_inference_attack}

To demonstrate that Gaussian noise injection can effectively protect model weights against reconstruction attempts, we have evaluated a white-box, RecoNN-based reconstruction membership inference attack (MIA). %~\cite{data_reconstruct}. 
In this setting, the adversary is assumed to know the training pipeline and architecture, has access to the released checkpoint, and holds a reference dataset from the same distribution, but cannot directly observe the private training data. The adversary’s goal is to reconstruct the private training samples from the released checkpoint. 
Following a shadow-model \& reconstructor approach, shadow models are first trained on the reference dataset, and a lightweight convolutional decoder is then trained to map model weights and outputs to image reconstructions. At test time, the released checkpoint is fed into the reconstructor to generate guesses of the private training samples.

Structural Similarity Index ($SSIM$) is adopted to quantify the similarity between reconstructed and original images. 
As shown in Tab.\ref{table-cluster}, visual comparisons further confirm that increasing the noise strength renders reconstructed samples almost indistinguishable from random noise. These results demonstrate that Gaussian perturbation of released checkpoints provides an effective defense against data reconstruction MIAs.

\subsection{Representative Replay-PoL Attack}
\label{appendix_reply-pol_attack}

Tab.~\ref{tab:pol-attack-applicability} shows, in one sentence, that the classic attacks against replay-based PoL are no longer applicable once verification is changed from gradient/trajectory replay to shard receipt + chained watermark extraction. Concretely, earlier attacks rely on verifier behaviors such as replaying training trajectories, checking L2 consistency between checkpoints, spot-checking only a subset of training steps, or coupling verification with dataset evidence; \textsc{PoLO} does none of these. Instead, it verifies whether the current shard checkpoint contains the correct watermark derived from the previous shard and whether the model still meets the task-accuracy requirement. So the point of Tab.~\ref{tab:pol-attack-applicability} is not that POLO experimentally ``outperforms'' those attacks, but that it removes their required attack surface by design: if the verifier no longer replays trajectories, then adversarial trajectory, checkpoint interpolation, step-subset abuse, and data-forging style replay attacks simply have no protocol entry point in \textsc{PoLO}.
% \input{tables/table-mapping-2}

%%%%%%%%%%%%%%%%%%%%%%%%%%%%%%%%%%%%%%%%%%%%%%%%%%%%%%%%%%%%
\subsection{Configuration of Shard}
\label{appendix_config_shard}
%%%%%%%%%%%%%%%%%%%%%%%%%%%%%%%%%%%%%%%%%%%%%%%%%%%%%%%%%%%%

To examine how learning rate $LR$ and watermark detection threshold $\eta_G$ influence shard formation, we conducted experiments on two benchmark datasets: CIFAR-10 with ResNet18, and AG News with TextCNN. In particular, we varied $LR$ and $\eta_G$, and measured their impact on shard count, shard rate 
$|s|$, and the main task accuracy $Acc_{main}$.

The results are reported in Tab.~\ref{tab:merged_lr_eta}. We observe that increasing the $LR$ consistently raises the $|s|$; for instance, on CIFAR-10, $|s|$ grows from 0.36 at $LR$=0.01 to 0.78 at 
$LR$=0.2, and on AG News, from 0.63 at $LR$=0.0005 to 0.89 at $LR$=0.01. In contrast, increasing $\eta_G$ reduces $|s|$, as stricter requirements slow down watermark embedding and result in fewer validated shards. For example, on CIFAR-10, $|s|$ drops from 0.86 at $\eta_G$=0.85 to 0.64 at $\eta_G$=0.99. Importantly, across all $LR$ and $\eta_G$ configurations, $Acc_{main}$ ($\approx$0.90–0.92), indicating that PoLO preserves task performance while allowing shard dynamics to be flexibly tuned.

%%%%%%%%%%%%%%%%%%%%%%%%%%%%%%%%%%%%%%%%%%%%%%%%%%%%%%%%%%%%

%%%%%%%%%%%%%%%%%%%%%%%%%%%%%%%%%%%%%%%%%%%%%%%%%%%%%%%%%%%%

\subsection{\textsc{PoLO} for LLM}
\label{appendix_PoLO_LLM}

PoLO is architecture- and stage-agnostic, and can be applied during alignment/fine-tuning, including PEFT (LoRA). It adds a parameter-level watermark loss to the usual objective,
$ L_{\text{total}} = L_{\text{task}} + \lambda_{\text{wm}}\, L_{\text{wm}} $.
Using a secret key, we index a small subset $S$ of trainable parameters in selected attention/FFN projections with LoRA; $S$ is taken from the adaptor matrices in those modules. For a secret bitstring $\{b_i\}$, we assign signed targets $t_i \in \{+\beta,\,-\beta\}$ 
%(e.g., $t_i=\beta(2b_i-1)$) 
and set $ L_{\text{wm}}=\frac{1}{|S|}\sum_{j\in S}\left\|\theta_j - t_{\pi(j)}\right\|_2^2 $,
where $\theta_j$ are the selected parameters and $\pi$ maps indices in $S$ to bits. Training incurs negligible additional overhead, and verification is checkpoint-only: read the same indices (adapter or merged weights), threshold signs $\operatorname{sgn}(\theta_j)$, compute the Hamming match-rate, and accept if it exceeds $\tau$.

To validate applicability of \textsc{PoLO} to LLMs, we conduct \emph{alignment-stage} experiments on \texttt{microsoft/DialoGPT-medium-345M}, \\
\texttt{EleutherAI/gpt-neo-1.3B}, and \texttt{LLaMA-3.1-8B}. All runs use a MacBook Pro (2024, Mac16,6) with an Apple M4 Max, 128 GB unified memory, 16 CPU cores (12 performance + 4 efficiency), and 40 GPU cores (Metal 3). For alignment, we draw $\sim$1,000 samples from \texttt{wikitext-2-raw-v1} and \texttt{allenai/wildchat}, so alignment completes in roughly 8 epochs on average.

In this setting, the alignment stage is treated as the \emph{honest training} performed by the model owner, into which \textsc{PoLO} is natively integrated for proof generation. This is conceptually different from the \emph{forgery} setting (OFA), where an attacker performs fine-tuning to remove the owner’s watermark and re-embed their own. The latter is evaluated separately in our attack experiments (Tab.~\ref{tab_gen-ver} and Fig.~\ref{attack}), which integrate \textsc{PoLO} into full training and naturally require many more epochs than this shorter alignment stage.

The results, reported in Tab.\ref{LLM}, show that the additional cost of proof generation introduced by PoLO is minimal: 4.31 and 3.67 minutes for DialoGPT-medium, and 0.53 and 0.38 minutes for GPT-Neo-1.3B. In contrast, verification time is dramatically reduced compared to the baseline PoL: for DialoGPT-medium, 0.26–0.28 minutes versus 14.56 minutes; for GPT-Neo-1.3B, 1.83–1.92 minutes versus 85.33 minutes. For LLaMA-3.1-8B, PoLO adds only about 1.2–2.8 minutes of proof generation (2.29/3.91 vs.\ 1.14 minutes) while cutting verification from roughly 77 minutes to under 1 minute (0.89/0.88 minutes). These findings confirm that PoLO can be effectively integrated into LLM alignment with negligible training overhead while offering substantial improvements in verification efficiency.
\begin{table}[t]
\centering
\caption{Comparison between PoLO and baseline LLM alignment/fine-tuning in terms of time consumption for proof generation (+training) and verification.}
\label{LLM}
\renewcommand{\arraystretch}{1.05}
\resizebox{\linewidth}{!}{%
\begin{tabular}{c|ccc|ccc}
\toprule
\multicolumn{1}{c}{\multirow{2}{*}{ \textbf{Model} }} &  \multirow{2}{*}{ \makecell{\textbf{Watermark} \\ \textbf{size} (bit)}} &  \multirow{2}{*}{ \textbf{Loss}} &  \multirow{2}{*}{ \makecell{\textbf{$|s|$} \\ (\%) } } & \multicolumn{3}{c}{\textbf{Time} (min)} \\

\multicolumn{1}{c}{} & & & & Train.  & Gener. & Verif. \\
\midrule
\multicolumn{7}{c}{\textbf{DialoGPT-medium-345M, wikitext-2-raw-v}} \\
\midrule
 GD-PoL &\cellcolor{yellow!15}  - & \cellcolor{yellow!15} 3.5893 & \cellcolor{yellow!15}  - &\cellcolor{yellow!15}  2.19 & \cellcolor{yellow!15}  0.03 & \cellcolor{yellow!15}  2.17 \\
 \multirow{2}{*}{PoLO}&  \cellcolor{blue!8}  512  & \cellcolor{blue!8} 3.5737 & \cellcolor{blue!8} 27.27 & \cellcolor{blue!8} 2.91  & \cellcolor{blue!8} 0.63 & \cellcolor{blue!8} 0.04 \\
 & \cellcolor{blue!8}  1024 &\cellcolor{blue!8} 3.5765 &\cellcolor{blue!8} 25.00 & \cellcolor{blue!8} 2.94 &\cellcolor{blue!8} 0.67 & \cellcolor{blue!8} 0.04 \\
 
 \midrule
 \multicolumn{7}{c}{\textbf{GPT-Neo-1.3B, allenai/wildchat}} \\
 \midrule
 
 GD-PoL & \cellcolor{yellow!15}  - & \cellcolor{yellow!15} 1.5017 & \cellcolor{yellow!15}  - & \cellcolor{yellow!15}  14.18 &  \cellcolor{yellow!15}  0.21 & \cellcolor{yellow!15}  14.77 \\
 \multirow{2}{*}{PoLO} & \cellcolor{blue!8}  512  & \cellcolor{blue!8}  1.5004 & \cellcolor{blue!8}  42.86 & \cellcolor{blue!8}  14.37  & \cellcolor{blue!8}  0.53 & \cellcolor{blue!8}  0.13 \\
 & \cellcolor{blue!8}  1024 & \cellcolor{blue!8}  1.5021 & \cellcolor{blue!8}  42.86 & \cellcolor{blue!8}  14.98  & \cellcolor{blue!8}  0.38 & \cellcolor{blue!8}  0.13 \\

\midrule
 \multicolumn{7}{c}{\textbf{LLaMA-3.1-8B, allenai/wildchat}} \\
 \midrule
 
 GD-PoL & \cellcolor{yellow!15}  - & \cellcolor{yellow!15} 1.2511 & \cellcolor{yellow!15}  - & \cellcolor{yellow!15}  76.85 &  \cellcolor{yellow!15}  1.14 & \cellcolor{yellow!15}  77.06 \\
 \multirow{2}{*}{PoLO} & \cellcolor{blue!8}  512  & \cellcolor{blue!8}  1.1951 & \cellcolor{blue!8}  37.50 & \cellcolor{blue!8}  81.18  & \cellcolor{blue!8}  2.29 & \cellcolor{blue!8}  0.89 \\
 & \cellcolor{blue!8}  1024 & \cellcolor{blue!8}  1.2273 & \cellcolor{blue!8}  33.33 & \cellcolor{blue!8}  85.30  & \cellcolor{blue!8}  3.91 & \cellcolor{blue!8}  0.88 \\
 
\bottomrule
\end{tabular}
}
%\begin{tablenotes} 
%    \footnotesize
%    \item $\bullet$ \multirow{3}{*}{}  \multirow{3}{*}{} 
%\end{tablenotes}
\end{table}

%%%%%%%%%%%%%%%%%%%%%%%%%%%%%%%%%%%%%%%%%%%%%%%%%%%%%%%%%%%%%

\subsection{Optional: Differential Privacy (DP) Strengthening for \textsc{PoLO}}
\label{appendix_DP}

%\textsc{PoLO} is by default designed to apply a Gaussian obfuscation to checkpoints to protect the published model parameters. 
By default, \textsc{PoLO} obfuscates released checkpoints with Gaussian noise to safeguard the published model parameters.
\textbf{Appendix}~\ref{data_inference_attack} shows that this mechanism effectively thwarts MIA that attempt to reconstruct training data from released weights. 

If deployments opt for high privacy level guarantees, we provide an optional variant that delivers formal $(\varepsilon,\delta)$-DP at release time.
% We partition the weights as $W = \big[\,W_{\mathrm{wm}},\, W_{\mathrm{nonwm}}\,\big]$.
% The watermark block $W_{\mathrm{wm}}$ is data-independent because we mask the task gradients on the selected coordinates $Y$ and drive them to fixed targets determined by $(\Lambda, k, Y)$, where the watermark bits, key, and selection matrix are derived from the verifier-provided secret nonce. Hence $W_{\mathrm{wm}}$ carries negligible information about the training data.
% For the non-watermarked block $W_{\mathrm{nonwm}}$, DP-based Gaussian noise is added once per shard at release time, avoiding per-step accounting. 
% 不要分水印和非水印部分了，直接全部加噪
In this \textsc{PoLO} variant, we globally clip checkpoints and then add a parameter-level Gaussian mechanism before release. Formally, for shard $s$,
\begin{equation}
\label{eq:release-map}
\theta_s(D)\;=\;\theta_{s-1}\;+\;\mathrm{Clip}\big(u_s(D),\,G_s\big),
\end{equation}
where $\theta_s$ and $\theta_{s-1}$ is the final parameters in shard $s$ and shard $s-1$ respectively, and $u_s(D)$ is the total parameter update accumulated within shard $s$ on dataset $D$, and
\begin{equation}
\label{eq:clip}
\mathrm{Clip}(u, G_s)\;=\;u\cdot \min\!\left(1,\,\frac{G_s}{\|u\|_2}\right).
\end{equation}

\begin{lem}[Global $\ell_2$-sensitivity]
\label{lem_global_sensitivity}
For any neighboring datasets $D\sim D'$ (differing in one sample), the global $\ell_2$-sensitivity $\Delta_2$ for $\theta_s$ is
\begin{equation}
\Delta_2(\theta_s)\;:=\;\sup_{D\sim D'}\big\|\theta_s(D)-\theta_s(D')\big\|_2
\;\le\; 2G_s .
\end{equation}
\end{lem}

\begin{proof}
By \eqref{eq:release-map},
\begin{equation}
\big\|\theta_s(D)\!-\!\theta_s(D')\big\|_2 \!=\!\big\|\mathrm{Clip}(u_s(D),G_s)\!-\!\mathrm{Clip}(u_s(D'),G_s)\big\|_2.
\end{equation}
The global clip guarantees $\|\mathrm{Clip}(u,G_s)\|_2\le G_s$ for all $u$; hence, by the triangle inequality,
\begin{equation}
\begin{aligned}
    & \big\|\mathrm{Clip}(u_s(D),G_s)-\mathrm{Clip}(u_s(D'),G_s)\big\|_2 \\
    & \le \|\mathrm{Clip}(u_s(D),G_s)\|_2+\|\mathrm{Clip}(u_s(D'),G_s)\|_2
\le 2G_s.
\end{aligned}
\end{equation}
\end{proof}
The bound in Lemma 1 holds for any fixed $\theta_{s-1}$, so it applies in the usual adaptive composition setting.

\begin{thm}[DP via the Gaussian mechanism]
\label{thm:dp}
Let $\xi_s\sim\mathcal{N}(0,\sigma_s^2 I)$ and release $\mathcal{M}_s(D)=\theta_s(D)+\xi_s$, where $\mathcal{M}(D)$ denotes the Gaussian mechanism applied to dataset $D$ (i.e., the released checkpoint obtained by adding noise $\xi$ to the model parameters $\theta(D)$).
If $ \sigma_s\;\ge\;\frac{2G_s\sqrt{2\ln(1.25/\delta_s)}}{\varepsilon_s}$,
then $\mathcal{M}_s$ satisfies $(\varepsilon_s,\delta_s)$-DP. Releasing $\{\mathcal{M}_s\}_{s=1}^S$ composes by standard DP composition (basic or advanced), yielding a global $(\varepsilon,\delta)$ budget.
For each s, the mechanism $\mathcal{M}_s$ is $(\varepsilon_s,\delta_s)$-DP conditional on all previous outputs $\mathcal{M}_1(D),\dots,\mathcal{M}_{s-1}(D)$; thus, by standard adaptive composition the sequence $\{\mathcal{M}_s\}_{s=1}^S$ is $(\sum_s\varepsilon_s,\sum_s\delta_s)$-DP (or tighter under advanced composition).
\end{thm}

\begin{table*}[t]
\centering
\caption{Analytical applicability of representative replay-PoL attacks under \textsc{PoLO}'s verifier surface.}
\label{tab:pol-attack-applicability}
\renewcommand{\arraystretch}{1.05}
\resizebox{\textwidth}{!}{
\begin{threeparttable}
\begin{tabular}{l|l|l|l|l}
\midrule
\multicolumn{1}{c}{\cellcolor{yellow!15}\textbf{PoL Attack}} &
\multicolumn{1}{c}{\cellcolor{yellow!15}\textbf{Core Exploit}} &
\multicolumn{1}{c}{\cellcolor{yellow!15}\textbf{Required by PoL Verifier}} &
\multicolumn{1}{c}{\cellcolor{yellow!15}\textbf{Required by \textsc{PoLO} Verifier?}} &
\multicolumn{1}{c}{\cellcolor{yellow!15}\textbf{Applicable in \textsc{PoLO}?}} \\
\midrule

\makecell[l]{Zhang et al. (S\&P'22)~\cite{pol_attack1}:\\ Adversarial trajectory} &
\makecell[l]{$\epsilon$-tolerance abuse\\ in gradient replay} &
\makecell[l]{Dataset +\\ trajectory replay} &
Neither &
\makecell[l]{\textbf{No} -- no trajectory\\ replay surface to forge} \\

\cmidrule{1-5}

\makecell[l]{Fang et al. (EuroS\&P'23)~\cite{pol_attack2}:\\ Checkpoint interpolation} &
\makecell[l]{Interpolated checkpoints\\ pass L2 trajectory tests} &
\makecell[l]{Trajectory L2\\ distance checks} &
\makecell[l]{Receipt + watermark\\ chain consistency} &
\makecell[l]{\textbf{No} -- interpolated checkpoints\\ lack valid chained watermark} \\

\makecell[l]{ Step-subset verification abuse} &
\makecell[l]{Sparse verifier\\ replay spot-checks} &
\makecell[l]{Step-level replay\\ sampling protocol} &
\makecell[l]{Shard-level extraction\\ + receipt checks} &
\makecell[l]{\textbf{No} -- no step-sampling\\ replay protocol to exploit} \\

\cmidrule{1-5}

\makecell[l]{Thudi et al. (USENIX'22)~\cite{pol_forge}:\\ Data forging} &
\makecell[l]{Mini-batch replacement\\ in replay evidence} &
\makecell[l]{Dataset included\\ in proof workflow} &
\makecell[l]{Dataset decoupled\\ from replay verification} &
\makecell[l]{\textbf{No} -- verifier does not\\ replay dataset-coupled steps} \\

\cmidrule{1-5}
\end{tabular}
\begin{tablenotes}
\small
\item This is an analytical attack-surface comparison (not a new experiment): \textsc{PoLO} replaces replay-gradient verification with shard-level receipt and watermark-chain verification, removing the exploited verifier requirements in these replay-PoL attacks.
\end{tablenotes}
\end{threeparttable}
}
\end{table*}

\begin{proof}
Let $D\sim D'$ be neighboring datasets. 
Consider the Gaussian mechanism $\mathcal{M}_s(D)=\theta_s(D)+\xi_s$ with $\xi_s\sim\mathcal{N}(0,\sigma_s^2 I)$. For any measurable set $S$,
\begin{equation}
\begin{aligned}
    &\Pr[\mathcal{M}_s(D)\in S]=\int_S \phi_\sigma(x-\theta_s(D))\,dx,\qquad \\
    &\Pr[\mathcal{M}_s(D')\in S]=\int_S \phi_\sigma(x-\theta_s(D'))\,dx,
\end{aligned}
\end{equation}
where $\phi_\sigma$ denotes the Gaussian density with covariance $\sigma_s^2 I$. For any $x\in\mathbb{R}^d$, the likelihood ratio is
\begin{equation}
\begin{aligned}
    & \frac{\phi_\sigma(x-\theta_s(D))}{\phi_\sigma(x-\theta_s(D'))}= \\
    &\exp\!\left(\frac{\langle x\!-\!\theta_s(D'),\,\theta_s(D)\!-\!\theta_s(D')\rangle}{\sigma_s^2}
\!-\!\frac{\|\theta_s(D)\!-\!\theta_s(D')\|_2^2}{2\sigma_s^2}\right).
\end{aligned}
\end{equation}
Let $\Delta=\|\theta_s(D)-\theta_s(D')\|_2$ and $v=(\theta_s(D)-\theta_s(D'))/\Delta$. Under $x\sim \mathcal{N}(\theta_s(D'),\sigma_s^2 I)$, the scalar
\begin{equation}
Z \;:=\; \frac{\langle x-\theta_s(D'),\,v\rangle}{\sigma_s} \;\sim\; \mathcal{N}(0,1).
\end{equation}
Hence, the privacy loss random variable $L$ satisfies
\begin{equation}
L \;=\; \log\frac{\phi_\sigma(x-\theta_s(D))}{\phi_\sigma(x-\theta_s(D'))}
\;=\; \frac{\Delta}{\sigma_s} Z \;-\; \frac{\Delta^2}{2\sigma_s^2}.
\end{equation}
For any $\varepsilon_s>0$,
\begin{equation}
\begin{aligned}
    \Pr[L>\varepsilon_s]
    &=\Pr\!\left[ Z \;>\; \frac{\varepsilon_s}{\Delta/\sigma_s} + \frac{\Delta}{2\sigma_s} \right] \\
    &\;\le\; \exp\!\left(-\frac{1}{2}\left(\frac{\varepsilon_s}{\Delta/\sigma_s} + \frac{\Delta}{2\sigma_s}\right)^2\right),
\end{aligned}
\end{equation}
using the standard Gaussian tail bound. The classical calibration for the Gaussian mechanism~\cite{dwork2014algorithmic} states that choosing
\begin{equation}
\sigma_s \;\ge\; \frac{\sqrt{2\ln(1.25/\delta_s)}\;\Delta}{\varepsilon_s}
\label{eq:gauss-cal}
\end{equation}
ensures $\Pr[L>\varepsilon_s]\le \delta_s$ and thus satisfies $(\varepsilon_s,\delta_s)$-DP. 

By monotonicity in $\Delta$ and Lemma~\ref{lem_global_sensitivity}, we can take $\Delta \le \Delta_2(\theta_s)\le 2G_s$ to obtain the sufficient condition
\begin{equation}
\sigma_s \;\ge\;  \frac{2G_s\sqrt{2\ln(1.25/\delta_s)}}{\varepsilon_s}.
\end{equation}

Therefore, $\mathcal{M}_s$ is $(\varepsilon_s,\delta_s)$-DP. Standard (basic/advanced) composition then yields overall DP across shards.
\end{proof}

By clipping the per–shard update and adding Gaussian noise, we obtain a $(\varepsilon,\delta)$-DP variant of \textsc{PoLO}. This bound is independent of the number of epochs/iterations inside the shard and is only applied to the non-watermark block while keeping the watermark block data-independent. 
In the DP variant, the watermark parameters are generated and embedded using only public randomness and previous DP releases, and are not updated via gradients computed from private training data. So the watermark block is a post‑processing of DP outputs and does not affect the $(\varepsilon,\delta)$ guarantee.

%%%%%%%%%%%%%%%%%%%%%%%%%%%%%%%%%%%%%%%%%%%%%%%%%%%%%%%%%%%%%

\subsection{Security Q\&A in \textsc{PoLO}}
\label{appendix_security_privacy}

%To enhance transparency and clarity, 

We address common threat scenarios against \textsc{PoLO} and explain how each is mitigated by our design.

\vspace{3pt}
\textcolor{violet}{\textbf{\ding{172} External malicious attackers.}}
\vspace{3pt}

\textit{Q1: What if an attacker fine-tunes the final released model and claims it as theirs? (OFA strategy)}  

\textit{A1:} \textsc{PoLO} embeds watermarks across chained shards, so a valid claim requires proving ownership of all preceding shards. Since the attacker does not know the embedding positions, they must fine-tune across many layers to both insert a forged watermark and erase the original. This process is as costly or even more costly than legitimate training, making the attack economically impractical.

\vspace{3pt}
\textit{Q2: What if an attacker blindly inserts watermark into model and claims ownership? (WMA strategy)}  

\textit{A2:} Randomly embedding a watermark without preserving model integrity can degrade task performance, causing the model to fail accuracy checks during verification.

\vspace{3pt}
\textit{Q3: What if the attacker tries to guess the embedding key $k_x$ that aligns with the watermark $\Lambda_x$ in shard $s_x$ to forge a \textsc{PoLO} proof?}  

\textit{A3:} Both the watermark and embedding key are deterministically derived from a cryptographic hash over the previous shard, identity information, and a verifier-provided secret nonce $\mu$. This nonce is unknown to the attacker during training and only revealed at verification. Without knowing $\mu$, the attacker cannot precompute the correct key or watermark, rendering attacks ineffective.

\vspace{3pt}
\textit{Q4: What if the attacker just wants to remove the watermark without embedding their own?}  

\textit{A4:}
In our EMM setting, only models with a valid, registered ownership watermark (incumbent or approved adjunct) are treated as valuable assets. A watermark-free model is automatically illegitimate: it cannot be listed, licensed, or rewarded, even if its accuracy is high. Thus, if an attacker clones a legitimate model and only strips its watermark, they merely destroy the model’s monetisable status while the original watermark-protected model remains valid. Since this yields no economic gain, pure removal attacks are classified as irrational, destruction-only behaviour and are excluded from our security target; we focus instead on rational forgery, where the attacker must both erase the owner’s watermark and successfully re-embed their own.

\smallskip
\textit{Q5: Can an attacker infer training data from public checkpoints using membership inference attacks?}  

\textit{A5:} For mitigation, PoLO applies Gaussian noise to each shard’s final model by perturbing a randomised subset of non-watermarked weights, preserving both watermark integrity and model performance while hindering data reconstruction.

\smallskip
\textit{Q6: What if an attacker gains access to the embedding key $k_S$ and selection matrix $Y$ for the final (released) model?}  

\textit{A6:} Even with full knowledge of $k_S$ and $Y$, an attacker can at most forge ownership of the final checkpoint, but cannot recover or claim earlier watermarks, as each shard overwrites fixed positions. This blocks proof of training effort, which is cryptographically chained across shards. This design allows \textsc{PoLO} to treat effort proof as auxiliary ownership evidence and supports clean ownership transfer by appending a new shard with a fresh watermark.

\vspace{3pt}
\textcolor{violet}{\textbf{\ding{173} Rational (self-interested) provers.}
}
\vspace{3pt}

\textit{Q7: What if a prover tries to forge $k_x$ to finalize a shard without performing actual training, inflating the number of claimed shards for higher rewards?}  

\textit{A7:} Since $k_x$ is deterministically derived from a hash involving a verifier-provided secret nonce $\mu$ (alongside previous weights and identity), the prover cannot arbitrarily select $k_x$ without knowing or controlling $\mu$. This linkage ensures that a valid watermark can only be embedded through legitimate training progression.

\vspace{3pt}
\textit{Q8: What if the prover embeds the watermark in an “inactive” layer that has little impact on performance, allowing extraction without meaningful training?}

\textit{A8:} To prevent this, provers are not allowed to choose embedding positions or secret nonces arbitrarily. Instead, the verifier, assumed to be honest but curious, generates the nonce during registration, and it is verifiably included in the hash computation. This setup mirrors a regulated environment, like obtaining a business license before training, ensuring all work is tied to an approved watermarking path.

\vspace{3pt}
\textcolor{violet}{\ding{174} \textbf{Honest-but-curious verifier.}}
\vspace{3pt}

\textit{Q9: Can a verifier infer sensitive training data from the checkpoints provided during verification?}  

\textit{A9:} No. The final model of each shard is protected via obfuscation, which adds calibrated noise to non-watermarked weights. This ensures that even an honest-but-curious verifier is hard to execute membership inference attacks or extract sensitive training data.

\vspace{3pt}
\textcolor{violet}{\textbf{\ding{175} Watermark types.}}
\vspace{3pt}

\textit{Q10: Why \textsc{PoLO} uses embedded watermarks instead of backdoor-based watermarks?}  

\textit{A10:} \textsc{PoLO} uses embedded watermarking directly in model weights rather than backdoor triggers, ensuring that watermark detection reflects actual training effort. Backdoor-based schemes allow attackers to implant triggers into pre-trained models without meaningful training, enabling cheap forgery and breaking PoL guarantees. In contrast, embedded watermarking requires real weight updates to produce detectable signals, tying watermark validity to genuine training. This integration makes forgery computationally impractical without comparable effort, preserving security\&accountability.

\begin{table}[!t]
\centering
\caption{Effect of learning rate ($LR$) and reference threshold ($\eta_G$) on shard rate and main-task accuracy. Values are averaged.}
\label{tab:merged_lr_eta}
\setlength{\tabcolsep}{5pt}
\resizebox{\linewidth}{!}{
\renewcommand{\arraystretch}{1.05}
\begin{tabular}{c|c|ccc|ccc}
\toprule

 \multicolumn{2}{c}{} & \multicolumn{3}{c}{\textbf{ResNet18, CIFAR-10}} & \multicolumn{3}{c}{\textbf{TextCNN, AG News}} \\

\midrule

\multicolumn{1}{c}{}  & \textbf{Metric} & \textbf{0.01} & \textbf{0.1} & \textbf{0.2} & \textbf{0.0005} &  \textbf{0.005} & \textbf{0.01} \\

\cmidrule{2-8}

\multirow{4}{*}{$LR$} & epochs & 120 & 113 & 101 & 105 & 101 & 101  \\
& shards &  44 &  73 &  79 &  67 &  72 &  90  \\
& $|s|$ (\%) &  \cellcolor{magenta!11} 36.67 &\cellcolor{magenta!11}  64.61 & \cellcolor{magenta!11} 78.21 &\cellcolor{magenta!11} 63.81 &\cellcolor{magenta!11} 71.28 & \cellcolor{magenta!11} 88.51  \\
& $Acc_{\text{main}}$ (\%) &  \cellcolor{teal!21}  91.88 & \cellcolor{teal!21}   91.66 & \cellcolor{teal!21}  92.11 &\cellcolor{teal!21}  90.53 & \cellcolor{teal!21}  90.98 & \cellcolor{teal!21}  88.51 \\

\midrule

\multicolumn{1}{c}{}  & \textbf{Metric} & \textbf{0.85} & \textbf{0.95} & \textbf{0.99} & \textbf{0.85} & \textbf{0.95} & \textbf{0.99} \\

\cmidrule{2-8}

\multirow{4}{*}{$\eta_G$}  & epochs & 105 & 101 & 113 & 100 & 102 & 101 \\
& shards &  91 &  78 &  73 &  95 &  94 &  72 \\
& $|s|$ (\%) & \cellcolor{magenta!11} 86.67 & \cellcolor{magenta!11} 77.22 & \cellcolor{magenta!11} 64.61  & \cellcolor{magenta!11} 95.00 &\cellcolor{magenta!11} 92.25 & \cellcolor{magenta!11} 71.28 \\
& $Acc_{\text{main}}$ (\%) & \cellcolor{teal!21}  92.21  &  \cellcolor{teal!21} 91.89 &  \cellcolor{teal!21} 92.02  &\cellcolor{teal!21}  90.76 & \cellcolor{teal!21} 90.31  & \cellcolor{teal!21} 90.98 \\
\bottomrule
\end{tabular}
}
\end{table}

\end{document}